\documentclass[compsoc, conference, a4paper, 10pt, times]{IEEEtran}

\usepackage[utf8]{inputenc}
\usepackage{color}
\usepackage[usenames,dvipsnames]{xcolor}

\usepackage{centernot}
\usepackage{hhline}
\usepackage{mdframed}

\usepackage{amsthm}
\usepackage{bm} % for bold symbols

\usepackage[scr=boondoxo]{mathalfa} % mathscr usa il font boondoxo (calligrafico)

\usepackage{siunitx}
\usepackage{textcomp}
\usepackage{caption}
\usepackage{pst-func}
\usepackage{amssymb}
\usepackage{graphicx}
\usepackage{paralist}
\usepackage{environ}
\usepackage{pgfplots}
\usepackage{tikz}
\usetikzlibrary{shadows}
\usepackage{mathtools}
\usepackage{array,multirow}
\usepackage{algorithm}
\usepackage[hyphens]{url}
\usepackage{cite}
\usepackage{pgfplotstable}
\usepackage{threeparttable}
\usepackage{extarrows}

\usepackage{etoolbox}
\usepackage{arydshln}

\usepackage{xspace} % Smart spaces after commands
 
\usepackage{dirtytalk} % for \say command
\usepackage{stmaryrd} % for llbracket and rrbracket\textbf{\textbf{}}

\usepackage[inline,shortlabels]{enumitem} 
\newlist{inlinelist}{enumerate*}{1}
\setlist*[inlinelist,1]{%
  label=(\roman*),
}

\usepackage{xifthen}  % Extended conditional commands
\usepackage{ifthen}

\usepackage{nicefrac}

\usepackage[hidelinks]{hyperref}
\usepackage{cleveref}
\usepackage{bigints}

\usetikzlibrary{pgfplots.dateplot}
\usetikzlibrary{matrix,chains,positioning,decorations.pathreplacing}
\usetikzlibrary{patterns,calc,fit,arrows}
\usetikzlibrary{shapes.geometric}

\usepackage[final,nomargin,inline,index]{fixme} % Simplified management of FIXME's
\fxusetheme{color}

\FXRegisterAuthor{bart}{anbart}{\color{magenta} {\underline{bart}}}
\FXRegisterAuthor{ric}{anric}{\color{red} {\underline{ric}}}
\FXRegisterAuthor{zun}{anzun}{\color{OliveGreen} {\underline{zun}}}

\hypersetup{
  breaklinks   = true,
  colorlinks   = true, %Colours links instead of ugly boxes
  urlcolor     = blue, %Colour for external hyperlinks
  linkcolor    = blue, %Colour of internal links
  citecolor    = red   %Colour of citations
}
\hypersetup{final} % Needed to generate hyperlinks even in draft mode (ERROR IN MAC)

\usepackage{listings}
\definecolor{listingBG}{HTML}{FFFFCB}%
\definecolor{listingFrame}{HTML}{BBBB98}%
\definecolor{listingLineno}{rgb}{0.5,0.5,1.0}%
\usepackage{tcolorbox}

\definecolor{LightGrey}{rgb}{0.975,0.975,0.975}

\lstdefinelanguage{illum}{
	commentstyle=\color{Gray},
	morecomment=[l]{//},
	morecomment=[s]{/*}{*/},
	classoffset=0,
        escapechar=\$,
        tabsize=2,
	morekeywords={precond_wallet,after,precond_if,auth,process},
	keywordstyle=\color{Plum}\bfseries,
	classoffset=1,
	morekeywords={clause,call,send},
	keywordstyle=\color{NavyBlue}\bfseries,
	classoffset=2,
	morekeywords={int,string,bool,address,uint,mapping},
	keywordstyle=\color{MidnightBlue}\bfseries,
	basicstyle=\fontseries{m}\normalsize\ttfamily
	\lst@ifdisplaystyle\scriptsize\fi,
}

\lstdefinelanguage{hellum}{
	commentstyle=\color{Gray},
	morecomment=[l]{//},
	morecomment=[s]{/*}{*/},
	classoffset=0,
        escapechar=\$,
        tabsize=2,
        literate={\ \ }{{\ }}1,
	morekeywords={contract,constructor,function,if,else,require},
	keywordstyle=\color{red},
	classoffset=1,
	morekeywords={balance,after,input,next,auth,transfer,view},
	keywordstyle=\color{red},
	classoffset=2,
	morekeywords={int,string,bool,address,uint,mapping},
	keywordstyle=\color{blue},
 	classoffset=3,
  	keywordstyle=\color{Plum}\bfseries,
	basicstyle=\fontseries{m}\normalsize\ttfamily
        \lst@ifdisplaystyle\scriptsize\fi,
}

\newcommand{\ifempty}[3]{%
  \ifthenelse{\isempty{#1}}{#2}{#3}%
}

\newcommand{\ifdots}[3]{%
  \ifthenelse{\equal{#1}{...}}{#2}{#3}%
}

\newcommand{\hidden}[1]{}

\newcommand{\change}[2][]{#2}

\newcommand{\keyterm}[1]{\textbf{\emph{#1}}}%

\newcommand{\illum}{\textsc{Illum}\xspace}
\newcommand{\hellum}{\textsc{HeLLUM}\xspace}
\newcommand{\langname}{\illum}

\newcommand{\illinline}[1]{\mbox{\lstinline[language=illum]{#1}}}
\newcommand{\hllinline}[1]{\mbox{\lstinline[language=hellum]{#1}}}
\newcommand{\hllfunction}[1]{\mbox{\lstinline[language=hellum,keywordstyle=\color{Plum}\bfseries]{#1}}}

\renewcommand{\vec}[1]{\boldsymbol{#1}}

\newcommand{\Real}[1]{\mathrm{Real}}

\newcommand{\codefont}{\fontsize{10}{10}\selectfont}
\newcommand{\code}[1]{{\tt\codefont {#1}}}

\newcommand{\hlltxColor}{\color{Plum}}
\newcommand{\contract}[1]{{\tt\codefont{\hlltxColor{#1}}}}
\newcommand{\txcode}[1]{{\text{\tt\codefont{\hlltxColor{#1}}}}}

\newcommand{\eg}{e.g.\@\xspace}
\newcommand{\ie}{i.e.\@\xspace}

\def\negcaptionspace{\vspace{-10pt}}

\theoremstyle{plain}% default
\newtheorem{thm}{Theorem}
\newtheorem{lem}[thm]{Lemma}
\newtheorem{prop}[thm]{Proposition}

\theoremstyle{definition}
\newtheorem{defn}{Definition}

  {}

\newcommand{\sig}[3][]{\mathit{sig}^{#1}_{#2}\ifempty{#3}{}{({#3})}}

\newcommand{\key}[2]{\ensuremath{K_{#1}\ifempty{#2}{}{({#2})}}}
\newcommand{\pubkey}[2]{\ensuremath{K^{p}_{#1}\ifempty{#2}{}{({#2})}}}
\newcommand{\privkey}[2]{\ensuremath{K^{s}_{#1}\ifempty{#2}{}{({#2})}}}

\newcommand{\BTC}{\textup{%
  \leavevmode
  \vtop{\offinterlineskip %\bfseries
    \setbox0=\hbox{B}%
    \setbox2=\hbox to\wd0{\hfil\hskip-.03em
    \vrule height .3ex width .15ex\hskip .08em
    \vrule height .3ex width .15ex\hfil}
    \vbox{\copy2\box0}\box2}}\xspace}

\def\pmvColor{\color{ForestGreen}}
\newcommand{\pmvFmt}[1]{{\pmvColor{\sf #1}}}

\newcommand{\Part}{\pmvFmt{Part}\xspace} % universe of ALL participants
\newcommand{\PartT}{\pmvFmt{{Hon}}\xspace} % set of honest participants
\newcommand{\pmv}[2][]{\pmvFmt{#2}_{\pmvColor{#1}}\xspace}
\newcommand{\pmvA}[1][]{\pmv[{#1}]{A}}
\newcommand{\pmvB}[1][]{\pmv[{#1}]{B}}
\newcommand{\pmvC}[1][]{\pmv[{#1}]{C}}
\newcommand{\Adv}{\pmv{Adv}} % adversary

\def\txColor{\color{MidnightBlue}}
\def\fieldColor{\color{Plum}}
\newcommand{\txFmt}[1]{{\txColor{\sf #1}}}

\newcommand{\tx}[2][]{\txFmt{#2}_{\txColor{#1}}}
\newcommand{\txT}[1][]{\tx[#1]{T}}

\newcommand{\txTi}[1][]{\txFmt{T'_{\txColor{{\it #1}}}}}
\newcommand{\txTii}[1][]{\txFmt{T''_{\txColor{{\it #1}}}}}

\newcommand{\txTag}[3][]{{\fieldColor\sf #3}\ifempty{#1}{\ifempty{#2}{}{: {#2}}}{({#1})\ifempty{#2}{}{: {#2}}}}

\newcommand{\txIn}[2][]{\txTag[{#1}]{#2}{in}}
\newcommand{\txWit}[2][]{\txTag[{#1}]{#2}{wit}}
\newcommand{\txOut}[2][]{\txTag[{#1}]{#2}{out}}
\newcommand{\txAfterAbs}[2][]{\txTag[{#1}]{#2}{absLock}}
\newcommand{\txAfterRel}[2][]{\txTag[{#1}]{#2}{relLock}}
\newcommand{\txf}{\txTag{}{f}} % transaction field
\newcommand{\txarg}[1][]{\txTag{}{arg_{#1}}}
\newcommand{\txscript}{\txTag{}{scr}}
\newcommand{\txval}{\txTag{}{val}}

\DeclareMathAlphabet{\mathbfsf}{\encodingdefault}{\sfdefault}{bx}{n}

\newcommand{\bcB}[1][]{{\mathbfsf{\txColor{B}}}_{\txColor{#1}}}

\newcommand{\eqdef}{\triangleq}

\newcommand{\mmid}{\,\|\,}

\newcommand{\irule}[2]{\dfrac{#1}{#2}}

\newcommand{\bnfdef}{::=}
\newcommand{\bnfmid}{\;|\;}

\newcommand{\sem}[2][]{\mbox{\ensuremath{\llbracket{#2}\rrbracket_{#1}}}}

\newcommand{\ran}[1]{\operatorname{ran} {#1}}

\newcommand{\Nat}{\mathbb{N}}

\newcommand{\setenum}[1]{\{#1\}}

\newcommand{\seqat}[2]{{#1}.{#2}}

\definecolor{LightGrey}{rgb}{0.95,0.95,0.95}
\definecolor{keyword}{HTML}{7F0055}

\def\tokColor{\color{Orange}}
\newcommand{\tokFmt}[1]{{ \mathrm{\tokColor{#1}} }}
\newcommand{\tok}[2][]{\tokFmt{#2}_{\tokColor{#1}}\xspace}
\newcommand{\tokT}[1][]{\tok[{#1}]{T}}

\newlength\replength
\newcommand\repfrac{.1}

\newcommand\rulewidth{.6pt}
\newcommand\tdashfill[1][\repfrac]{\cleaders\hbox to \replength{%
  \smash{\rule[\arraystretch\ht\strutbox]{\repfrac\replength}{\rulewidth}}}\hfill}

\newcommand\tdotfill[1][\repfrac]{\cleaders\hbox to \replength{%
  \smash{\raisebox{\arraystretch\dimexpr\ht\strutbox-.1ex\relax}{.}}}\hfill}

\newcommand{\var}[2][]{#2_{#1}} % variables

\newcommand{\varY}[1][]{\var[#1]{y}}

\newcommand{\const}[2][]{#2_{#1}} % constants

\newcommand{\constPK}[1][]{\const[{#1}]{pk}}

\newcommand{\val}[2][]{#2_{#1}} % values
\newcommand{\valV}[1][]{\val[#1]{v}}
\newcommand{\valVi}[1][]{\val[#1]{v'}}

\newcommand{\versigName}{{\sf versig}}
\newcommand{\versig}[2]{\versigName({#1},{#2})}
\newcommand{\rtx}{{\sf rtx}}

\newcommand{\afterAbs}[2]{{\sf absAfter}~{#1}:{#2}}
\newcommand{\afterRel}[2]{{\sf relAfter}~{#1}:{#2}}

\newcommand{\notE}[1]{{\sf not}~{#1}}

\newcommand{\ifE}[3]{\mathsf{if}~{#1}~\mathsf{then}~{#2}~\mathsf{else}~{#3}}

\newcommand{\andE}{~{\sf and}~}
\newcommand{\orE}{~{\sf or}~}

\newcommand{\txo}{{\fieldColor{\sf o}}} % metavariable for rtxo, stxo, ptxo

\newcommand{\txof}[2]{{#1}.{#2}}
\newcommand{\ctxo}[1]{{\sf ctxo}\ifempty{#1}{}{.{#1}}}

\newcommand{\rtxo}[2]{{\sf rtxo}({#2})\ifempty{#1}{}{.{#1}}}

\newcommand{\inidx}{{\sf inidx}}

\newcommand{\inlen}[1]{{\sf inlen({#1})}}
\newcommand{\outlen}[1]{{\sf outlen({#1})}}

\newcommand{\verscript}[2]{{\sf verscr}\ifempty{#1}{}{({#1},{#2})}}
\newcommand{\verrec}[1]{{\sf verrec}\ifempty{#1}{}{({#1})}}

\newcommand{\utxo}[1]{\mathit{UTXO}({#1})}

\newcommand{\true}{\mathit{true}}
\newcommand{\false}{\mathit{false}}

\def\contrColor{\color{RubineRed}}
\newcommand{\contrFmt}[1]{{\contrColor{\it #1}}}
\newcommand{\contrC}[1][]{\mathord{\contrFmt{C}_{\contrColor{#1}}}}
\newcommand{\contrCi}[1][]{\mathord{\contrC[#1]\contrColor{'}}}

\newcommand{\contrD}[1][]{\mathord{\contrFmt{D}_{\contrColor{#1}}}}
\newcommand{\contrDi}[1][]{\mathord{\contrD[#1]\contrColor{'}}}
\newcommand{\contrDii}[1][]{\mathord{\contrD[#1]\contrColor{''}}}

\newcommand{\contrAdv}[2]{\setenum{#1}{#2}}  % symbolic contract advertisement
\newcommand{\contrAdvC}[2]{\mathcal{C}} % computational contract advertisement

\newcommand{\wherename}{\textup{\texttt{if}}}
\newcommand{\elsename}{\textup{\texttt{else}}}

\newcommand{\splitB}[2]{{#1} \rightarrow {#2}}

\newcommand{\withdrawname}{\textup{\texttt{send}}}
\newcommand{\withdrawC}[1]{\withdrawname\ifempty{#1}{}{\; {#1}}}

\newcommand{\authC}[2]{{\pmv{#1}}\,\textup{\texttt{:}}\,{#2}}

\newcommand{\afterName}{\texttt{after}}
\newcommand{\afterRelName}{\texttt{afterRel}}
\newcommand{\afterC}[2]{\textup{\afterName}\,{#1}\,\textup{\texttt{:}}\,{#2}}
\newcommand{\afterRelC}[2]{\textup{\afterRelName}\,{#1}\,\textup{\texttt{:}}\,{#2}}

\newcommand{\cVar}[2][]{\texttt{\textup{#2}}_{#1}}
\newcommand{\cVarX}[1][]{\cVar[#1]{X}}
\newcommand{\cVarY}[1][]{\cVar[#1]{Y}}

\newcommand{\sexp}[1][]{\mathcal{E}_{#1}}

\newcommand{\ncadv}[1]{\textup{\texttt{call}}\ifempty{#1}{}{\;}{#1}}

\newcommand{\confG}[1][]{\Gamma_{#1}}
\newcommand{\confGi}[1][]{\Gamma'_{#1}}
\newcommand{\confD}[1][]{\Delta_{#1}}

\newcommand{\confContr}[3][]{\langle {#2}, {#3} \rangle_{#1}}
\newcommand{\confDep}[3][]{\langle {#2}, {#3} \rangle_{#1}}
\newcommand{\confAuth}[2]{{#1} [{#2}]}

\newcommand{\authBranch}[2]{{#1} \rhd {#2}}

\def\compileColor{\color{Plum}}
\newcommand{\compileFmt}[1]{{\compileColor{\mathbf{#1}}}}

\newcommand{\txMap}{\compileFmt{txout}}

\newcommand{\compilename}{\BTC}

\newcommand{\compile}[1]{\compileFmt{\compilename_{adv}}\ifempty{#1}{}{({#1})}}

\newcommand{\labS}[1][]{\alpha_{#1}}        % symbolic label (of the timed LTS between configuration)
\newcommand{\labSi}[1][]{\alpha'_{#1}}      % symbolic label
\newcommand{\labAdv}{\lambda^{\! s}}        % output of the symbolic adversary strategy
\newcommand{\LabS}{\Lambda^{\! s}}          % set of symbolic labels of the timed LTS between configuration
\newcommand{\nonce}[1]{r_{#1}}   % nonce
\newcommand{\rndR}{r}            % randomness source

\newcommand{\stratS}[1]{\Sigma_{#1}^{\it s}}
\newcommand{\stratC}[1]{\Sigma_{#1}^{\it c}}

\newcommand{\stratSSet}{\mathbf{\Sigma}^{\it s}} % set of symbolic strategies
\newcommand{\stratCSet}{\mathbf{\Sigma}^{\it c}} % set of computational strategies

\newcommand{\stratMap}{\aleph}

\newcommand{\coher}[5]{\mathit{coher}({#1},{#2},{#3},{#4},{#5})}
\newcommand{\coherRel}[3]{{#1} \sim_{#3} {#2}}
\newcommand{\txMapC}{\mathit{txout}}

\newcommand{\txMapCi}{\mathit{txout}'}

\newenvironment{nscenter}
 {\parskip=0pt\par\nopagebreak\centering}
 {\parskip=2pt\par\noindent} % \ignorespacesafterend

\newcommand{\commaCode}{\textup{\texttt{,}}}

\newcommand{\nexp}[1][]{{\pmvColor\mathcal{N}_{{#1}}}}

\newcommand{\stPar}[1][]{\mathit{In}_{#1}  }
\newcommand{\dinPar}[1][]{\mathit{Ex}_{#1} }
\newcommand{\stVal}[1][]{\mathrm{\overline{In}}_{#1}  } 
\newcommand{\dinVal}[1][]{\mathrm{\overline{Ex}}_{#1} }

\newcommand{\clauseDef}[3]{{#1}(\vec{{#2}} \textup{\texttt{;}}  \vec{{#3}} )}
\newcommand{\clauseDefX}[3][]{\clauseDef{\cVarX[{#1}]}{{#2}}{{#3}}}
\newcommand{\clauseDefY}[3][]{\clauseDef{\cVarY[{#1}]}{{#2}}{{#3}}}
\newcommand{\clauseDefNonVector}[3]{{#1}({#2} \textup{\texttt{;}}  {#3})}
\newcommand{\clauseDefParam}[2][]{\clauseDef{{#2}}{\stPar[{#1}]}{\dinPar[{#1}]} }
\newcommand{\clauseDefXParam}[1][]{\clauseDefParam[{#1}]{\cVarX[{#1}]} }

\newcommand{\clauseCall}[3]{{#1}\langle \vec{{#2}} \textup{\texttt{;}} \vec{{#3}} \rangle}
\newcommand{\clauseCallX}[3][]{\clauseCall{\cVarX[{#1}]}{{#2}}{{#3}}}
\newcommand{\clauseCallY}[3][]{\clauseCall{\cVarY[{#1}]}{{#2}}{{#3}}}
\newcommand{\clauseCallNonVector}[3]{{#1} \langle{#2} \textup{\texttt{;}}  {#3} \rangle}
\newcommand{\clauseCallParam}[2][]{\clauseCall{{#2}}{\stVal[{#1}]}{\dinVal[{#1}]} }
\newcommand{\clauseCallXParam}[1][]{\clauseCallParam[{#1}]{\cVarX[{#1}]} }

\newcommand{\clauseFunding}[2]{\setenum{ {#1} \ifempty{#2}{}{{\texttt{ if } {#2}} } }}
\newcommand{\clauseAdv}[3]{\ifempty{#2}{\clauseFunding{#1}{} \ {#3}}{\clauseFunding{#1}{#2} \ {#3}}}

\newcommand{\confContrTime}[4][]{\langle {#3}, {#4} \rangle_{#1}^{#2}}
\newcommand{\confTaint}[1]{ \mathcal{D}\ifempty{#1}{}{({#1})}}

\newcommand{\confF}[1][]{\Phi_{#1}}

\newcommand{\confAdv}[5]{\left[ {#1}\, \textup{\texttt{;}} \,{#2}  \ifempty{#3}{}{ \, \textup{\texttt{;}} \, {#3} } \, \textup{\texttt{;}} \, ({#4} \textup{\texttt{,}} {#5} ) \right]}
\newcommand{\confAdvInit}[3]{\left[ {#1} \, \textup{\texttt{;}} \, {#2} \ifempty{#3}{}{ \, \textup{\texttt{;}} \, {#3} } \right]}
\newcommand{\confAdvDes}[2]{\left[ {#1}  \ifempty{#2}{}{ \, \textup{\texttt{;}} \, {#2} } \right]}

\newcommand{\confAdvN}[6]{\left[ {#1}\, \textup{\texttt{;}} \,{#2}  \ifempty{#3}{}{ \, \textup{\texttt{;}} \, {#3} } \, \textup{\texttt{;}} \, ({#4} \textup{\texttt{,}} {#5} ) \right]_{{#6}} }
\newcommand{\confAdvInitN}[4]{\left[ {#1} \, \textup{\texttt{;}} \, {#2} \ifempty{#3}{}{ \, \textup{\texttt{;}} \, {#3} } \right]_{{#4}}}
\newcommand{\confAdvDesN}[3]{\left[ {#1}  \ifempty{#2}{}{ \, \textup{\texttt{;}} \, {#2} } \right]_{{#3}}}

\def\contrComplColor{\color{RoyalPurple}}
\newcommand{\complFmt}[1]{{\contrComplColor{\it #1}}}

\newcommand{\complD}{\complFmt{\bar{D}}}
\newcommand{\complDi}{\complFmt{\bar{D'}}}

\newcommand{\callname}{\textup{\texttt{call}}}

\newcommand{\authDonate}[2]{{#1} \rhd {#2}}
\newcommand{\decC}[2]{{#1}\,\textup{\texttt{:}}\,{#2}}
\newcommand{\authSet}[1]{\mathcal{A}_{#1}}  % authorization actions 

\newcommand{\runnameS}{\mathit{R}}
\newcommand{\runnameC}{\mathit{R}}

\newcommand{\runS}[1][]{\runnameS^{s}_{#1}}
\newcommand{\runSi}[1][]{\dot{\runnameS}^{s}_{#1}}

\newcommand{\runC}[1][]{\runnameC^{c}_{#1}}
\newcommand{\runCi}[1][]{\dot{\runnameC}^{c}_{#1}}

\newcommand{\ctx}{{\sf ctx}}
\newcommand{\rtxw}{{\sf rtxw}}
\newcommand{\scr}{{\sf scr}}

\newcommand{\arglen}[1]{{\sf arglen({#1})}}

\newcommand{\labC}{\lambda^{\! c}}          
\newcommand{\LabC}{\Lambda^{\! c}}

\newcommand{\txbranch}{\txTag{}{branch}}
\newcommand{\txname}{\txTag{}{name}}
\newcommand{\txnonce}{\txTag{}{nonce}}
\newcommand{\txowner}{\txTag{}{owner}}
\newcommand{\txalpha}[1][]{\txTag{}{\alpha_{{#1}}}}
\newcommand{\txbeta}[1][]{\txTag{}{\beta_{{#1}}}}

\newcommand{\keyMap}{\compileFmt{key}}

\newcommand{\keyMapC}{\mathit{key}}
\newcommand{\advMapC}{\mathit{prevTx}}
\newcommand{\keyMapCi}{\mathit{key'}}
\newcommand{\advMapCi}{\mathit{prevTx'}}

\def\paramCmpColor{\color{teal}}
\newcommand{\paramCmpFmt}[1]{{\paramCmpColor{\mathbf{#1}}}}
\newcommand{\inCmp}[1][]{\paramCmpFmt{in_{{#1}}}}

\newcommand{\nonceCmp}[1][]{\paramCmpFmt{nonce_{{#1}}}}

\newcommand{\tRelCmp}[1][]{\paramCmpFmt{t_{{#1}}}}

\newcommand{\tAbsCmp}[1][]{\paramCmpColor{\mathrm{t_0}}}
\newcommand{\tAbsCmpi}[1][]{\paramCmpFmt{\mathrm{t'_0}}}

\newcommand{\thenname}{\textup{\texttt{then}}}

\newtoggle{anonymous}
\togglefalse{anonymous}

\newtoggle{arxiv}
\togglefalse{arxiv}

\usepackage[hyphens]{url}
\usepackage{hyperref}
\usepackage[hyphenbreaks]{breakurl}

\makeatletter
\renewcommand\paragraph{\@startsection{paragraph}{4}{\z@}%
  {2.25ex \@plus 1ex \@minus .2ex}%
  {-0.75em}%
  {\normalfont\normalsize\bfseries}}
\makeatother

\newcommand{\github}{\iftoggle{anonymous}{\emph{[Reference to github
repository omitted for double-bind review]\xspace}}{https://github.com/bitbart/illum-lang/}}

\begin{document}

% \title{Secure compilation of a high-level smart contract language in a bare-bone UTXO model}

\title{Secure compilation of rich smart contracts on poor UTXO blockchains}

\iftoggle{anonymous}{
\author{
  \IEEEauthorblockN{Given Name Surname}
  \IEEEauthorblockA{\textit{dept. name of organization (of Aff.)} \\
    City, Country}
}
}{
\author{
  \IEEEauthorblockN{Massimo Bartoletti}
  \IEEEauthorblockA{\textit{Universit\`a degli Studi di Cagliari}}
  Cagliari, Italy \\
  \textit{bart@unica.it}
  \and
  \IEEEauthorblockN{Riccardo Marchesin}
  \IEEEauthorblockA{\textit{Universit\`a degli Studi di Trento}}
  Trento, Italy \\
  \textit{riccardo.marchesin@unitn.it}
  \and
  \IEEEauthorblockN{Roberto Zunino}
  \IEEEauthorblockA{\textit{Universit\`a degli Studi di Trento}}
  Trento, Italy \\
  \textit{roberto.zunino@unitn.it}
}
}

\maketitle

\pagestyle{plain} % REMOVE BEFORE SUBMISSION

\begin{abstract}
  Most blockchain platforms from Ethereum onwards
  render smart contracts as stateful reactive objects
  that update their state and transfer crypto-assets
  in response to transactions.
  %In this way, they support the development of contracts in the imperative procedural paradigm, familiar to most programmers.
  A drawback of this design is that when users submit a transaction,
  they cannot predict in which state it will be executed.
  This exposes them to transaction-ordering attacks,
  a widespread class of attacks 
  where adversaries with the power to construct blocks of transactions
  can extract value from smart contracts
  (the so-called MEV attacks).
  The UTXO model is an alternative blockchain design that
  thwarts these attacks by requiring new transactions to spend past ones:
  since transactions have unique identifiers, 
  reordering attacks are ineffective.
  Currently, the blockchains following the UTXO model
  either provide contracts with limited expressiveness (Bitcoin),
  or require complex run-time environments (Cardano).
  We present \langname, an Intermediate-Level Language for the UTXO Model.
  \langname can express real-world smart contracts, \eg those found in Decentralized Finance.
  We define a compiler from \langname to a bare-bone UTXO blockchain with loop-free scripts.
  Our compilation target only requires minimal extensions to Bitcoin Script:
  in particular, we exploit covenants,
  a mechanism for preserving scripts along chains of transactions.
  We prove the security of our compiler: namely, any attack targeting the compiled contract is also observable at the \langname level. Hence, the compiler does not introduce new vulnerabilities that were not already present in the source \langname contract.
  \change{We evaluate the practicality of \langname as a compilation target for higher-level languages.
  To this purpose, we implement a compiler from a contract language inspired by Solidity to \langname, and we apply it to a benchmark or real-world smart contracts.}
\end{abstract}

%  We present a framework for smart contracts in the UTXO model, that allows expressive contracts to be securely executed by bare-bone UTXO blockchains with loop-free scripts enriched with covenants.


\section{Introduction}

Smart contracts are agreements between mutually untrusted parties
that are enforceable by a computer program, without the need of
a trusted intermediary.
Currently, most implementations of smart contracts are based on
permissionless blockchains, where the conjunction with crypto-assets
has given rise to new applications,
like decentralized finance (DeFi)~\cite{werner2021sok}
and decentralized autonomous organizations (DAOs)~\cite{Wang19tcss},
that overall control nearly 90 billion dollars worth of assets today~\cite{defillama}.

Two main smart contracts models have emerged so far.
In the \emph{account-based model},
contracts are reactive objects that live on the blockchain
and process user transactions
by updating their state and transferring crypto-assets among users~\cite{Sergey17wtsc}.
% according to their logic
In the \emph{UTXO model}, instead,
contracts, their state, and the ownership of assets
are encoded within transactions:
when a new transaction is published in the blockchain,
it replaces (``spends'') an old transaction,
effectively updating the contract state and the assets ownership.
The UTXO model was first proposed by Bitcoin,
where the idea of blockchain-based contracts originated in 2012.
The account-based model was later introduced in 2015 by Ethereum,
where contracts were popularized.
Most blockchain platforms today follow the account-based model:
besides Ethereum, also other mainstream blockchains such as
Solana, Avalanche, Hedera, Algorand and Tezos are account-based
(albeit with differences, sometimes notable, from case to case).

\paragraph{Account-based vs.\ UTXO blockchains}
%
% For most developers, it is natural to think of contracts
In the account-based model, contracts can be seen 
as objects with a state accessible and modifiable by methods,
as in object-oriented programming.
For instance, to withdraw 10 token units from a \contract{Bank}
contract, a user $\pmvA$ sends a transaction \txcode{withdraw(10)}
to \contract{Bank}, which will react by updating its state
and $\pmvA$'s wallet.
% This resemblance with the object-oriented programming
% explains the dominance of account-based blockchains.
Programming contracts in the UTXO model requires instead a paradigm shift
from the common object-oriented style~\cite{Brunjes20isola}.
Indeed, a UTXO transaction does not directly represent a contract action:
rather, it encodes a transfer of crypto-assets from its \emph{inputs}
to its \emph{outputs}.
Transaction outputs specify the assets they control,
the contract state, 
and the conditions under which the assets can be transferred again.
Transaction inputs are references to unspent outputs of previous transactions,
and provide the values that make their spending conditions true.
The blockchain state is given by the set of Unspent Transaction Outputs (UTXO).
A transaction can spend one or more of outputs in the UTXO set,
specifying them as its inputs:
this effectively removes these outputs from the blockchain state,
and creates new ones. 
The new outputs update the state of the contracts,
% of the spent ones
and redistribute the assets according to their spending conditions.
These conditions are specified in a \emph{scripting language},
the expressiveness of which is reflected on that of contracts.
For instance, in the banking use case above,
the state of the \contract{Bank} contract
could be scattered among a set of outputs.
To withdraw, $\pmvA$ must send a transaction which spends one or more
of these outputs, and whose output has a spending condition that can be
satisfied only by $\pmvA$ (\eg, a signature verification against
$\pmvA$'s public key).
In addition, the \contract{Bank} state in the new output
must be a correct update of the old state
(\eg, in the new state $\pmvA$'s account must have 10 tokens less
than in the old state).
Programming contracts in this model is more complex than in the
account-based model, since the links to familiar programming abstractions
are weaker.

Despite this additional complexity,
the UTXO model has a series of advantages over the account-based model.
A first problem of account-based stateful platforms like Ethereum
is to undermine the concurrent execution of transactions.
Namely, a multi-core blockchain node cannot simply execute
transactions in parallel,
since they may perform conflicting accesses to shared parts of the state,
possibly leading to an inconsistent state~\cite{Dickerson17podc}.
In such platforms, there is no efficient way to detect when
transactions can be safely parallelized: in general,
determining the accessed parts of the state requires to fully execute them.
In the UTXO model, instead, it is easy to detect
when transactions are parallelizable:
just check if they spend disjoint outputs,
which can be done efficiently~\cite{BGM21lmcs}.

Another problem of the account-based model is that a user sending a transaction
to the blockchain network cannot accurately predict the state
in which it will be executed.
This has several negative consequences, such as the unpredictability
of transaction fees and the susceptibility to
maximal extractable value (MEV)
attacks~\cite{Eskandari19sok,Daian20flash,Qin21quantifying}.
Fees are a common incentive mechanism
for the blockchain network to execute transactions and
a defence against denial-of-service attacks.
To be accepted, a transaction must pay a fee
which is proportional to the computational resources needed to validate it.
The actual amount of these resources heavily depends on the initial state
where the transaction is performed:
so, to be sure that their transactions are accepted,
users specify a maximum fee they are willing to pay. 
Besides forcing users to over-approximate fees,
this also opens to attacks where the adversary front-runs transactions
so that they are executed in a state where the paid fee is insufficient: 
the consequence is that users pay the fee even for rejected transactions
that do not update the contract state according to their intention.
With MEV attacks instead, the adversary colludes with malicious blockchain nodes
to propose blocks where the ordering of transactions is 
profitable for the adversary (to the detriment of users).
These attacks are very common in account-based blockchains
(targeting in particular DeFi contracts),
and are estimated to be worth more than USD 1 billion~\cite{flashbotsPostmerge} so far.
% flashbotsPremerge

The UTXO model naturally mitigates these attacks.
Indeed, when a user sends a transaction $\txT$ to the blockchain network,
they know \emph{exactly} in which state it will be executed, 
since this state is completely determined by $\txT$'s inputs.
Therefore, if an adversary $\pmv{M}$ front-runs $\txT$ with their transaction
$\txT[\pmv{M}]$,
the transaction $\txT$ will be rejected by the blockchain network,
since some of its inputs are spent by $\txT[\pmv{M}]$.
If the user still desires to perform the action in the new state,
they must resend $\txT$, updating its inputs
(and therefore, specifying the new state where the action is executed).
This thwarts both the fees exhaustion attacks and the MEV attacks
described before.

\paragraph{UTXO designs: Bitcoin vs. Cardano}

Currently, the two main UTXO blockchains are Bitcoin and Cardano.
These platforms follow radically different design choices
in the structure of transactions and
in the scripting languages to specify their spending conditions.
% Bitcoin and Cardano transactions are substantially different:
These differences deeply affect the expressiveness of their contracts
and the complexity of their runtime environments.
On the one hand, Bitcoin has a minimal scripting language,
featuring only basic arithmetic and logical operations,
conditionals, hashes, and (limited) signature verification~\cite{bitcointxm}.
This imposes a stringent limit on the expressiveness of contracts in Bitcoin:
contracts requiring unbounded computational steps, or
transfers of tokens different than native crypto-currency, 
cannot be expressed~\cite{bitmlracket}.
Neglecting the lack of expressiveness, the design choice
of keeping the scripting language  minimal
has some positive aspects: besides limiting the attack surface and
simplifying the overall design (\eg, no gas mechanism is needed),
it facilitates the formal verification of contracts.
On the other side of the spectrum, Cardano's scripting language
is an untyped lambda-calculus~\cite{plutus-core},
which makes Cardano scripts, and in turn contracts, Turing-complete.
This increase in expressiveness comes at a cost,
in that the static verification of general contract properties
is undecidable.
Furthermore, since Cardano's scripts feature unbounded iteration,
a gas mechanism is needed to abort time-consuming computations
and suitably reward blockchain nodes for validating transactions.
Although the gas needed to execute a transaction is statically known
(unlike in account-based blockchains, where it depends on the
actual state where the transaction is executed),
still it would be safer to avoid the gas mechanism altogether.
For instance, a misalignment of the gas incentives led to
DoS attack on Ethereum~\cite{Perez20ndss}.
Another drawback of Bitcoin, Cardano, and UTXO blockchains in general,
is that, due to the absence of explicit state,
writing stateful contracts 
is more difficult than in account-based blockchains~\cite{Brunjes20isola}. 

Our research question is whether one can find a balance
between the two approaches which also overcomes their usability issues.
More specifically, ours is a quest for a contract language and
a UXTO model such that:
\begin{itemize}

\item the expressiveness of contracts is enough for real-world use cases
  (contracts are Turing-complete);

\item executing contracts requires a simple blockchain design
  (individual transaction scripts are not Turing-complete, and
  no gas mechanism is required);

\item it can serve as a compilation target of developer-friendly higher-level contract languages.

\end{itemize}

\paragraph*{Contributions}

We address this research question 
by proposing an expressive 
intermediate-level contract language
that compiles into transactions executable by
a bare-bone UTXO blockchain
(with no gas mechanism).
The key insight is to scatter the execution of complex contract actions across
multiple UTXO transactions.
Even if each of these transactions contains only
simple (loop-free) scripts, the overall chain of transactions can encompass
complex (possibly recursive) behaviours.

We summarize our main contributions as follows:
\begin{itemize}

\item \langname, an Intermediate Level Language for UTXO blockchains.  
  \langname is a Turing-complete clause language with
  primitives to exchange crypto-assets.
  We evaluate \langname on a few use cases,
  including gambling games, auctions and Ponzi schemes
  (\Cref{sec:overview}).
  % to complex DeFi protocols like Lending Pools (\Cref{sec:discussion}).
  
\item a compiler from \langname to UTXO transactions.
  The scripting language used in these transactions is 
  Bitcoin Script extended with \emph{covenants},
  operators to constrain the output scripts of the redeeming
  transactions~\cite{BLZ20isola}.
  This is a lightweight mechanism, which can be implemented
  with minimal overhead on the runtime of
  UTXO blockchains~\cite{Moser16bw,Oconnor17bw}.
  
\item a proof of the security of the \langname compiler.
  % through a \emph{computational soundness} theorem.
  Namely, we prove that,
  even in the presence of adversaries, with overwhelming probability
  there is a step-by-step correspondence between
  the execution of an \langname contract and
  that of the chain of transactions resulting from its compilation.
  \change{Our security result essentially establishes Robust Trace Property Preservation~\cite{Patrignani19csur,Abate19csfw}, ensuring that each computational trace (involving any computational adversary) has a symbolic counterpart (involving a suitable symbolic adversary).}
  Its proof is quite complex, as it matches every possible
  contract action in \langname (more than 20 cases) with some action
  at the blockchain level.

\item \change{a prototype implementation of a compiler from a Solidity-like high-level contract language to \langname. We illustrate our compilation technique in~\Cref{sec:hellum}.}

\item \change{an evaluation of the practicality of our approach, based on a benchmark of common smart contracts that we implement in the high-level language and then translate to \langname with our prototype compiler. Overall, our evaluation shows that it is feasible to reconcile the UTXO model with the familiar procedural programming style supported by Solidity, effectively making UTXO contracts more usable in practice.}

\end{itemize}

\change{Because of space constraints, we refer to a technical report for the proofs~\cite{BMZ23arxiv}, and to a github repository\footnote{\url{\github}} for the code of the prototype compiler and of the benchmark.}

\section{Overview}
\label{sec:overview}

In this~\namecref{sec:overview} we overview our approach, 
discussing its main features and results.
% focussing on the core of our contribution, \ie the security of the compilation
% from \langname into UTXO transactions.
Here we will mostly focus on intuition, leveraging on examples
and postponing the full technical development to later sections.

\subsection{An intermediate contract language}

\langname is a clause-based process calculus
% compilation target of high-level contract languages. 
that can serve as an intermediate contract language
% between high-level languages and
and compiles to a bare-bone UTXO model.
\langname contracts are sets of \keyterm{clauses},
each having a defining equation of the form:
\[
	\clauseDefXParam = \clauseAdv{\vec{v \tokT}}{p}{ \contrC} 
\]
Here, $\cVarX$ is the clause name, $\vec{\stPar}$ and $\vec{\dinPar}$ are sequences of formal parameters (respectively, \emph{internal} and \emph{external}),
$\clauseFunding{\vec{v\tokT}}{p}$ is the \emph{funding precondition}
(namely: ``$\vec{v}$ units of tokens $\vec{\tokT}$ are available
and the condition $p$ is true''),
and $\contrC$ is a \emph{process} encoding the clause behaviour.
We provide some intuition about these constructs with an example.

\paragraph*{Example: a ``double or nothing'' game}

We specify a gambling game between two players $\pmvA$ and $\pmvB$
as follows:
\begin{align*}
  \clauseDefNonVector{\cVar{Init}}{}{}
  & = \clauseFunding{0}{} \ ( \ncadv{\clauseCallNonVector{\cVar[\pmvA]{Play}}{1}{}} +  \ncadv{\clauseCallNonVector{\cVar[\pmvB]{Play}}{1}{}} )
  \\
  \clauseDefNonVector{\cVar[\pmvA]{Play}}{v}{} &= \clauseFunding{v\tokT}{} \ (
  \ncadv{\clauseCallNonVector{\cVar[\pmvB]{Play}}{2v}{}}
  \\
  & \qquad \quad\, + \afterRelC{5}{\withdrawC{ (\splitB{v\tokT}{\pmvA} )}} )  
  \\
  \clauseDefNonVector{\cVar[\pmvB]{Play}}{v}{} &= \clauseFunding{v\tokT}{} \ (
  \ncadv{\clauseCallNonVector{\cVar[\pmvA]{Play}}{2v}{}}
  \\
  &\qquad \quad\, + \afterRelC{5}{\withdrawC{ (\splitB{v\tokT}{\pmvB} )}} ) 
\end{align*}
Here, we have not used external parameters, and we have
omitted writing $\code{if}\ true$ in the funding precondition.

To start the game, participants can invoke 
the $\cVar{Init}$ clause.
Since it has no funding precondition, this does not require paying
tokens upfront.
After that, the contract gives two options:
either calling $\cVar[\pmvA]{Play}$ or $\cVar[\pmvB]{Play}$.
Choosing either option requires the player to satisfy its funding precondition:
here, both clauses require paying $v=1$ tokens $\tokT$,
since the internal parameter $v$ is set to $1$ by the caller $\cVar{Init}$.
Now, assume that $\cVar[\pmvA]{Play}$ was chosen.
The clause offers two new options: 
either calling $\cVar[\pmvB]{Play}$ with a doubled internal parameter $v$,
or sending $v$ units of $\tokT$ to player $\pmvA$ after 5 time units.
The first option requires to pay other $v=1$ tokens to the contract,
so that its new balance satisfies the funding precondition of $\cVar[\pmvB]{Play}$.
After that, both players can take turns doubling the contract balance,
until one of them fails to do so within 5 time units.
When this happens, the other player can redeem the whole contract balance,
ending the game.

%% The funding precondition of both $\cVar{Play}$ clauses is $v\tokT$.
%% These clauses are \emph{mutually recursive},
%% having a branch that calls the other while doubling the internal parameter $v$.
%% When $\cVar[\pmvA]{Play}$ is active and $\cVar[\pmvB]{Play}$
%% has not been invoked within 5 time units,
%% then $\pmvA$ has won the game and can terminate the contract
%% by choosing $\withdrawC(\splitB{v\tokT}{\pmvA})$.
%% The behaviour of $\cVar[\pmvB]{Play}$ is dual.
%% Note that, if $\cVar[\pmvA]{Play}$ is active,
%% then half of the sum needed to invoke $\cVar[\pmvB]{Play}$
%% is already held by the process.
%% Player $\pmv{B}$ only needs to add the other half in order to perform the $\callname$.

In the game seen so far, the contract balance is fully determined by the contract itself, with players having no choice in regard to how many tokens they can add.
External parameters allow us to make the game more interesting,
letting players arbitrarily raise the bet as long as it is greater than
the double of the previous balance:
\begin{align*}
  \clauseCallNonVector{\cVar[\pmvA]{Play}}{v}{w}
  & =
  \clauseFunding{w\tokT}{ w>2v}
  \ (
  \ncadv{\clauseCallNonVector{\cVar[\pmvB]{Play}}{w}{?}}
  \\
  &  \qquad  +\afterRelC{5}{\withdrawC{ (\splitB{w\tokT}{\pmvA} )}} \ )
  \\
  \clauseCallNonVector{\cVar[\pmvB]{Play}}{\bar{v}}{\bar{w}}
  & =
  \clauseFunding{\bar{w}\tokT}{\bar{w} > 2 \bar{v}}
  \ (
  \ncadv{\clauseCallNonVector{\cVar[\pmvA]{Play}}{\bar{w}}{?}}
  \\
  &  \qquad  +\afterRelC{5}{\withdrawC{ (\splitB{\bar{w}\tokT}{\pmvB} )}} \ )  
\end{align*}

For instance, assume that 
$\clauseCallNonVector{\cVar[\pmvA]{Play}}{v}{w}$ is active
when a player executes
\mbox{$\ncadv{\clauseCallNonVector{\cVar[\pmvB]{Play}}{w}{?}}$}.
The internal parameter \mbox{$\bar{v}=w$} is determined by the process, while
the external parameter \mbox{$\bar{w}=?$} is chosen by the player,
provided that it respects the funding precondition of $\cVar[\pmvB]{Play}$,
namely \mbox{$\bar{w} > 2 \bar{v}$}.
This means that the player must pay enough tokens to
make the contract balance reach $2 w$ tokens,
doubling the previous balance $w$.

\paragraph*{More on \langname clauses}

Generalising from the previous examples,
processes $\contrC$ are choices $\contrD[1] + \cdots + \contrD[k]$
among one or more branches.
Branches have two forms:
\begin{itemize}

\item \mbox{$\withdrawC{(\splitB{v_1\tokT[1]}{\pmvA[1]}\Vert \cdots \Vert \splitB{v_n\tokT[n]}{\pmvA[n]})}$}.
  This branch sends $v_i$ tokens of type $\tokT[i]$ to
  participant $\pmvA[i]$ (for all $i$).
  The needed funds are taken from
  the process balance and from the contributing participants.
  % terminates the process
  
\item \mbox{$\ncadv{(\clauseCallX[1]{\sexp[1]}{?}\Vert \cdots\Vert \clauseCallX[n]{\sexp[n]}{?})}$}.
  This branch invokes the clauses $\cVarX[i]$ in parallel,
  with the actual internal parameters given by
  the \emph{expressions} $\vec{\sexp[i]}$.
  The external parameters $\vec{?}$ represent values chosen by participants.  
  The $\callname$ operation is enabled when the funding preconditions
  of \emph{every} $\cVarX[i]$ are satisfied.
  Like in the previous case,
  the needed funds are taken from the process balance,
  and from additional funds possibly sent by participants.
\end{itemize}

Branches can be decorated with time constraints and authorizations.
% similarly to the high-level language.
Time constraints enable a branch only after a certain time has passed.
They can be absolute ($\afterC{t}{\contrD}$), making the branch $\contrD$
enabled since a certain time $t$, 
or relative ($\afterRelC{\delta}{\contrD}$),
making $\contrD$ enabled after a delay $\delta$ since the previous contract step.
Authorizations ($\authC{\pmvA}{\contrD}$)
enable a branch only when $\pmvA$ has provided their authorization.
Note that multiple authorizations are possible
(\eg, $\authC{\pmvA}{\authC{\pmvB}{\contrD}}$).
In this case, all the involved participants must agree on the chosen branch.
Furthermore, if the branch is a call, then they must also agree on the values
of the external parameters.
Effectively, all the external parameters are chosen by the authorizers.

To stipulate a contract, participants spawn an instance of a clause,
providing the funds required by its funding precondition $\vec{v\tokT}$.
Note that $\vec{v\tokT}$ can be a sequence
$v_1 \tokT[1] \cdots v_n \tokT[n]$, meaning that 
$v_i$ units of each token $\tokT[i]$ are required.
%When a contract is stipulated participants chose an initial clause which is made \emph{active}.
Doing so activates the clause, running its process, which can in turn
spawn instances of other clauses via $\callname$,
transferring the control to them.
Spawning multiple instances of clauses is possible by
exploiting the inherent parallelism of the UTXO model.
We take advantage of this parallelism in the following example.

\paragraph*{Exploiting parallelism}

% So far we have seen contracts with a single active process at any time.
% In some use cases, it is useful to split the contract over \emph{multiple}
% processes running in parallel.
We specify a ``Ponzi scheme'' contract as follows:
\begin{align*}
  \clauseDefNonVector{\cVar{P}}{v}{A} &= \clauseFunding{v \tokT}{} \
  \ncadv{(\clauseCallNonVector{\cVar{S}}{2v\commaCode A}{}\Vert \clauseCallNonVector{\cVar{P}}{2v}{?}\Vert \clauseCallNonVector{\cVar{P}}{2v}{?})}
  \\
  \clauseDefNonVector{\cVar{S}}{v \commaCode A}{} &= \clauseFunding{v\tokT}{} \ \withdrawC{(\splitB{v\tokT}{A})}
\end{align*}

The clause $\cVar{P}$ takes an integer $v$ as an internal parameter,
and a participant $A$ as an external parameter (denoting the owner).
The funding precondition requires $v$ units of $\tokT$.
The clause $\cVar{P}$ calls $\cVar{S}$ along with two copies of itself
(each one doubling the internal parameter $v$).
The clause $\cVar{S}$ simply transfers some tokens according to its
internal parameters $v$ and $A$
(note that they are before the ``\code{;}'').
When $\clauseCallNonVector{\cVar{P}}{v}{A}$ is active,
to continue the contract we need to satisfy the preconditions of all called clauses,
which require $6v\tokT$ overall. Since the contract balance is $v\tokT$,
participants must provide $5v\tokT$.
In practice, the owner $A$ will need to convince two participants $B$ and $C$
to provide $2.5v\tokT$ each, in exchange for setting themselves
as owners in the newly spawned copies of $\cVar{P}$.
When the $\callname$ is performed, the former owner $A$ receives $2v\tokT$. 
If $B$ and $C$ later manage to enrol two other participants in the scheme,
they will receive $4v\tokT$, gaining $4v\tokT - 2.5v\tokT=1.5v\tokT$.
Note that if $C$ does not find new participants, but $B$ does, $B$ can still continue its process, since each copy of $\cVar{P}$ executes independently.
We remark that $B$ and $C$ are \emph{not} known when $A$ is enrolled:
this is why we need to use external parameters.

Upon compilation, parallel active contracts can be concurrently executed by
UTXO blockchain nodes by exploiting their internal parallelism.
This would not be possible with a traditional stateful account-based
implementation, where a single contract would process all transactions.

\subsection{The UTXO model}
\label{sec:overview:utxo}

\newcommand{\newBid}{\txTag{}{newBid}}
\newcommand{\oldBid}{\txTag{}{oldBid}}
\newcommand{\Bidder}{\txTag{}{Bidder}}
\newcommand{\auctionA}{\txTag{}{A}}
\newcommand{\auctionV}{\txTag{}{v}}

\langname contracts can be executed on a bare-bone UTXO blockchain,
with no Turing-complete scripting language and no gas mechanism.
Basically, a UTXO model similar to Bitcoin's is enough,
with the addition of custom tokens and
\emph{covenants}~\cite{Moser16bw,Oconnor17bw,BLZ20isola}.

A transaction $\txT$ in this UTXO model is a tuple with
the following fields, similarly to Bitcoin:
\begin{itemize}

\item $\txOut{}$ is a sequence of \emph{transaction outputs}, \ie triples of the form
  $(\txarg, \txscript, \txval)$,
  where $\txarg$ is a sequence of values
  (which we use to encode the contract state),
  $\txscript$ is the script specifying the spending condition,
  and $\txval$ encodes the tokens held by the output ($\vec{v \tokT}$).  

\item $\txIn{}$ is a sequence of \emph{transaction inputs},
  referring to the transaction outputs which are going to be spent by $\txT$.
  An input $(\txTi,i)$ refers to the $i$-th output of
  a previous transaction $\txTi$.

\item $\txWit{}$ is a sequence of \emph{witnesses} (sequences of values),
  passed as parameters to the
  scripts of input transactions.
  More precisely, if $\txIn[j]{} = (\txTi, i)$,
  then $\txWit[j]{}$ is the witness
  passed to the script $\txTi.\txOut[i]{}.\txscript$.

\end{itemize}

As a simple example, we display below a transaction $\txT[1]$
containing a single output,
% (the other fields are omitted),
which holds $1 \tokT$
and can be redeemed by any transaction carrying a signature of $\pmvA$
in its $\txWit{}$ field (referred to by $\rtx.\txWit{}$).

\resizebox{0.85\columnwidth}{!}{
  \hspace{-0pt}
  \begin{nscenter}
    \small
    \begin{tabular}[t]{|l|}
      \hline
      \\[-9pt]
      \multicolumn{1}{|c|}{$\txT[1]$} \\[-1pt]
      \hline
      $\cdots$ \\
      \txOut[1]{$\{ \txarg: \cdots$,} \\[1pt]
      \hspace{30pt} $\txscript: \versig{\constPK[\pmvA]}{\rtx.\txWit{}}$, 
      \hspace{0pt} \textcolor{Gray}{// verify signature} \\[1pt]
      \hspace{30pt} $\txval: 1 \tokT \}$ \\[1pt]                  
      \hline
    \end{tabular}
  \end{nscenter}
}

\medskip
In order to redeem the tokens held in a transaction output
$\txT.\txOut[j]{}$,
a transaction $\txTii$ has to satisfy the spending condition
$\txT.\txOut[j]{}.\txscript$.
This script can access the fields of $\txT$ and $\txTii$,
perform basic arithmetic, logical and cryptographic  operations
(hashing, signature verification of the redeeming transaction),
and enforce time constraints.
Covenant operators allow the script in $\txT$ to constrain
the scripts in $\txTii$.
The covenant  $\verscript{\sigma}{n}$ mandates the $n$-th output
of $\txTii$ to have a script equal to $\sigma$.
The covenant $\verrec{n}$ requires the $n$-th output of $\txTii$
to have the same script of $\txT.\txOut[j]{}$, the one currently being checked.

As an example, consider the following transaction
redeeming (the only output of) $\txT[1]$:

\resizebox{0.95\columnwidth}{!}{
  \hspace{-25pt}
  \begin{nscenter}
    \small
    \begin{tabular}[t]{|l|}
      \hline
      \\[-9pt]
      \multicolumn{1}{|c|}{$\txT[2]$} \\[-1pt]
      \hline \\[-7pt]      
      \txIn[1]{$(\txT[1],1)$} \\[1pt]
      \txWit[1]{$\sig{\pmvA}{\txT[2]}{}$} \\[1pt]
      \txOut[1]{$\{ \txarg: \constPK[\pmvA], $} \\[1pt]
      \hspace{30pt} $\txscript: \versig{\seqat{\ctxo{\txarg}}{1}}{\rtx.\txWit{}} \andE$ 
      \hspace{0pt} \textcolor{Gray}{// verify signature} \\[1pt]
      \hspace{48pt} $\rtxo{\txval}{1} = 1 \tokT \andE$ 
      \hspace{32pt} \textcolor{Gray}{// preserve value} \\[1pt]
      \hspace{48pt} $\verrec{1},$ 
      \hspace{80pt} \textcolor{Gray}{// preserve script} \\[1pt]
      \hspace{30pt} $\txval: 1 \tokT \}$ \\[1pt]
      \hline
    \end{tabular}
  \end{nscenter}
}

\bigskip
The transaction $\txT[2]$ can redeem its input,
since it carries a witness (a signature of $\pmvA$)
that satisfies the script $\txT[1].\txOut[1]{}.\txscript$. 
Furthermore the assets in $\txT[2]$'s output do not exceed those in its inputs.
The spending condition of $\txT[2]$ is given by the
$\txarg$ field (containing $\pmvA$'s public key $\constPK[\pmvA]$),
and the script $\txscript$.
The script is a conjunction of three conditions:
\begin{itemize} 

\item the $\txWit{}$ness of the redeeming transaction
  must contain a transaction signature, to be verified against
  the public key stored in the 1st element of the $\txarg{}$ sequence
  of the current transaction output
  (denoted by $\ctxo{}.\txarg.1$).
  In $\txT[2]$, this is just $\constPK[\pmvA]$.

\item the $1$st output of the redeeming transaction must have $1 \tokT$ value
  (here $\rtxo{}{i}$ is the $i$-th output of the \emph{redeeming} transaction).
  
\item the $1$st script of the redeeming transaction must be equal to
  the current one.
  This is enforced using the covenant $\verrec{1}$.

\end{itemize}

This transaction actually implements a sort of Non-Fungible Token (NFT).
To transfer the NFT to $\pmvB$, $\pmvA$
spends $\txT[2]$ with a redeeming transaction $\txT[3]$,
writing her signature in the $\txWit{}{}$ field of $\txT[3]$, and 
setting the $\txarg{}$ field to $\pmvB$'s public key $\constPK[\pmvB]$.
Note that the script and the balance are preserved along transactions
thanks to the covenant.

\subsection{The \langname compiler}

One of our main contributions is a compiler from
\langname contracts to UTXO transactions.
Intuitively, we encode active clauses into transaction outputs, where
the $\txval$ field records the contract balance,
$\txscript$ enforces the contract logic, and
$\txarg$ records the contract state.

% Consider an active clause with actual parameters
% $\clauseCallXParam$
% $\stVal[i]$ and $\dinVal[i]$.
The $\txarg$ field of a transaction output encoding an active clause
will have one entry for each actual parameter,
% ($\txTag{}{In_i}$ and $\txTag{}{Ex_i}$),
and the following additional entries:
$\txname$ to represent the clause name,
$\txbranch$ to represent the index of the executed branch, and $\txnonce$ to keep the behaviour faithful
to the \langname semantics.
The $\txscript$ field of \emph{all} transactions outputs
resulting from the compilation of an \langname contract is the same,
and it is preserved along chains of transactions
by using the $\verrec{}$ covenant.

Here we illustrate the compilation of an auction contract focussing on the construction of the script.

\paragraph*{An auction contract}

\newcommand{\ILLnewBid}{\mathit{newBid}}
\newcommand{\ILLoldBid}{\mathit{oldBid}}
\newcommand{\ILLBidder}{\mathit{Bidder}}
The contract (\Cref{fig:auction:illum}) consists of three clauses: $\cVar{Init}$ that takes no parameters and initialises the auction with a starting bid of 0 tokens; $\cVar{Bid}$ that allows a $\ILLBidder$ to raise the $\ILLoldBid$ to a $\ILLnewBid$; and $\cVar{Pay}$ which transfers tokens to a participant.
The contract flows as follows.
After the initialisation, a participant can call the clause $\cVar{Bid}$ to start the auction, setting themselves as the highest bidder.
Then, we must execute one of the two branches of $\cVar{Bid}$.
The first branch raises the bid, setting the values of $\ILLnewBid$ and $\ILLBidder$ through external parameters, and refunds the previous bidder.
The second branch closes the auction, transferring the tokens to the hardcoded $\pmv{Owner}$.
Since the first branch of $\cVar{Bid}$ recursively calls the clause $\cVar{Bid}$ itself, it can only be taken if the funding precondition is satisfied, which means that the new bid must be greater than the previous one.
On the other hand, the second branch of $\cVar{Bid}$ can only be executed by the owner after a deadline of 1000 time units.
By choosing this branch the owner closes the auction,
and receives the highest bid.

\begin{figure}
  \begin{mdframed}
    \small
    \negcaptionspace
  \begin{align*}
    & \clauseDefNonVector{\cVar{Init}}{}{} =
    \clauseFunding{0}{} \ \ncadv{\clauseCallNonVector{\cVar{Bid}}{0}{ ? \commaCode ? }}
    \\[5pt]
    & \clauseDefNonVector{\cVar{Bid}}{\ILLoldBid}{\ILLnewBid \commaCode \ILLBidder}
    \\
    & = \clauseFunding{\ILLnewBid \, \tokT}{ \ILLnewBid > \ILLoldBid \andE \ILLBidder \neq \pmv{Null} }
    \\
    & \hspace{15pt} \ncadv{ ( \clauseCallNonVector{\cVar{Bid}}{\ILLnewBid}{? \commaCode ?} \ \Vert \ \clauseCallNonVector{\cVar{Pay}}{\ILLnewBid\commaCode \ILLBidder}{} )} 
    \\ 
    & \ + \afterC{1000}{} \authC{\pmv{Owner}}{}\withdrawC{(\splitB{\ILLnewBid \, \tokT}{\pmv{Owner}})}
    \\[5pt]
    & \clauseDefNonVector{\cVar{Pay}}{v \commaCode A }{} = \clauseFunding{v \tokT}{} \ \withdrawC{(\splitB{v\tokT}{A})}
  \end{align*}
  \end{mdframed}
  \negcaptionspace  
  \caption{An auction contract in \langname.}
  \label{fig:auction:illum}
\end{figure}

% In this example we have directly specified the auction contract in $\langname$.
% We remark that our approach also enables the
% \emph{automatic translation into \langname of any contract written in the high-level language}.
% as we will show in~\Cref{sec:illum}. bart: adesso e' in appendice

Compiling the contract, we obtain the following script:
\begin{align*}
  \scr_{ \sf Auction } \coloneqq \ &\inidx = 1 \andE
  \\
  &\ifE{ (\ctxo{\txname} = \cVar{Init}  ) }{ \scr_{\sf Init} }{}
  \\
  &\ifE{ (\ctxo{\txname} = \cVar{Bid}  ) }{ \scr_{\sf Bid}}{}
  \\
  &\ifE{(\ctxo{\txname} = \cVar{Pay} )}{ \scr_{\sf Pay} }{\false}
\end{align*}

The condition \mbox{$\inidx = 1$} checks that this output
is redeemed by an input at position 1.
This is needed to thwart attacks where a transaction spends two contracts at once, effectively cancelling one of them.
The rest of the script is a switch among the possible clauses,
where $\ctxo{\txname}$ denotes the $\txname$ item of the
$\txarg$ field in the current transaction output.
For brevity here we only illustrate the most interesting case, \ie $\scr_{\sf Bid}$.
Recall that the $\cVar{Bid}$ process is a choice between two branches. Consequently, the associated script has the following form:
\begin{align*}
  \scr_{\sf Bid} \coloneqq \ &  {\sf if \ BranchCond1	\ then \ } \scr_{ \sf Branch1 }
  \\
  & {\sf else \ if \ BranchCond2 \ then \ } \scr_{ \sf Branch2 } 
  \\ 
  & {\sf else \ } \false
\end{align*}

The first branch calls \emph{two} clauses, \ie $\cVar{Bid}$ and $\cVar{Pay}$:
\[
 \ncadv{ ( \clauseCallNonVector{\cVar{Bid}}{\ILLnewBid}{? \commaCode ?} \ \Vert \ \clauseCallNonVector{\cVar{Pay}}{\ILLnewBid\commaCode \ILLBidder}{} )}
\]
Consistently, ${\sf BranchCond1}$ checks that the redeeming transaction
($\rtx$) has exactly \emph{two} outputs:
this is the goal of the condition $\outlen{\rtx} = 2$ below.
Furthermore, ${\sf BranchCond1}$ checks that the chosen branch
is indeed the first one: this is done
by requiring that both outputs in the redeeming transaction
($\rtxo{}{1}$ and $\rtxo{}{2}$)
have the argument $\txbranch$ set to 1. 
\begin{align*}
  {\sf BranchCond1} \coloneqq & \ \outlen{\rtx} = 2 \andE \rtxo{\txbranch}{1} = 1 
  \\
  &\ \andE \rtxo{\txbranch}{2} = 1
\end{align*}

The second branch instead performs a $\withdrawname$ operation
with exactly \emph{one} recipient, so we check that
the redeeming transaction has exactly \emph{one} output, which has its $\txbranch$ argument set to 2.
\[
  {\sf BranchCond2} \coloneqq \ \outlen{\rtx} = 1 \andE \rtxo{\txbranch}{1} = 2
\]
  
The script $\scr_{ \sf Branch1 }$ verifies that the two outputs of the redeeming transaction encode, respectively, the clauses $\cVar{Bid}$ and
$\cVar{Pay}$.
The part corresponding to $\cVar{Bid}$ is:
\begin{align*}
  & \rtxo{\txname}{1} = \cVar{Bid} \andE
  \\
  & \verrec{1} \andE |\rtxo{\txarg}{1}| = 6 \andE
  \\
  & \rtxo{\oldBid}{1} = \ctxo{\newBid} \andE
  \\
  & \rtxo{\txval}{1} = \rtxo{\newBid}{1}\ \tokT \andE
  \\
  & \rtxo{\newBid}{1} > \rtxo{\oldBid}{1} \andE
  \\
  & \rtxo{\Bidder}{1} \not= \pmv{Null} \andE \cdots
\end{align*}

The script performs the following checks on the redeeming transaction:
\begin{inlinelist}

\item the clause name in the first output is indeed $\cVar{Bid}$;
  
\item the script is preserved (via the $\verrec{}$ covenant);
  
\item the number of arguments is correct
  (3 arguments for $\newBid$, $\oldBid$ and $\Bidder$, and 3 arguments
  for $\txname$, $\txbranch$ and $\txnonce$);
  
\item the value of the $\oldBid$ of the redeeming transaction
  is set to the $\newBid$ of the current transaction,
  coherently with the passing of parameters in the $\cVar{Bid}$ clause;

\item the amount of tokens of type $\tokT$ transferred to
  the redeeming transaction is the one specified in the funding precondition;
  
\item the guard in the funding precondition is satisfied.
\end{inlinelist}

The part of the script corresponding to $\cVar{Pay}$ is obtained in the same way:
\begin{align*}
  \cdots \
  & \rtxo{\txname}{2} = \cVar{Pay} \andE
  \\
  & \verrec{2} \andE |\rtxo{\txarg}{2}| = 5 \andE
  \\
  & \rtxo{\auctionA}{2} = \ctxo{\Bidder} \andE
  \\
  & \rtxo{\auctionV}{2} = \ctxo{\newBid} \andE
  \\
  & \rtxo{\txval}{2} = \ctxo{\newBid} \ \tokT
\end{align*}

The script $\scr_{\sf Branch2}$ encodes the $\withdrawname$ operation
in \Cref{fig:auction:illum}, enforcing the
authorization of $\pmv{Owner}$ and the absolute time constraint
of 1000 time units:
\begin{align*}
& \afterAbs{1000}{} \versig{\pmv{Owner}}{\rtx.\txWit{}(1)} \andE  
\\
& \verscript{ \versig{\ctxo{\txowner}}{\rtx.\txWit{}(1)} }{1} \andE
\\
& |\rtxo{\txarg}{1}|= 3 \andE
\\
& \rtxo{\txowner}{1} = \pmv{Owner} \andE
\\ 
& \rtxo{\txval}{1} = \ctxo{\newBid} \ \tokT 
\end{align*}

The script performs the following checks on the redeeming transaction:
\begin{inlinelist}

\item ${\sf absAfter}$ forces the redeeming transaction to be published
  at a time greater than 1000;

\item the first $\versigName$ checks the presence of $\pmv{Owner}$'s
  signature in the witness of the redeeming transaction;
  
\item the only output of the redeeming transaction
  has 3 arguments ($\txowner$, $\txbranch$, and $\txnonce$),
  and a script that accepts any transaction signed by the $\txTag{}{owner}$
  (this is enforced by the $\verscript{}{}$ covenant). This effectively transfers the ownership of the tokens to $\pmv{Owner}$.  
\end{inlinelist}
The details of our compilation technique are in~\Cref{sec:compiler}.

\subsection{Security of the \langname compiler}

Our main technical result is the security of the compiler
establishing a strict correspondence between
actions at the \langname level
and those at the blockchain level.
This ensures that any contract behaviour 
that is observable at the blockchain level is also observable
in the semantics of \langname.
In particular, any attack that may happen at the blockchain level
can be detected by inspecting the symbolic semantics of \langname contracts.
This is a fundamental step towards static analysis tools
for the verification of security properties of contracts in the UTXO model.

Here we outline how this result is proved
in~\Cref{sec:adversary,sec:soundness}.
We start by defining the adversary model,
both at the symbolic level of \langname
and at the computational level of the blockchain.
The adversary is modelled as a PPTIME algorithm that schedules the
actions chosen by participants, possibly interleaving them with
adversarial actions.
The bridge between the symbolic and the computational level
is given through a \emph{coherence relation},
which associates symbolic actions (\eg, a $\callname$ action)
with their computational counterparts (\eg, a transaction).
The definition of this coherence relation is quite gruelling,
as it must consider all possible actions,
which amount to 20 cases: this complexity of course reflects on the proofs.
As a first sanity check, we show in  \Cref{lem:balance-preservation} that the coherence relation precisely characterises the exchange of assets: namely, the asset ownership is consistent between symbolic and computational executions whenever they are coherent.
Our main security result (\Cref{th:computational-soundness})
guarantees that any computational execution
is coherent with some symbolic execution, up-to a negligible error probability.
Together with \Cref{lem:balance-preservation},
it proves that computational exchanges of assets, including those mediated by contracts,
are always mirrored at the symbolic level.

\section{The \illum intermediate language}
\label{sec:illum}

We now refine the description of \illum given in~\Cref{sec:overview},
by providing its syntax and semantics.
% We also illustrate the translation from the high-level language seen before
% and \mbox{\illum}, and establish its Turing completeness.
Because of space constraints, we omit some technicalities,
relying on examples and intuitions: 
see \iftoggle{arxiv}{Appendix~\ref{sec:app-symbolic}}{\cite{BMZ23arxiv}} for full details.

\paragraph*{Syntax}

We assume a set $\Part$ of participants,
(ranged over by $\pmvA, \pmvB, \cdots$, and by a dummy participant $\pmv{Null}$).
Contracts and deposits are denoted by the lowercase letters $x,y,\cdots$,
while clauses will have names $\cVarX,\cVarY, \cdots$.
Clause parameters will be denoted by $\stPar[i]$ and $\dinPar[i]$,
while the actual values substituted to those parameters will be denoted by $\stVal[i]$ and $\dinVal[i]$.
Arithmetic expressions
(integer constants and parameters, basic operations, hashes)
are denoted by $\sexp[i]$, while participant expressions (constants and parameters) are denoted by $\nexp[i]$.
The value of parameters and expressions can also be key-value mappings.
The domain of a mapping can be chosen to be either the integers or the
participants.
The codomain can be chosen similarly.
If $\mathcal{M}$
is a mapping expression, we denote with $\mathcal{M}[\cdots]$
the access to one of its values, and with
$\mathcal{M}[\cdots \rightarrow \cdots]$
the update of one of its associations.
Sequences are denoted in bold, with $\vec{x} = x_1 \cdots x_n$.

We remark that the precise set of data types
% and their operations
that can be used in parameters and expressions is not fundamental to the
design of \illum.
Indeed, our design can be easily adapted to different
data types by suitably altering the syntax and semantics of expressions.
To support compilation to the UTXO model, we simply require that the
underlying blockchain scripting language supports the same data types
and operations.
We assume that data types include at least
integers and participants, since their r\^ole is crucial to \illum
constructs.
Key-value mappings, instead, are not as crucial, and could
be removed if not supported by the underlying blockchain, at the
expense of reducing the usability of \illum.
Throughout the paper, we mostly showcase examples that use only the
fundamental data types (integers and participants).
Notably, the \illum compiler handles all these contracts
without requiring mappings to be supported by the compilation target.
In~\Cref{sec:discussion} we will also exploit mappings to
discuss more complex contracts.

\begin{defn}[Clauses]
  A clause is defined by an equation
  \[
  \clauseDefXParam = \clauseAdv{\vec{\sexp\tokT}}{p}{\contrC}
  \]
  where $\clauseFunding{\vec{\sexp\tokT}}{p}$ is the \emph{funding precondition}, and $\contrC$ is a \emph{process}.
  % (abbreviated as contract, with a slight abuse of terminology).
  The clause takes two sequences of parameters $\vec{\stPar}$ and $\vec{\dinPar}$.
  Parameters can be of any type (integers, participants, or mappings).
  These types are always clear from the context, hence omitted.
  We require that all the parameter names in
  $\vec{\sexp}$, $p$, and $\contrC$ are present
  in $\vec{\stPar}$, $\vec{\dinPar}$.
\end{defn}

%More specifically: $\stPar[j]$ and $\dinPar[j]$ are respectively the \emph{internal} and \emph{external} parameters.
When $\cVarX$ is invoked, the calling process provides
the actual internal parameters, while the participants
who are performing the call choose the actual external ones.
The funding precondition $\clauseFunding{\vec{v\tokT}}{p}$
encodes the requirements for the invocation of $\cVarX$.
Namely, the sequence $\vec{v\tokT} = v_1 \tokT[1] \cdots v_n \tokT[n]$
asks the participants to transfer $v_i$ tokens of type $\tokT[i]$
to the process $\contrC$ (for all $i$).
Moreover, $p$ is a boolean condition on the parameters that must hold.
We write $\clauseFunding{\vec{v\tokT}}{}$ for $\clauseFunding{\vec{v\tokT}}{\true}$.

\begin{defn}[Processes]
  Processes have the following syntax:
  \begin{align*}
    \contrC & \bnfdef \textstyle \sum_{i \in I} \contrD[i]
    && \text{process}
    \\[4pt]
    \contrD	& \bnfdef
    && \text{branch} 
    \\
    & \ncadv{( \cdots \mmid \clauseCallX[i]{\stVal[i]}{?} \mmid \cdots )}
    && \text{call clauses $\cdots \cVar[i]{X} \cdots$}
    \\
    \bnfmid & \withdrawC{(\cdots \mmid \sexp[i]\tokT[i] \rightarrow \nexp[i] \mmid \cdots )}%\sexp[n] \rightarrow \nexp[n])}
    && \text{transfer $\sexp[i]\tokT[i]$ to each $\nexp[i]$}
    \\
    \bnfmid	& \authC{\nexp}{\contrD}
    && \text{wait for $\nexp$ authorization}
    \\
    \bnfmid	& \afterC{\sexp}{\contrD}
    &&  \text{wait until time $\sexp$}
    \\
    \bnfmid	& \afterRelC{\sexp}{\contrD}
    && \text{wait $\sexp$ after activation} 
  \end{align*}
  where we assume that:
  \begin{inlinelist}

  \item each clause name $\cVar{X}$
    % appearing within a $\callname$
    has a unique defining equation
    \mbox{$\clauseDefXParam= \clauseAdv{\vec{\sexp\tokT}}{p}{\contrC}$};
    
  \item the sequence $\vec{\stVal}$ of actual parameters passed
    to a called clause $\cVarX[i]$ matches, in length and typing,
    the sequence $\vec{\stPar}$ of formal (internal) parameters;
    % specified in $\cVarX[i]$'s defining equation;

  \item the order of decorations is immaterial.
    
  \end{inlinelist}
\end{defn}

A clause $\cVarX(\cdots)$
together with two correctly typed sequences of actual parameters
$\vec{\stVal}$ and $\vec{\dinVal}$ is said to be an \emph{instantiated clause}, and denoted by $\clauseCallXParam$.

\paragraph*{Semantics}

The execution of contracts is modelled as a transition relation
between \keyterm{configurations}, that are abstract representations
of the blockchain state.
In a configuration, tokens can be stored in deposits and active contracts.

A \keyterm{deposit} $\confDep[x]{\pmvA}{v  \tokT}$ represents $v$ units of token $\tokT$ owned by~$\pmvA$.
It is uniquely identified by the name $x$,
and can only be spent upon $\pmvA$'s authorization.
A term $\confContrTime[x]{t}{\contrC}{\vec{v\tokT}}$ is an
\keyterm{active contract}, where $x$ is the unique identifier,
$t$ is the activation time,
and $\vec{v\tokT}$ is the balance, 
which can only be transferred according to the contract logic specified by~$\contrC$.
\begin{comment}
All the outputs not following this form are abstracted in our configuration as the term $\confTaint{\vec{w\tokT}}$ where $\vec{w\tokT}$ is the sum of all the outputs' values.
We assume that only dishonest participants can spend the
outer tokens $\confTaint{\vec{w\tokT}}$.
\bartnote{togliere??}
\end{comment}

Besides the terms used to store tokens,
in a configuration we also have advertisements and authorizations.

\keyterm{Advertisement} terms are used by a participant
to propose one of the following actions:
\begin{itemize}

\item The activation of a new contract.
  This is done by advertising
  $\confF = \confAdvInitN{\clauseCallXParam}{\vec{z}}{}{h}$, which specifies the instantiated clause $\clauseCallXParam$
  and a nonempty list $\vec{z}$ of deposit names that will be spent to fund the contract.
  % together with amount of additional outer tokens $\vec{w\tokT}$ (if no outer tokens are used $\vec{w\tokT}$ is replaced by $\star$).
  The index $h$ is just a nonce used to differentiate between two otherwise identical advertisements.

\item The continuation of an active contract.
  This is done by advertising $\confF = \confAdvN{\complD}{\vec{z}}{}{x}{j}{h}$.
  The list $\vec{z}$ specifies the deposits that will be spent
  for the continuation and added to the balance of $x$.
  The index $h$ is again a nonce. % z può essere non vuota a differenza del caso precedente
  The term $\complD$ is an \emph{advertised branch},
  constructed by taking $\contrD$, the $j$-th branch of $x$, and instantiating the question marks $?$ appearing in a $\callname$
  with the actual values $\dinVal$.

%% \item The \emph{destruction} of a set of deposits, which transforms their balances into outer tokens. This is done by advertising $\confF=\confAdvDesN{\vec{z}}{}{h}$, where $\vec{z}$ is the list of deposits that are going to be destroyed, and $h$ is like in the previous cases. 

\end{itemize} 

\keyterm{Authorization} terms are used by participants
to enable the spending of deposits
and to enable the execution of a contract branch
decorated by $\authC{\pmvA}{\cdots}$.
Authorizations have the form $\confAuth{\pmvA}{\chi}$,
where $\pmvA$ is the authorizing participant,
and $\chi$ denotes the authorized action.
We see here two cases of authorization terms,
relegating the others to Appendix~\ref{sec:app-symbolic}:
\begin{itemize}

\item $\confAuth{\pmvA}{\authBranch{z}{\confF}}$
  authorises the spending of a deposit $z$ owned by $\pmvA$ that appears in
  the advertisement $\confF$;

\item $\confAuth{\pmvA}{\authBranch{x}{\confF}}$
  authorises the continuation of the \mbox{$j$-th} branch of a contract $x$, as advertised by
  \mbox{$\confF= \confAdvN{ \authC{\pmvA}{\complD} }{\vec{z}}{}{x}{j}{h}$}.

\end{itemize}

\begin{defn}[Configurations]
  A configuration $\confG$ is a term $\tilde{\confG} \mid t$,
  where $t \in \Nat$ denotes the time,
  and the pre-configuration $\tilde{\confG}$ has the following syntax:
  \begin{align*}
    %	\confG \ & \bnfdef (\tilde{\confG} \mid \confTaint{\vec{w\tokT}} \mid t) && \text{Full configuration with time $t$ and} 
    %	\\
    %	& &&\text{ destroyed funds counter $\confTaint{\vec{w\tokT}}$}
    %	\\
    \tilde{\confG} \ \bnfdef \;
    % \text{pre-configuration} 
    %% \ & \emptyset   && \text{empty}
    %% \\
    & \confContrTime[x]{t}{\contrC}{\vec{v\tokT}} && \text{active contract}
    \\
    \bnfmid & \confDep[x]{\pmvA}{v\tokT} && \text{deposit}
    \\
    \bnfmid & \confF && \text{advertisement}
    %% \\
    %% \bnfmid & \Theta && \text{incomplete advertisement}
    \\
    \bnfmid &\confAuth{\pmvA}{\chi} && \text{authorization}
    \\
    \bnfmid & (\tilde{\confG} \mid \tilde{\confG}) && \text{parallel composition}
  \end{align*}
  We assume that:
  \begin{inlinelist}
  \item the parallel composition is associative and commutative;
  \item all parallel terms are distinct;
  \item names are unique;
  \item all expressions occurring in active contracts are reduced to constants
    (integers or names).
  \end{inlinelist}
\end{defn}

\paragraph*{An example}

In \Cref{sec:overview} we have discussed the intuition behind
the language semantics.
Here we refine this intuition by precisely illustrating the evolution
of the configuration during the execution of a simple contract.
This example shows the semantics of the main language constructs,
and the role of advertisements and authorizations terms.
\begin{align*}
  \clauseDefNonVector{\cVar{X}}{a}{b}
  & = \clauseAdv{b\tokT}{b>a }{\contrFmt{Wait}(b)}
  \\
  \contrFmt{Wait} (b)
  & = \afterRelC{10}{\withdrawC{(\splitB{b\tokT}{\pmvA})}} \\
  & \, + \authC{\pmvB}{\ncadv{\clauseCallNonVector{\cVar{X}}{b}{?}}}
\end{align*}

The first branch allows $\pmvA$ to withdraw the whole balance
after 10 time units since the contract activation.
The second branch allows $\pmvB$ to temporarily prevent $\pmvA$
from withdrawing: this requires $\pmvB$ to restart the contract
with an increased balance.
We start from the initial configuration:
\[
\confG[0] =  \confDep[{z_1}]{\pmvC}{1\tokT} \mid \confDep[{z_2}]{\pmvB}{2\tokT} \mid \confDep[{z_3}]{\pmvB}{3\tokT} \mid t
\] 
Participant $\pmvC$ starts by advertising
$\confF[0] =\confAdvInitN{\clauseCallNonVector{\cVar{X}}{0}{1}}{z_1}{}{h}$,
then authorizes the use of their deposit $z_1$ in the stipulation,
and finally stipulates the contract, reaching configuration $\confG[1]$:
\begin{align*}
  \confG[0]
  & \rightarrow \confG[0]\mid \confF[0]
  \;\rightarrow\;
  \confG[0]\mid \confF[0] \mid \confAuth{\pmvC}{\authBranch{z_1}{\confF[0]}}
  \\
  &\rightarrow \confContrTime[x]{t}{ \contrFmt{Wait}(1) }{1\tokT} \mid \confDep[{z_2}]{\pmvB}{2\tokT} \mid \confDep[{z_3}]{\pmvB}{3\tokT} \mid t = \confG[1]
\end{align*}

From $\confG[1]$ there are multiple possible continuations. For instance, $\pmvB$ can choose to execute the second branch of $\contrFmt{Wait}$. 
To do so, $\pmvB$ first produces the advertisement
$\confF[1] = \confAdvN{\authC{\pmvB}{\ncadv{\clauseCallNonVector{\cVarX}{1}{3}}} }{z_2}{}{x}{2}{h}$.
Then, $\pmvB$ gives two authorizations: one to satisfy the decoration $\authC{\pmvB}{\cdots}$ in the second branch of $\contrFmt{Wait}$,
and another one to allow the spending of $z_2$.
With these, the configuration can evolve as follows:
\begin{align*}
  \confG[1]
  & \rightarrow \confG[1]\mid \confF[1]
  \; \rightarrow \;
  \confG[1] \mid \confF[1] \mid \confAuth{\pmvB}{\authBranch{x}{\confF[1]}}  
  \\
  & \rightarrow
  \confG[1] \mid \confF[1] \mid \confAuth{\pmvB}{\authBranch{x}{\confF[1]}}\mid \confAuth{\pmvB}{\authBranch{z_2}{\confF[1]}}
  \\
  & \rightarrow \confContrTime[x']{t}{ \contrFmt{Wait}(3) }{3\tokT} \mid \confDep[{z_3}]{\pmvB}{3\tokT} \mid t = \confG[2]
\end{align*}

$\pmvB$ can again choose the second branch,
this time spending $z_3$ to fund its execution.
Otherwise, if $\pmvB$ lets the time pass,
$\pmvA$ can advertise the continuation:
\[
\confF[2] = \confAdvN{\afterRelC{10}{\withdrawC{(\splitB{3\tokT}{\pmvA})}} }{ }{ }{x'}{1}{h}
\]
and then withdraw the contract balance:
\begin{align*}
  \confG[2]
  & \rightarrow \confContrTime[x']{t}{ \contrFmt{Wait}(3) }{3\tokT} \mid \confDep[{z_3}]{\pmvB}{3\tokT} \mid t + 10 = \confG[3]
  \\
  & \rightarrow \confG[3] \mid \confF[2]
  \; \rightarrow \;
  \confDep[z_4]{\pmvA}{3\tokT} \mid \confDep[{z_3}]{\pmvB}{3\tokT} \mid t + 10
\end{align*}

The semantics of \illum has a set of rules for
reducing contracts, and another set for deposits
(see Appendix~\ref{sec:app-symbolic}).

\paragraph*{Turing-completeness}
\illum is Turing-complete: indeed, we can simulate in \illum
any counter machine~\cite{FMR68},
a well-known Turing-complete computational model.
The proof is similar to that in~\cite{BLZ21csf}:
we simulate any counter machine by 
storing each counter in the arguments of recursive clauses.
Incrementing and decrementing the counters is simply done by
specifying the new values of the arguments inside the $\callname$.
Conditional jumps are simulated as choices,
also exploiting clause preconditions.
This construction does not exploit key-value mappings:
it is only based on the assumption that integers are unbounded, as usual.
Notice that, despite its Turing-completeness, \illum can be compiled down to a ``poor'' UTXO blockchain, \ie one with non-Turing-complete scripts.
This is accomplished by spreading the execution of a compiled contract
across multiple transactions, each with its own loop-free script. 
Note that our key-value mappings just feature operators to lookup a single key
and to update a single association:
in this way, even with maps, UTXO scripts can be run in nearly constant-time.
This makes the gas mechanism unnecessary. % for our UTXO model.

\section{Compiling \langname to UTXO scripts}
\label{sec:low-level-model}
\label{sec:compiler}

% As mentioned in~\Cref{sec:overview}, 
The compilation target of \langname
is a UTXO blockchain that is close to Bitcoin,
with minimal extensions in the structure of transactions and in the
scripting language to overcome its expressiveness limitations.
% are tuples
% $(\txIn{}, \txWit{}, \txOut{}, \txAfterAbs{}, \txAfterRel{})$.
% (with minimal notational differences)
% and allows the user to access fields of the current and redeeming transaction.
%% \begin{itemize}
%% \item $\txAfterAbs{}$ specifies an absolute time before which $\txT$
%%   cannot be included in the blockchain.

%% \item $\txAfterRel{}$ is another time constraint:
%%   if $\txAfterRel[j]{}= \delta$ then $\txT$ can only be included
%%   in a block after $\delta$ time units have passed
%%   from the inclusion of its $j$-th input.
%% \end{itemize}

\paragraph*{Scripting language}

We consider a scripting language that extends Bitcoin Script with
covenants, borrowing from~\cite{BLZ20isola}
(see  
\iftoggle{arxiv}{Appendix~\ref{sec:app-computational}}{\cite{BMZ23arxiv}}
for its syntax and semantics).
Here we recap some operators that are used by our compiler.
First,
% we have a way to reference outputs and their fields.
$\ctxo{\txf}$ denotes the field $\txf$ of the current transaction output
that is being spent.
Similarly, $\rtxo{\txf}{e}$ denotes the field $\txf$ of the $e$-th output of the redeeming transaction.
Then, we have the covenants:
$\verscript{\scr}{n}$
checks that the $n$-th output script of the redeeming transaction
is equal to $\scr$, while
$\verrec{n}$ checks that the $n$-th output script of the redeeming transaction
is equal to the one of the output being spent.
The operators $\inidx$ and $\rtxw$ denote, respectively,
the position of the redeeming input among the ones of the redeeming transaction,
and the witness associated to it.
To improve readability, we will use names instead of indices
when referring to arguments in the $\txarg$ sequence
(\eg, we write $\ctxo{\txowner}$ for $\ctxo{\txarg}.1$).

While constructing a contract script we will need to replace the parameters $\stPar[i]$ and $\dinPar[i]$ appearing in an expression $\sexp$ with the respective arguments $\ctxo{\txTag{}{In_i}}$ and $\ctxo{\txTag{}{Ex_i}}$.
To simplify the notation, we denote this substitution with $\ctxo{\sexp}$.
For instance, if the contract contains a term
$\afterC{t}{\cdots}$ with $ t = \stPar[1]+ \dinPar[1] + 3$,
we will write $\ctxo{t}$ instead of
$\ctxo{\txTag{}{In_1}}+ \ctxo{\txTag{}{Ex_1}} + 3$.
Similarly, whenever an expression uses the arguments of a redeeming transaction's output, we denote it as $\rtxo{\sexp}{j}$.

\paragraph*{Representing deposits}

%% The deposits that appear in the symbolic configuration
%% have a low-level counterpart, that encodes them in the blockchain.
We represent a deposit $\confDep{\pmvA}{v\tokT}$ in an \langname configuration
% at the computational level by
as a transaction output with value $v\tokT$,
argument $\txowner$ set to $\pmvA$,
% and the other are $\txnonce$, $\txbranch$, and are used for technical reasons),
and the script:
\[
\versig{\ctxo{\txowner}}{\rtxw.1}
\]
allowing $\pmvA$ to spend the funds by providing her signature.

\paragraph*{How the compiler works}

Representing an active contract $\confContrTime[x]{t}{\contrC}{\vec{v\tokT}}$ 
at the blockchain level is more complex:
we need to consider the clause $\clauseCallXParam$ from which it originated.
The output representing $x$ has a value $\vec{v\tokT}$,
and its arguments are the following:
$\txname$ for the clause name $\cVarX$,
for the actual parameters $\txTag{}{In_i}$ and $\txTag{}{Ex_i}$,
and two technical arguments $\txnonce$ and $\txbranch$. 
The output script of a contract is preserved along executions:
we detail its construction in the next paragraphs,
refining and generalising the intuitions given in~\Cref{sec:overview}
(the full technical details are in 
\iftoggle{arxiv}{Appendix~\ref{sec:app-compiler}}{\cite{BMZ23arxiv}}).
% so it must be constructed by looking at the initial clause that  determined by looking at the initial clause that stipulated it. We detail this construction in the next paragraphs.
%% In the overview, we presented part of the script that encoded a specific auction contract. 
%% Here we generalize, showing how the compiler constructs a script starting from a generic contract.

Let $\cVarX[0]$ be the initial clause of a contract,
and let $\cVarX[1] \cdots \cVarX[n]$ be the clauses
that can be reached by recursively following every $\callname$ operation
that appears in $\cVarX[0]$'s definition.
We assume that $\cVarX[i]$ defines the process $\contrC[i]$.
We generate a script for the overall contract as follows.
The script requires that the output must be redeemed
from an input in the first position.
Then, it performs a switch on the $\txname$ argument to see which clause
is currently encoded in the transaction output,
and choose accordingly which script is going to be executed:
\begin{align*}
  \txscript \coloneqq \;
  & \inidx = 1 \andE
  \\
  & \ifE{ (\ctxo{\txname} = \cVarX[0]  ) }{\scr_{X_0}}{}
  \\[-4pt]
  & \vdots
  \\
  & \ifE{(\ctxo{\txname} = \cVarX[n] )}{\scr_{X_n}}{\false}
\end{align*}

To construct the script associated to a clause $\cVarX$
in $\cVarX[0] \cdots \cVarX[n]$,
we inspect its process $\contrC = \contrD[1]+\cdots+\contrD[m]$,
and associate an integer $n_j$ with each $\contrD[j]$ as follows:
if $\contrD[j]$ ends in a $\callname$,
then $n_j$ is equal to the number of called clauses;
otherwise, if $\contrD[j]$ ends in a $\withdrawname$,
then $n_j$ is equal to the number of participants receiving the funds.
This $n_j$ will be the number of outputs of a transaction
that redeems the $j$-th branch of $\contrC$.
This transaction must also specify the value $j$
in the $\txbranch$ argument of each of its outputs.
To check these conditions, we use:
\begin{align*}
  \mathsf{B}_j \coloneqq \;
  & (\outlen{\rtx} = n_j \andE \rtxo{\txbranch}{1} = j \andE
  \\
  & \cdots \andE \rtxo{\txbranch}{n_j} =j)
\end{align*}
Then, we handle all the branches, with a conditional
\begin{align*}
  \scr_{X} \coloneqq
  & \ \ifE{\mathsf{B}_1}{\scr_{\contrD[1]}}{}
  \\
  &\quad \vdots \qquad \quad \vdots \quad \qquad \vdots
  \\
  & \ \ifE{\mathsf{B}_m}{\scr_{\contrD[m]}}{\false}
\end{align*}

Each branch $\contrD$ in $\contrD[1] \cdots \contrD[m]$
is a sequence of decorations ended by a $\callname$ or $\withdrawname$. 
To construct $\scr_{\contrD}$, we first focus on the decorations.
If there is an authorization decoration, then the
witnesses of the redeeming transaction requires a signature
by the authorizing participant:
\begin{align*}
	\scr_{\authC{\pmvA[1]}{\decC{\cdots}{ \authC{\pmvA[k]}{\contrDi} } }} \coloneqq& \ \versig{\ctxo{\pmvA[1]} }{\rtxw.1} \andE \cdots
	\\ & \ \andE\versig{\ctxo{\pmvA[k]}}{\rtxw.k}  \andE {\scr_{\contrDi}}
\end{align*}
where $\contrDi$ does not contain any authorization decoration.
The ``after'' decorations are handled by the corresponding script operators
${\sf absAfter}$/${\sf relAfter}$ for absolute/relative timelocks:
% to force the respective timelock to be greater
% than the value specified in the decoration:
\begin{align*}
  \scr_{\afterC{t}{\contrDii}}
  & \coloneqq \afterAbs{\ctxo{t}}{\scr_{\contrDii}} 
  \\
  \scr_{\afterRelC{\delta}{\contrDii}}
  & \coloneqq \afterRel{\ctxo{\delta}}{\scr_{\contrDii}}
\end{align*}

Finally, we describe the terminal parts of the script, \ie
$\withdrawname$ and $\callname$.
First, we consider the case:
\[
\contrDii = \withdrawC{(\splitB{v_1\tokT[1]}{\pmvA[1]} \mmid \cdots \mmid  \splitB{v_n\tokT[n]}{\pmvA[n]})}
\]
Here, we want each output of the redeeming transaction
to encode a deposit of value $\valV[k]\tokT[k]$ owned by $\pmvA[k]$. 
We use the operator $\verscript{}{}$
to force the redeeming transaction to have the correct script,
and $|\cdot|$ to check that it has
exactly 3 arguments
(corresponding to $\txnonce$, $\txbranch$, and $\txowner$).
The script also checks the output values and the owners:
%% \begin{align*}
%%   \scr_{\contrDii} & \coloneqq 
%%   \\
%%   & |{\rtxo{\txarg}{1}}|= 3 \andE \cdots \andE |\rtxo{\txarg}{n}|= 3
%%   \\
%%   & \verscript{ \versig{\ctxo{\txowner}}{\rtxw.1}}{1} \andE
%%   \\
%%   & \rtxo{\txowner}{1} = \ctxo{\pmvA[1]}  \andE \rtxo{\txval}{1} = \ctxo{v_1} \tokT[1] 	
%%   \\	
%%   & \cdots
%%   \\
%%   & \verscript{ \versig{\ctxo{\txowner}}{\rtxw.1}}{n} 
%%   \\
%%   & \rtxo{\txowner}{n} = \ctxo{\pmvA[n]}  \andE \rtxo{\txval}{n} = \ctxo{v_n} \tokT[n] 
%% \end{align*}
\begin{equation*}
  \scr_{\contrDii} \coloneqq
  \left.
  \begin{aligned}
    & |{\rtxo{\txarg}{i}}|= 3 \andE
    \\
    & \verscript{ \versig{\ctxo{\txowner}}{\rtxw.1}}{i} \andE
    \\
    & \rtxo{\txowner}{i} = \ctxo{\pmvA[i]} \andE
    \\
    & \rtxo{\txval}{i} = \ctxo{v_i} \tokT[i]
    \andE \cdots
  \end{aligned}
  \right\}
  \rotatebox{90}{\hspace{-12pt}\small$i \in 1..n$}
\end{equation*}

\noindent
The last case is that for a $\callname$:
\[
\contrDii= \ncadv{ ( \clauseCallY[1]{ \stVal[1] }{?} \mmid \cdots \mmid \clauseCallY[n]{\stVal[n]}{?})}
\]
Let
$\clauseDefY[i]{\stPar[i]}{\dinPar[i]} = \clauseAdv{\vec{\sexp[i]\tokT[i]}}{p_i}{\!\contrC[i]}$,
where $|\vec{\stPar[i]}| = k_i$ and $|\vec{\dinPar[i]}|= h_i$.
The script $\scr_{\callname i}$
requires that the $i$-th output of the redeeming transaction
encodes the contract specified by
$\clauseCallY[i]{\stVal[i]}{\dinVal[i]}$,
for some choice of the parameters $\vec{\dinVal[i]}$.
% \ricnote{Attenzione ai doppi indici. Originariamente era il vettore di argomenti $\vec{a^i}$, con elementi $a^i_j$.}
We use $\verrec{}$ to preserve the contract script,
and then check that the output has the correct number of arguments.
We also require that the $\txname$ is $\cVarY[i]$,
and that the arguments $\txTag[]{}{\mathit{In}_{i j}}$
match the actual parameters.

Finally, we check the funding precondition:
\begin{align*}
  \scr_{\callname i} \coloneqq
  & \ \verrec{i} \andE |\rtxo{}{i}|= 3 + k_i + h_i \andE
  \\
  & \ \rtxo{\txname}{i} = \cVarY[i] \andE
  \\
  & \ \rtxo{\txTag[]{}{\mathit{In}_{i1}}}{i} = \ctxo{\stVal[{i \, 1}]}
  \andE \cdots \andE 
  \\
  & \ \rtxo{\txTag[]{}{\mathit{In}_{i k_i}}}{i} = \ctxo{\stVal[{i \, k_i}]} \andE
  \\
  & \ \rtxo{\txval}{i}=\rtxo{ \vec{\sexp[i] \ \tokT[i] } }{i} \andE
  \\
  & \ \rtxo{p_i}{i}   
\end{align*}
This must be done for all $\cVarY[1] \cdots \cVarY[n]$, obtaining:
\[
\scr_{\contrDii} \coloneqq \scr_{\callname 1} \andE \cdots  \andE \scr_{\callname n} 
\]

%% Note that the execution can carry on only if the conditions are verified for all of the redeeming outputs, mirroring the fact that symbolically all contracts are activated simultaneously.

\paragraph*{Executing a compiled contract}
The \langname compiler translates an \langname contract into a script.
In this way, 
the compilation process creates a correspondence between
active contracts and transaction outputs on the blockchain.
We will formalise this \emph{coherence relation} later when
establishing the security of the compiler.
For now, we just note that each execution step 
of an active \langname contract corresponds, in the blockchain,
to a new transaction redeeming the previous output. 

\paragraph*{A hint about the correctness of the compiler}

The script produced by the compiler imposes very stringent conditions
on the redeeming transaction $\txT$.
In particular, its outputs are almost completely determined
by the compiled script:
the number of outputs and their assets are fixed;
their script is determined either by $\verrec{}$ (in the $\callname$ branches)
or by $\verscript{}{}$ (in the $\withdrawname$ branches);
the number of arguments is fixed, and
the value of most of the arguments (which encode the contract state)
is determined by the script.
% and most of their arguments are set to precise values. 
The only ``free'' fields in $\txT$ are the arguments representing
the external contract parameters, which are only subject to respect
the funding precondition.
This mirrors the \langname semantics,
where participants can choose the actual external
parameters at runtime.
This \say{rigidity} is important in establishing
the correctness of the \langname compiler:
if a UTXO encodes an active contract
$\confContrTime[x]{t}{\contrC}{\vec{v\tokT}}$,
then any transaction that redeems it must behave
in ``agreement'' with one of the branches of $\contrC$.
The full details of the proof of correctness are presented in \iftoggle{arxiv}
{\Cref{sec:app-correctness-of-compiler}}
{\cite{BMZ23arxiv}}.

\section{Adversary model}
\label{sec:adversary}

The semantics of \langname describes \emph{all}
actions that can be performed on contracts and deposits. 
For this reason the set of reachable configurations is very broad.
In particular, it always contains the configuration obtained by donating all the deposits to a single participant.
However, in a realistic scenario, this configuration would not be reached because participants would have no interest in authorizing the donations. 
To avoid considering these unrealistic executions,
we need to restrict the semantics according to the behaviour of participants,
which can decide whether to authorize or not any given action.
To this aim, we introduce \emph{strategies}, which are algorithms
that model the participants' behaviour, computing the actions chosen
by a participant at each execution step.
We assume a subset $\PartT$ of participants for whom the strategies are known.
The strategies of these \emph{honest} participants
are instrumental in defining the adversary model.
Namely, we see the adversary $\Adv$ as an entity that 
controls the scheduling of the actions chosen by honest participants, and possibly inserts their own actions.
This is consistent with adversarial miners/validators in blockchains,
who can read user transactions in the mempool, and produce blocks containing some of these transactions, suitably reordered and possibly interleaved with their own transactions.
Intuitively, we model the adversary as a strategy that controls all participants outside of $\PartT$, can observe the actions outputted by the strategies of honest participants, and can decide to perform one of these actions or one of their own.

\paragraph*{Symbolic runs}

A symbolic run $\runS$ is a sequence of configurations $\confG[i]$ connected by semantic actions $\labS[i]$.
The first configuration in the sequence only contains deposits, and has time $t_0 = 0$.
We denote with $\confG[\runS]$ the last configuration of $\runS$. A run is written as $\runS = \confG[0] \xrightarrow{ \labS[0]} \confG[1] \xrightarrow{\labS[1]} \cdots$.

\paragraph*{Symbolic strategy of honest participants}

Each honest participant $\pmvA$ has a strategy $\stratS{\pmvA}$,
\ie a PPTIME algorithm that takes as input the symbolic run $\runS$ and outputs the set of \say{choices} of $\pmvA$,
\ie the \langname actions that $\pmvA$ wants to perform.
The strategy is subject to well-formedness constraints:
\begin{inlinelist}
\item the actions in $\stratS{\pmvA}(\runS)$ must be enabled by the semantics in $\confG[\runS]$; 
\item each authorization action in $\stratS{\pmvA}(\runS)$ must be of
  the form  $\confAuth{\pmvA}{\authBranch{\cdots}{\cdots}}$, forbidding $\pmvA$ to impersonate another participant.
\end{inlinelist}

\paragraph*{Symbolic adversarial strategy}

Dishonest participants are controlled by the adversary $\Adv$,
who is also in charge of scheduling updates to the run. 
Their strategy $\stratS{\Adv}$ is a PPTIME algorithm
that takes as inputs the run $\runS$ and the sets of choices
given by the honest participants' strategies.
The output of $\stratS{\Adv}$ is a single \langname action $\labAdv$
that will be used to update the run.
$\Adv$'s strategy is subject to the following constraints:
\begin{enumerate*}[$(i)$]
\item $\labAdv$ must be enabled in $\confG[\runS]$;
\item if $\labAdv$ is an authorization action by a honest $\pmvA$,
  then it must have been chosen by $\pmvA$;
\item if $\labAdv$ is a delay, then it must have been chosen by all honest participants.
\end{enumerate*}
The second condition prevents $\Adv$ from forging signatures,
while the third condition ensures that $\Adv$ cannot
prevent honest participants from meeting deadlines.

\paragraph*{Symbolic conformance}

Since a strategy $\stratS{\pmvA}$ is probabilistic,
we implicitly assume that it takes as input
an infinite sequence of random bits $\rndR_{\pmvA}$.
Consider now a set of strategies $\stratSSet$
including those of honest participants, $\stratS{\Adv}$,
and a random source $r$ from which the sequences $r_{\pmvA}$ are derived.
We can uniquely determine a run $\runS$ by
% starting from the initial configuration and applying the transition labels
performing the actions outputted by $\stratS{\Adv}$.
Such a run is said to \emph{conform} to $(\stratSSet, r)$.

\paragraph*{Computational runs}

Above, we have defined an adversarial model at the symbolic level.
We model adversaries at the computational level in  a similar way,
replacing symbolic actions with computational ones.
% we obtain a notion of computational runs and strategies.
A computational run $\runC$ is a sequence of actions $\labC$
in one of these forms:
\begin{align*}
  & \txT
  && \text{appending transaction $\txT$ to the blockchain}
  \\
  & \delta
  && \text{performing a delay}
  \\
  & \pmvA \rightarrow \ast:m
  && \text{broadcasting of message $m$ from $\pmvA$}
\end{align*}

The first action in the run $\runC$ is an initial transaction
that distributes tokens to participants, and it is followed by the broadcast of each participant's public keys.

\paragraph*{Computational strategies of honest participants}

Each honest $\pmvA$ is associated with
a computational strategy, \ie a PPTIME algorithm $\stratC{\pmvA}$
that takes as input a computational run and outputs
the set of choices of $\pmvA$.
If $\stratC{\pmvA}(\runC)$ includes an action $\labC = \txT$,
then $\txT$ must be \emph{consistent} with $\runC$,
essentially meaning that $\txT$ is a valid transaction in the blockchain state
reached after the run $\runC$.

\paragraph*{Computational adversarial strategy}

Like in the symbolic case, the adversary is given scheduling power.
The strategy $\stratC{\Adv}$ takes as input
the run $\runC$ and the actions chosen by honest participants,
and outputs a single action that will be used to update $\runC$.
As for honest participants, the adversary cannot output invalid transactions.
Like in the symbolic case, $\stratC{\Adv}$ is only allowed to output a delay
if it has been chosen by all honest participants.
We allow the adversary to impersonate any honest participants $\pmvA$.
However, since $\Adv$ does not know the random source $r_{\pmvA}$, and $\stratC{\Adv}$ is PPTIME, $\Adv$ will be, with overwhelming probability,
unable to forge $\pmvA$'s signatures.

\paragraph*{Computational conformance}

Like in the symbolic case, a set of strategies $\stratCSet$ and a random source $r$ can be used to uniquely determine a computational run $\runC$,
that is said to conform to the pair $(\stratCSet,r)$.

\section{Security of the \langname compiler}
\label{sec:soundness}

Symbolic and computational runs describe the evolution of contracts
at two different level of abstraction: in \Cref{sec:compiler}
we have shown how transaction outputs encode \langname deposits and contracts.
Formally this correspondence between runs is modelled as a relation, which we call \emph{coherence}. Intuitively, coherence holds when the symbolic steps in $\runS$ and the computational steps in $\runC$ have the same effects on contracts and deposits.

\paragraph*{Coherence}

The coherence relation \mbox{$\coherRel{\runS}{\runC}{\txMapC}$}
is parameterized by a map $\txMapC$ that relates the symbolic names of deposits and contracts to transaction outputs.
% and we write it as .
Coherence is defined inductively, by exhaustively listing the possible actions of $\runS$. 
Here we present the most important cases, relegating
the full definition to \iftoggle{arxiv}{Appendix~\ref{sec:app-coherence}}{\cite{BMZ23arxiv}}.
\begin{itemize}
  
\item Advertising $\confF$ in $\runS$ is matched by the broadcast of a message $m$ in $\runC$. The message $m$ encodes a transaction $\txT[\confF]$ representing the action advertised by $\confF$. In particular, the script of $\txT[\confF]$'s outputs must be the one produced by the compiler.
Note that $\txT[\confF]$ is not yet appended to the blockchain, but only broadcast.

\item Sending an authorization $\confAuth{\pmvA}{\authBranch{\cdots}{\confF}}$ in $\runS$ is matched by sending a message $m$ in $\runC$, containing a corresponding signature from $\pmvA$ on $\txT[\confF]$.

\item Initiating a contract in $\runC$ consumes the respective advertisement $\confF$ and the required authorizations and deposits to insert a term $\confContrTime[x]{t}{\contrC}{\vec{v\tokT}}$ in the configuration. 
  This is matched in $\runC$ by appending $\txT[\confF]$ to the blockchain.
  This uses the signatures that were broadcast together with the symbolic authorizations.
  Moreover, the map $\txMapC$ is updated so that the new symbolic name $x$ is mapped to $\txT[\confF]$'s output.

\item Continuing a contract in $\runC$ is similar to the contract initiation described above, except that it may produce multiple deposits or contracts instead of a single one. 
  Again, it is matched in $\runC$ by $\txT[\confF]$, where $\confF$ is the continuation advertisement, and $\txMapC$ is updated to map the new names to each output of $\txT[\confF]$. 

\end{itemize}

The full definition also deals with deposit operations, delays, and transactions that spend inputs that are outside of the image of $\txMapC$.
The coherence relation is instrumental to establish correspondence results
between the two models. 
Notably, \Cref{lem:balance-preservation} shows that the coherence relation precisely characterizes the exchange of assets.

\begin{lem}
  \label{lem:balance-preservation}
  If $\coherRel{\runS}{\runC}{\txMapC}$ and $\confDep[x]{\pmvA}{v\tokT}$
  is a deposit appearing in the last configuration $\confG[\runS]$,
  then $\txMapC(x)$ is an unspent output in $\runC$ and
  encodes the deposit $x$ (\ie it has the structure presented
  in~\Cref{sec:low-level-model}).
  Notably, this means that the token balance is preserved by $\txMapC$.
\end{lem}

This coherence result is lifted to an analogous lemma for contracts. 
We also establish the injectivity of the map $\txMapC$, which ensures that no two distinct symbolic deposits or contracts are represented by the same transaction output. This means that the whole volume of assets is preserved by the map.
Moreover, the coherence relation can be used as a guide to algorithmically translate the symbolic strategy of an honest participant into an equivalent computational strategy.  

\paragraph*{From symbolic to computational strategies}
Here we present the map $\stratMap$ that translate strategies.
Given the symbolic strategy $\stratS{\pmvA}$, the computational strategy $\stratC{\pmvA} = \stratMap(\stratS{\pmvA})$ will do the following:
first parse $\runC$ to create a corresponding symbolic run $\runS$ (using the coherence relation), then run $\stratS{\pmvA}(\runS)$ producing a set of symbolic labels $\LabS$, and lastly use the coherence again to transform each symbolic label into the corresponding computational label, which will be the output of the strategy.

\paragraph*{Security of the compiler}

\Cref{th:computational-soundness} gives us a way to construct a symbolic run $\runS$ that is coherent to a given computational run $\runC$ and conform to a set of given honest symbolic strategies $\stratSSet$.  This is done under the assumption that $\runC$ conforms to the honest computational strategies obtained by translating $\stratSSet$.
By contrast, we make no assumption on the computational adversarial strategy used to construct $\runC$.
Together with \Cref{lem:balance-preservation}, 
% the security theorem 
computational exchanges of assets, including those mediated by contracts, are mirrored at the symbolic level.

\begin{thm}[Security of the compiler]
  \label{th:computational-soundness}
  Let $\stratSSet$ be a set of symbolic strategies for honest participants,
  let $\stratCSet = \stratMap(\stratSSet)$,
  and let $\stratC{\Adv}$ be a computational adversarial strategy.
  If $\runC$ is a run with polynomial length
  conforming to $\stratCSet \cup \setenum{\stratC{\Adv}}$,
  then there exist, with overwhelming probability,
  a symbolic run $\runS$ and an adversarial strategy $\stratS{\Adv}$
  such that
  \begin{inlinelist}
  \item $\runS$ is coherent with $\runC$, and
  \item $\runS$ conforms to $\stratSSet \cup \setenum{\stratS{\Adv}}$.
  \end{inlinelist}
\end{thm}
\begin{proof}[Proof (sketch)]
  We match step-by-step the computational moves with the symbolic moves according to the coherence relation.
  In particular, looking at possible transactions, we have two main cases:
  \begin{itemize}
  \item A deposit operation (e.g. donating a deposit). This requires participant signatures. If in the symbolic run, there are the corresponding authorizations, then this operation has an immediate symbolic counterpart.
    Otherwise, the computational signatures have been forged, which happens with negligible probability.
  \item A contract operation (\eg, a call to a new clause).
    This can be done only with a transaction that satisfies the contract script. Since the script closely matches the symbolic semantics, we can construct the corresponding symbolic move
    (again, the only case where this is not possible is that of a signature forgery).
  \end{itemize}
  The full proof in
  \iftoggle{arxiv}
  {Appendix~\ref{sec:app-computational-soundness}}
  {\cite{BMZ23arxiv}}
  considers all the possible computational moves
  (\eg, outputting a message, waiting),
  and relates them to a specific symbolic action
  that maintains the coherence relation.
  Technically, this requires examining all the \emph{twenty} cases in the definition of coherence and proving that whenever 
  none of them applies, the adversary must have managed to forge signatures.
\end{proof}

\change{%
The above security result can be seen in terms of Robust Trace Property Preservation (RTP)~\cite{Patrignani19csur,Abate19csfw}. RTP can be equivalently~\cite{Abate19csfw} stated in this form:
\[
\forall {\sf \color{blue} P}. \forall {\bf \color{orange} C_T}. \forall t.\ 
{\bf\color{orange} C_T[}{\sf compile}({\sf \color{blue} P}) {\bf \color{orange} ]} \rightsquigarrow t 
\implies \exists {\sf \color{blue} C_S}. \ {\sf \color{blue} C_S [P]} \rightsquigarrow t
\]
This can be read in our setting as ``whenever a \illum contract $\sf\color{blue} P$ is compiled and run in a computational adversarial context $\bf\color{orange} C_T$, producing an execution trace $t$, then there exists a symbolic adversarial context $\sf\color{blue} C_S$ where the original contract $\sf\color{blue} P$ produces the same trace $t$''.

The above property can not be proved as-is in our setting, for a number of reasons. First, computational and symbolic traces have different nature, so we can not claim to have the same trace in both worlds -- we instead claim that we have two traces $\runC, \runS$ which are related by the coherence relation. Furthermore, computational adversaries always have a negligible probability to break the underlying cryptography, so RTP can only hold with overwhelming probability and for traces having polynomial length.
The statement of~\Cref{th:computational-soundness} accounts for these peculiarities. Finally, in our formulation, the adversarial contexts are interpreted as (symbolic/computational) adversarial strategies.
}% change
% \section{Evaluating \langname on complex use cases}
\section{From high-level languages to \illum}
\label{sec:conclusions}
\label{sec:discussion}
\label{sec:hellum}

% (1) Elaborate more on how feasible it is to translate from an actual high-level language to ILLUM. Add an example written in an existing high-level language like Cardona and discuss how it can be translated to ILLUM.
% (2) A prototype implementation and evaluation to show the practicality of the approach.

Although \illum provides an abstraction layer over the UTXO transaction model, its clause-based nature may make it unwieldy for developers familiar with the procedural style, which is currently mainstream in the smart contracts community thanks to languages like Solidity.
We show in this section that it is possible to reconcile the UTXO model with the familiar high-level imperative procedural style. More specifically, we consider an expressive fragment of Solidity, and we show how to compile it down to \illum. We evaluate our approach by developing a prototype compiler and interpreter for the high-level language ($\sim$2000 LoC of OCaml code), and by applying it to a benchmark of common smart contracts, including complex DeFi protocols like Automated Market Makers and Lending Pools.
Overall, one can benefit from the formal security guarantees of \illum, while sticking to a familiar development process.

\paragraph{The \hellum contract language}

As a high-level language for contracts in the UTXO model,
we consider a fragment of Solidity, a widespread smart contract language that has been popularized by Ethereum.
To make the compilation into UTXO possible, we get rid of a couple of problematic features, \ie loops and external contract calls. To compensate for the absence of external calls, which are the basis to implement custom tokens in Solidity, our language supports custom tokens natively.
% next-function modifiers

The resulting High-Level Language for the UTXO Model, hereafter dubbed \hellum, is exemplified in \Cref{fig:hll:crowdfunding} through a crowdfunding contract.
The contract collects funds from donors until a \code{deadline}, then it transfers them to the \code{owner} only if the donations reach a given \code{target} amount. If the target is not met, then every donor is entitled to take back their donations.
The constructor sets the contract parameters. The \hllinline{next} modifier constrains which functions can be called next.
The \txcode{deposit} function receives donations, and updates the map \code{funds} accordingly. The modifier \hllinline{input(x:T)} 
requires the caller to pay an amount \code{x} of tokens \code{T} upon a call. The \hllinline{require} command sets the minimum donation to $10$ token units.
The \txcode{finalize} function can only be called after the deadline is reached. If the collected funds (\ie the contract balance) exceed the target, then they are transferred to the owner: otherwise the funds are kept in the contract.
Executing \txcode{finalize} enables the \txcode{withdraw} function, through which donors can take back their donations if the target has not been reached (note that if the target was met, then \txcode{withdraw} transfers no funds).
The modifier \hllinline{auth(a)} requires that any withdraw to  \code{a} must be authorized by the user controlling that address (\ie, the one who knows \code{a}'s private key).

\begin{figure}
\begin{lstlisting}[language=hellum,morekeywords={Crowdfund,init,withdraw,deposit,finalize},frame=single]
contract Crowdfund {
    mapping (address => uint) funds;
    uint deadline;
    uint target;
    address owner;

    constructor(uint d, uint t, address o) {
        owner = o;
        deadline = d; 
        target = t;
    } next(deposit,finalize)

    function deposit(uint x, address a) input(x:T) {
        require(x>=10);
        funds[a] += x;
    } next(deposit,finalize)

    function finalize() after(deadline) {
        if (balance(T) >= target)
            owner.transfer(balance(T):T);
    } next(withdraw)

    function withdraw(address a) auth(a) {
        a.transfer(funds[a]:T);
        funds[a] = 0;
    } next(withdraw)
}
\end{lstlisting}
  \negcaptionspace
  \caption{A crowdfund contract in \hellum.}
  \label{fig:hll:crowdfunding}
\end{figure}

We argue that this variant of Solidity is still practical for a wide range of applications (see~\Cref{tab:hellum:benchmark}).
Regarding loops, we note that in general they are discouraged even in Solidity, since they may vehicle gas exhaustion attacks~\cite{solidity-gas-exhaustion} where an adversary causes an iteration to exceed the block gas limit, thereby making the users pay the gas fees for failed transactions, and, possibly, making the contract stuck~\cite{ABC17post}. Despite this limitation, our language features unbounded data structures, in the form of key/value mappings. Iterative behaviours can be obtained by shifting the duty of performing iterations to users, by requiring them to perform repeated calls to contract functions (see \eg the withdrawal of funds in the crowdfunding contract).
Regarding external calls, while in Solidity they are the basis for any interaction between a contract and the environment (including  pure transfers of assets), in the UTXO model they are unnatural, since transaction validation must only involve the scripts referred to in the transaction inputs. Cardano, the main  smart contract platform based on the UTXO model, does not support external calls. In our high-level language we use a special primitive \hllinline{transfer} to exchange tokens, and we restrict calls to internal pure functions. 

\paragraph{Compiling \hellum to \illum}

\begin{figure}
\begin{lstlisting}[language=hellum,morekeywords={f,g},frame=single]
function f(...) {
    require expr0; // this is the only require in f 
    // chain of conditional statements
    if (cond1) {
        // sequence of token transfers
        a1_1.transfer(amt1_1:T1_1);
        ...
        a1_n1.transfer(amt1_n1:T1_n1);
        // single simultaneous assignment
        x_1,...,x_m = expr1_1,...,expr1_m;
    } else if (cond2) {
        a2_1.transfer(amt2_1:T2_1);
        ...
        a2_n2.transfer(amt2_n2:T2_n1);
        x_1,...,x_m = expr2_1,...,expr2_m;
    }
    ...
    else {
        ... // same structure as previous blocks
    }
}
\end{lstlisting}
  \negcaptionspace
  \caption{Normal form of \hellum functions.}
  \label{fig:hll:nf}
\end{figure}

\hellum contracts can be automatically compiled to \illum. Here we summarize the translation process (see Appendix~\ref{sec:hllc} for more details, and \github for the implementation).
We use the \contract{Test} contract in~\Cref{fig:hll:test7} as a working example to illustrate the compilation process.

First, we process each function in the \hellum contract, passing it through code transformations which bring it to the normal form displayed in~\Cref{fig:hll:nf}. More specifically, a function is in normal form when:
\begin{itemize}
    \item expressions do not contain internal calls to pure functions (these calls are macro-expanded); 
    % into the corresponding expressions);
    \item the function starts with a single \hllinline{require} statement, which is the only one appearing in the function body;
    \item after the \hllinline{require}, the rest of the function body is a chain of conditional statements;
    \item each conditional block starts with a sequence of \hllinline{transfer} statements, followed by a single simultaneous assignment of all of the contract variables. This assignment also defines auxiliary variables representing the
    new contract balance after the transfers.
\end{itemize}

For example, the normal form obtained for the \contract{Test} contract is displayed in~\Cref{fig:hll:test7-nf}. In the transformed contract, we use the expression \hllinline{balance_pre(T)} to denote the contract balance of token \hllinline{T} \emph{before} the function call, and the auxiliary variable \hllinline{bal_T_fin} to denote the balance of \hllinline{T} \emph{after} the call.
When in normal form, functions are amenable to be translated into \illum clauses, since the simultaneous assignments effectively specify the new contract state as a function of the old state.

\begin{figure}
\begin{lstlisting}[language=hellum,morekeywords={Test,f,g},frame=single]
contract Test {
    uint x; 
  
    function f(uint y, address a) {
        x = balance(T) - y;
        if (x > 10) {
            a.transfer(x:T);
            require balance(T) > 20;
            x += 1;
        }
    } next(f,g)

    function g() { }
}
\end{lstlisting}
  \negcaptionspace
  \caption{A simple contract in \hellum.}
  \label{fig:hll:test7}
\end{figure}

\begin{figure}
\begin{lstlisting}[language=hellum,morekeywords={Test_NF,f,g},frame=single]
contract Test_NF {                    
  uint x;

  function f(uint y,address a) {
    require ((balance_pre(T)-y>10) && y>20) || 
                     (balance_pre(T)-y<=10);
    if (balance(T)-y>10) {
      a.transfer(balance_pre(T)-y:T);
      x,bal_T_fin = (balance_pre(T)-y)+1,
        balance_pre(T)-(balance_pre(T)-y);
    }
    else {
      x,bal_T_fin = balance_pre(T)-y,balance_pre(T);
    }
  } next(f,g)

  function g() {  
      x,bal_T_fin = x,balance_pre(T);
  }
}
\end{lstlisting}
  \negcaptionspace
  \caption{Normal form of the \contract{Test} contract.}
  \label{fig:hll:test7-nf}
\end{figure}

The \hellum compiler transforms each function \hllinline{f} into two \illum clauses, called \illinline{f_run} and \illinline{f_next}. The clause \illinline{f_run} is used to take the parameters of \hllinline{f} and run the function body. It has one branch for each conditional branch of \hllinline{f}: these branches are enabled by the same conditional guards, and perform the payments alongside with calling \illinline{f_next} with the updated state passes as parameter. 
The funding precondition of \illinline{f_run} is computed taking into account the \hllinline{input} modifiers, as well as the expression enclosed in the \hllinline{require} statement.

The clause \illinline{f_next} allows the execution to continue by calling one of the contract functions, as constrained by the \hllinline{next} modifier in the \hellum function \illinline{f}. To this purpose, \illinline{f_next} has one branch for each of the possible continuation functions. The branches of \illinline{f_next} use the \illum decorators to implement the behaviour of the \hllinline{auth} and \hllinline{after} modifiers of the called \hellum function.

The output of the compiler on the \contract{Test} contract is displayed in~\Cref{fig:ill:test7}, where we use the concrete \mbox{\illum} syntax supported by the compiler. There, we can observe how the clause \illinline{f_run} contains a process with two branches, one for each conditional branch in~\Cref{fig:hll:test7-nf}. Both of these branches call the \illinline{Check} clause so to enable the whole \hllinline{call} if and only if the corresponding conditional branch in \hellum would be taken. The clause \illinline{Check} requires in its precondition that its argument is true, so blocking the \illinline{f_run} branches which do not correspond to the \hellum execution.
The \illum branches call clause \illinline{Pay} to transfer the assets according to the \hllinline{a.transfer(...)} commands found in the corresponding conditional branch of \hllinline{f}
in~\Cref{fig:hll:test7-nf}.
Finally, each branch calls \illinline{f_next} passing the updated state in the parameters.

The correctness of the compilation from \hellum to \illum is straightforward.
First, the code transformations used to bring the \hellum contract in normal form, detailed in Appendix~\ref{sec:hllc}, are standard and clearly preserve the semantics of contracts.
Second, the \illum contract clauses are generated precisely following the simple structure of the obtained normal form, so their semantics is faithful to the original code by construction. Indeed, we perform the same conditional checks in \illum, we transfer the same tokens, and we we update the state variables in the same way it is done by the simultaneous assignment of the \hellum normal form.

\begin{comment}
% While our compiler uses two separate \langname clauses for each \hellum function, it would be possible in some cases to optimize the generated \langname code so to use a single clause.
% For instance, this could be performed when the parameters of the \hellum function do not affect \hllinline{transfer} commands nor \hllinline{after}
% Note that our \hellum compiler has room for optimization: indeed certain methods could be implemented with a single clause call instead of with two: for instance this is the case when payments performed by the function only depend on the contract state and not on function parameters.\ricnote{Alternativa: note that we need two clauses only when the runtime parameters taken by \hllinline{f} are used by its \hllinline{auth} or \hllinline{after} modifiers: in other cases the \hellum compiler could be optimized to only create one \illum clause.} For the sake of simplicity, our prototype compiler only follows the general case.
\end{comment}

\begin{figure}
  \vspace{-5pt}
\begin{lstlisting}[language=illum,morekeywords={},frame=single]
clause f_run(x,bal_T; y,a) {        
  precond_wallet: bal_T:T
  precond_if: ((bal_T-y>10) && y>20) || 
               (bal_T-y<=10)
  process: 
    call( Check(bal_T-y>10) | Pay(a,bal_T-y,T) | f_next((bal_T-y)+1,y) )
    call( Check(bal_T-y<=10) | f_next(bal_T-y,bal_T) )
}
clause f_next(x,bal_T; ) {
  precond_wallet: bal_T:T
  precond_if: true
  process: 
    call f(x,bal_T)
    call g(x,bal_T)
}
clause g_run(x,bal_T; ) {
  precond_wallet: bal_T:T
  precond_if: true
  process: 
    call g_next(x,bal_T)
}
clause g_next(x,bal_T; ) {
  precond_wallet: bal_T:T
  precond_if: true
  process: 
    call f(x,bal_T)
    call g(x,bal_T)
}
clause Pay(a,v,t; ) {
  precond_wallet: v:t
  precond_if: true
  process: 
    send(v:t -> a)
}
clause Check(b; ) {
  precond_wallet: 
  precond_if: b
  process: 
    send()
}
\end{lstlisting}
  \negcaptionspace
  \caption{Translation of the \contract{Test} contract in \illum.}
  \label{fig:ill:test7}
\end{figure}

\paragraph{Evaluation}

\begin{table}
\small
\centering
\begin{tabular}{|c|cc|cc|}
\hline
\multirow{2}{*}{\textbf{Contract}} & 
\multicolumn{2}{c|}{\hellum} & \multicolumn{2}{c|}{\illum} \\
& \textbf{LoC} & \textbf{B} & \textbf{LoC} & \textbf{B} \\
\hline
Crowdfund & 29 & 949 & 64 & 2150 \\
Auction & 31 & 772 & 52 & 1577 \\
Payment splitter & 37 & 1030 & 77 & 3183 \\
Vault & 39 & 984 & 90 & 4070 \\
Automated Market Maker & 40 & 1213 & 88 & 3642 \\
Voting & 42 & 1296 & 91 & 4519 \\
Vesting wallet & 44 & 1194 & 69 & 3360 \\
Escrow & 45 & 1359 & 99 & 3602 \\
King of the hill & 50 & 2062 & 69 & 4509 \\
Blind auction & 57 & 1742 & 86 & 5619 \\
Lending pool & 75 & 2062 & 132 & 6581 \\
Lottery & 78 & 2297 & 136 & 6401 \\
\hline
\end{tabular}
\caption{Benchmark of smart contracts in \hellum, displaying lines of code (LoC) and size in bytes (B).}
\label{tab:hellum:benchmark}
\end{table}

To evaluate the practicality of \illum as a compilation target
of higher-level contract languages, we construct a benchmark of  smart contracts.
The benchmark comprises common use cases, like \eg those in the \href{https://docs.openzeppelin.com/contracts}{OpenZeppelin library} of Solidity contracts.
Besides that, we also include more complex contracts like those found in DeFi: in particular, we implement a constant-product Automated Market Maker (AMM) and a Lending Pool.
All the contracts in our benchmark are implemented in \hellum, and automatically translated into \illum by our prototype compiler. 
\Cref{tab:hellum:benchmark} shows the size (LoC and bytes) of the \hellum contracts and of the corresponding \illum clauses. 
Despite the compilation into \illum produces just a 2x-3x expansion in the size, the \illum code is inherently less readable than the original \hellum contract, as usual for intermediate-level languages.

\section{Related work}
\label{sec:related}

\newcommand{\Scilla}{\textsc{Scilla}\xspace}

Intermediate languages have already been studied
that, like our \langname,
can serve as a compilation target of high-level smart contract languages.
\Scilla~\cite{Sergey19pacmpl} is an intermediate language
that targets \emph{account-based}
blockchains and is executed natively by the Zilliqa blockchain.
\Scilla has an imperative core featuring loop-free statements
(with operators to update state variables and transfer assets),
and a higher-order functional core 
with structural recursion on lists and naturals.
This gives a form of iteration,
and consequently requires the underlying blockchain to implement
a gas mechanism to thwart denial-of-service attacks.
Instead, in \langname every operation has a \emph{bounded} computational cost,
thus eliminating the need for a gas mechanism.
Nonetheless, \langname achieves Turing-completeness
by spreading complex computations across multiple basic actions.
The goals of \Scilla and \langname are different:
\Scilla is meant to be directly interpreted by blockchain nodes, while \langname is meant to be compiled to a lower-level script language,
demanding for a weaker runtime support from the underlying blockchain. 

In the UTXO realm, a variety of contract languages have been proposed,
starting from Bitcoin Script~\cite{bitcoinscript}, 
a low-level, stack-based language that is interpreted by Bitcoin nodes.
Since writing spending conditions directly in Bitcoin Script
can be quite complex, a few languages have been proposed
to relieve programmers from this task,
like Simplicity~\cite{Oconnor17plas} and Miniscript~\cite{miniscript}.
Although these languages allow for representing Bitcoin scripts
in a more structured and human-readable manner,
they do not make writing contracts in Bitcoin much easier
(except for basic single-transaction use cases).
In general, Bitcoin contracts take the form of protocols where participants
exchange messages and send transactions to the blockchain~\cite{bitcoinsok}.
The languages~\cite{Oconnor17plas,miniscript} however
can only specify the individual transactions used in these protocols,
and not the overall global contract.
BitML~\cite{BZ18bitml} is a higher-level language that allows to
specify global contracts and compile them to \emph{sets}
of Bitcoin transactions.
To be compliant with the strict constraints of Bitcoin,
the expressive power of BitML is limited to contracts with
\emph{bounded} execution lengths.
This rules out relevant use cases, like \eg the auction
in~\Cref{sec:overview} and the crowdfunding in~\Cref{sec:hellum},
which allow for an unbounded number of steps.
The work~\cite{BLMZ22lmcs} enhances the expressiveness of BitML
with a weak form of recursion: each recursive step
can only be performed  with the approval of \emph{all} participants.
In \langname instead recursion is unconstrained:
participants cannot prevent an enabled recursion step from happening.
This expressiveness gain comes at a cost, in that \langname cannot be compiled
into standard Bitcoin transactions.
Executing \langname on Bitcoin would be possible
by extending Bitcoin Script with covenants,
in a form that is just a bit more expressive than a recently proposed covenant opcode~\cite{BIP119}.

To overcome the expressiveness limitations of the Bitcoin UTXO model,
the Cardano blockchain extends it with
some additional functionalities~\cite{Chakravarty20wtsc,cardanoeutxo}:
\begin{inlinelist}

\item \label{eutxo:datum}%
  special transaction fields to store contract state;
  
\item \label{eutxo:covenants}%
  a mechanism to preserve contract code along chains of transactions;
  
\item \label{eutxo:tokens}%
  native custom tokens~\cite{Chakravarty20utxoma};

\item \label{eutxo:plutuscore}%
  an expressive scripting language~\cite{plutus-core}.

\end{inlinelist}
The first three functionalities are present also in our UTXO model:
in particular, we use $\txarg$ fields to encode the contract state,
and covenants to preserve the contract code.
The main difference between our UTXO model and Cardano's is
the scripting language.
Cardano's scripting language is an untyped lambda calculus enriched
with built-in functions to interact with the blockchain.
This makes Cardano scripts Turing-complete, and consequently requires
a complex runtime environment (including a gas mechanism).
Our scripts instead are not Turing-complete, but still our contracts are such,
as shown in~\Cref{sec:illum}.
\change{Existing smart contract languages for Cardano (\eg, Plutus, Aiken), although based on high-level languages (\ie, Haskell), impose a \emph{low-level} programming style for smart contracts, requiring developers to reason at the level of transactions, not too distantly from the awkward UTXO programming style exemplified in~\Cref{sec:overview:utxo}. Programming in this style is inherently more complex than using higher-level procedural languages, which are mainstream in the blockchain developers community.
Indeed, in existing Cardano languages, performing a contract action amounts to replacing the old state with a new one 
(\ie, spending some transaction outputs with a new transaction). Accordingly, programming a contract action amounts to verifying through the redeem scripts that the new state is a correct update of the old one, checking multiple transaction fields that encode the contract state. This programming style is quite burdensome, since forgetting even a single check may give rise to vulnerabilities (\eg, adversaries could be able to set a data field of the new state to a value at their choice).
To the best of our knowledge, we are the first to propose a practical procedural high-level language for smart contracts that can be automatically compiled to UTXO blockchains.}

Our UTXO model can be implemented efficiently.
Most operators of our scripting language are borrowed from Bitcoin Script,
which is interpreted very efficiently by Bitcoin nodes.
Implementing $\txarg{}$ fields and the opcodes to access them
poses no challenge.
Covenants, both of kind $\verscript{}{}$ and $\verrec{}$,
can be implemented by exploiting a mechanism similar to
``Pay to Script Hash'' in Bitcoin~\cite{BIP16},
which stores in the $\txscript{}{}$ field the hash of the script,
instead of the script itself.
For the $\verscript{}{}$ covenant,
we would specify the hash $h$ of the script (rather than the script)
in the first argument:
then, $\verscript{h}{i}$ would simply check that the hash
in $\rtxo{}{i}.\txscript{}$ is equal to $h$.
Similarly, the $\verrec{i}$ covenant would check that
the hash in $\rtxo{}{i}.\txscript{}$ is equal to the hash in
the current script, $\ctxo{}.\txscript{}$.
Both checks can be done very efficiently,
as one just needs to compare two hashes.
A further optimization can be achieved by exploiting Taproot~\cite{BIP341}, a mechanism allowing users to reveal the parts of the contract (clause branches) only when they are executed.
This decreases the size of witnesses that must be included
along with transactions,
which in turn decreases the transaction fees.

One of the main advantages of UTXO blockchains over account-based ones
is the possibility of parallelizing transaction validation over multiple cores.
Indeed, there is an easy criterion to determine
if two UTXO transactions are parallelizable,
\ie checking that their inputs are disjoint.
Instead, in account-based blockchains two transactions,
even targeting different contracts, may read/write the same part of the state,
\eg when they update the same account.
A few works study how to overcome this limitation:
% experimenting with the parallelization of Ethereum transactions:
some of them exploit dynamic techniques adopted from software
transactional memory
\cite{Dickerson17podc,Dickerson18eatcs,Anjana19pdp,Saraph19tokenomics},
while some others are based on the static analysis of contracts
\cite{BGM21lmcs,Pirlea21pldi}.
In particular, \cite{Dickerson17podc} provides empirical evidence
about the effectiveness of parallelizing transaction execution in Ethereum,
showing an overall speedup of 1.33x for miners and 1.69x for validators,
using only three cores,
based on a benchmark of representative contracts.

\iftoggle{anonymous}{}{\paragraph*{Acknowledgments}

Work partially supported by projects PRIN 2022 DeLiCE (F53D23009130001) and SERICS (PE00000014) under the MUR National Recovery and Resilience Plan funded by the European Union -- NextGenerationEU.}

\bibliographystyle{IEEEtran}
\bibliography{main}

\clearpage
\appendices
\section{Symbolic model of \langname contracts}
\label{sec:app-symbolic}

\newcommand{\irulespacemid}{20pt}
\newcommand{\irulespacepre}{5pt}

In this appendix we fully define the symbolic model. We start with  the syntax of contracts, clauses and configurations. We will then define the semantics of \illum as a state transition system in \Cref{fig:semantics-contracts}.
% , and we will conclude with the definition of symbolic strategies and the adversary model.

\paragraph*{Notation}

To improve readability, in the appendices we slightly simplify the model presented in the main text. 
First, we assume a single type of token. From a technical standpoint, handling multiple tokens would just require to change the semantics so that sums of values become sums of sequences of tokens. We prefer to omit this, as it would bloat an already heavy notation. 
We also simplify the arithmetic of the blockchain. In particular, we assume integers to be the only numerical data type. This restricts the arithmetic operations that are possible in contracts. Again, having rationals and divisions can be done without changing the fundamental theory developed in these appendices. We also omit mappings, since they are not strictly needed in the definition of the compiler, and could be easily added to the model.
Lastly, we adopt a different notation in the naming of the internal and external parameters of a clause: instead of the $\stPar$ and $\dinPar$, hereafter we use $\alpha$ and $\beta$. Actual values passed to clauses are also changed from $\stVal$ to $a$ and from $\dinVal$ to~$b$.

\renewcommand{\stPar}[1][]{\alpha_{#1}}
\renewcommand{\stVal}[1][]{a_{#1}}
\renewcommand{\dinPar}[1][]{\beta_{#1}}
\renewcommand{\dinVal}[1][]{b_{#1}}

\paragraph*{Syntax of expressions}
Before introducing the terms of \langname's symbolic model,
we define the syntax of expressions. 
First we have arithmetic expressions $\sexp$, defined as:
\begin{align*}
  \sexp \bnfdef
  & \phantom{\bnfmid} k && \text{(constants)}
  \\
  & \bnfmid \alpha,\beta && \text{(variables)}
  \\
  & \bnfmid |\sexp| && \text{(size)}
  \\
  & \bnfmid \sexp + \sexp 
  \\
  & \bnfmid \sexp - \sexp
  \\
  & \bnfmid H(\sexp) && \text{(hash)}
  \\
  & \bnfmid \ifE{p}{\sexp}{\sexp}
\end{align*}
Then there are name expressions $\nexp$, defined as:
\[
		\nexp \bnfdef \pmvA \text{ (names)} \bnfmid \alpha,\beta \text{ (variables)}
\]
We also define boolean expressions, or conditions as
\[
	p \bnfdef true \bnfmid \notE{p}  \bnfmid p \andE p \bnfmid \sexp= \sexp  \bnfmid \nexp = \nexp \bnfmid \sexp < \sexp
\]
In the following, we will also freely use other boolean operations that can be derived from the ones listed above.

\begin{defn}[Contracts]
  The syntax of contracts is:
  \begin{align*}
    \contrC & \bnfdef \textstyle \sum_{i \in I} \contrD[i]
    && \text{contract}
    \\[4pt]
    \contrD	& \bnfdef
    && \text{contract branch} 
    \\
    & \ncadv{( \cdots , \clauseCallX[i]{\mathcal{P}_i}{?},  \cdots )}%\clauseCallX[n]{\sexp[n]}{\nexp[n]}{\cdot}{\pmv{\cdot}} ) }
    && \text{call to clauses $\cVar[1]{X}\cdots\cVar[n]{X}$}
    \\
    \bnfmid & \withdrawC{(\cdots, \sexp[i] \rightarrow \nexp[i] \cdots )}%\sexp[n] \rightarrow \nexp[n])}
    && \text{transfer $\sexp[i]$ to each $\nexp[i]$}
    \\
    \bnfmid	& \authC{\nexp}{\contrD}
    && \text{wait for $\nexp$ authorization}
    \\
    \bnfmid	& \afterC{\sexp}{\contrD}
    &&  \text{wait until time $\sexp$}
    \\
    \bnfmid	& \afterRelC{\sexp}{\contrD}
    && \text{wait $\sexp$ after activation} 
  \end{align*}
  where $\vec{\mathcal{P}_i}$ is a sequence of arithmetic expressions $\sexp$ and name expressions $\nexp$. We also assume that: 
  \begin{enumerate}
  \item[$(i)$] each recursion variable has a unique defining equation $\clauseDefXParam= \clauseAdv{\sexp}{p}{\contrC} $, with the syntax below;
  \item[$(ii)$] the sequence of expressions $\vec{\mathcal{P}_i}$ passed to a called clause $\cVarX$ matches, in length and typing, the sequence of formal parameter $\vec{\alpha}$ of formal parameters;
  \item[$(iii)$] the order of decorations is immaterial, for instance $\authC{A}{\afterC{t}{\contrD}}$ is identified with $\afterC{t}{\authC{A}{\contrD}}$.
  \end{enumerate}
\end{defn}

\begin{defn}[Clauses]
  A clause is defined by an equation
  \[
  \clauseDefXParam = \clauseAdv{\sexp}{p}{\contrC}.
  \]
  where $\clauseFunding{\sexp}{p}$ is the \emph{funding precondition} and $\contrC$ is a contract.
  The clause takes two sequences of parameters $\vec{\stPar}$ and $\vec{\dinPar}$. Parameters are of two types: integers and participants, and we will assume that in the sequences all integer parameters precede participants.
  We ask that the only variables in $\sexp$, and in all the expressions contained in $\contrC$ and $p$, are the ones taken as parameters by $\cVarX$.
\end{defn}

The term $\clauseFunding{\sexp}{p}$ gives conditions that must hold in order to activate $\contrC$. In particular, $\sexp$ denotes the amount of tokens that must be stored in $\contrC$. These tokens will be taken both from the calling contract, and from additional deposits. 
The proposition $p$ is a predicate on the actual values that are passed to the clause at call time. If $p$ is not satisfied then the clause cannot be called, and $\contrC$ will not be activated.
When $p = true$, we simply write $\clauseAdv{\sexp}{}{\contrC}$.

\paragraph*{Evaluation and closed form}
We specify below the substitution of actual values for parameters.
By substituting the parameters of a clause $\cVarX$ with two sequences of actual values $\vec{\stVal}$ and $\vec{\dinVal}$ (with $\stVal[i]\in \mathbb{Z} \cup \Part$, and $\dinVal[i] \in \mathbb{Z}\cup \Part \cup \setenum{\ast}$) we produce an \emph{instantiated clause}, denoted with $\clauseCallXParam$. 
We define a relation $\clauseCallX{a}{b}\equiv \clauseAdv{v}{}{\contrCi}$ that holds iff no $b_j$ is equal to $\ast$ and the following conditions hold:
\begin{enumerate}

\item[$(i)$] The actual values are well-typed, \ie the types of $\vec{\stVal}$ and $\vec{\dinVal}$ match the ones of $\vec{\stPar}$ and $\vec{\dinPar}$ respectively. In particular, there must be $\ast$ among the elements of $\vec{\dinVal}$.

\item[$(ii)$] $\sem{\sexp\setenum{\vec{a}/\vec{\alpha}, \vec{b}/\vec{\beta}}} = v$; 

\item[$(iii)$] $\sem{p\setenum{\vec{a}/\vec{\alpha}, \vec{b}/\vec{\beta}}} = true$;

\item[$(iv)$] $ \sem{\contrC\setenum{\vec{a}/\vec{\alpha}, \vec{b}/\vec{\beta}}} = \contrCi$.

\end{enumerate}

Here, writing  $\vec{a}/\vec{\alpha}$ means that we replace every parameter $\alpha_i$ in the expression with the value $a_i$, and $\sem{\cdot}$ is a simple evaluation operator that performs all arithmetic and logic operations present in an expression.
Notice that, after the evaluation, every expression inside $\contrCi$ is reduced to an constant. Such a contract is said to be in \emph{closed form}.
Unless specified otherwise, from this point onward we will be working with contracts in closed form.

\begin{defn}[Configurations]
  A configuration $\confG$ is a term $\tilde{\confG} \mid t \mid \confTaint{w}$, where $t\in \mathbb{N}$ denotes the time, $\confTaint{w}$ is the destroyed funds counter, and $\tilde{\confG}$ has the following syntax:
  \begin{align*}
    \tilde{\confG} \bnfdef \ & \emptyset   && \text{empty}
    \\
    \bnfmid & \confContrTime[x]{t}{\contrC}{\valV} && \text{active contract}
    \\
    \bnfmid & \confDep[x]{\pmvA}{\valV} && \text{deposit}
    \\
    \bnfmid & \confF && \text{complete advertisement}
    \\
    \bnfmid & \Theta && \text{incomplete advertisement}
    \\
    \bnfmid &\confAuth{\pmvA}{\chi} && \text{authorization}
    \\
    \bnfmid & \tilde{\confG} \mid \tilde{\confG} && \text{parallel composition}
  \end{align*}
  We also assume that
  \begin{enumerate*}
  \item[$(i)$] parallel composition is associative and commutative;
  \item[$(ii)$] all parallel terms are distinct and names are never repeated;
  \item[$(iii)$] all contracts $\contrC$ appearing in a configuration are in closed form.
  \end{enumerate*}
\end{defn}

\paragraph*{Active contracts}
An \emph{active contract} is a term $\confContrTime[x]{t}{\contrC}{v}$.
It is uniquely identified by its name $x$, and it represents an amount of $v$ tokens (its \emph{balance}) that can only be spent according to the conditions set by one of $\contrC$'s branches. The integer $t$ is the time when the contract has been added to the configuration. We assume $\contrC$ to be in closed form.

\paragraph*{Deposits}
The terms $\confDep[x]{\pmvA}{v}$ in a configuration, called \emph{deposits}, are uniquely identified by their name $x$, and  represent an amount $v$ of tokens owned by participant $\pmvA$.
The owner of a deposit is the only one who can provide the authorization to spend it. \Cref{fig:semantics-deposits} defines the semantics of deposits: $\confDep[x]{\pmvA}{v}$ can be donated to another participant, split into two smaller deposits, or merged with another one.
Moreover, a deposit can be spent to fund the activation of a new contract, or the execution of a contract step.
Deposit can be destroyed. In this case the destroyed tokens are added to a counter 
% counter $\confTaint{}$. 
% \paragraph*{Destroyed fund counter}
$\confTaint{w}$, which keeps track of the tokens that was stored in a deposit which the owner has decided to destroy, denoting with $w \in \mathbb{N}$ the total. 
We assume that only dishonest participants can spend the tokens in the counter, and that they can do so freely, without the need to produce any authorization term.

\begin{figure}[t!]
\small
\centering
  \[
  \begin{array}{c}
    \irule
        {\confG \text{ contains }  \confDep[{z_1}]{\pmvA}{\valV[1]} \text{ and } \confDep[{z_2}]{\pmvA}{\valV[2]} }
        {\confG \xrightarrow{{\it auth-join}(\pmvA,z_1,z_2,i)}	\confG \mid \confAuth{\pmvA}{\authBranch{z_1,z_2,i}{v_1 + v_2}}}
    \\[\irulespacemid]
    \irule
        {\confG = \confGi \mid \confAuth{\pmvA}{\authBranch{z,z'}{v + v'}} \mid \confAuth{\pmvA}{\authBranch{z,z'}{v + v'}} \quad y \; \text{fresh}}
	{ \confG \mid \confDep[z]{\pmvA}{\valV} \mid \confDep[z']{\pmvA}{\valVi}
	  \xrightarrow{{\it join}(x,y)} \confGi \mid
	  \confDep[y]{\pmvA}{\valV + \valVi} 
	}
    \\[\irulespacemid]
    \irule
	{ \confG \text{ contains }  \confDep[z]{\pmvA}{\valV+ \valVi} }
	{ \confG
	  \xrightarrow{{\it auth-divide}(\pmvA,z,\valV,\valVi)} 
	  \confG \mid \confAuth{\pmvA}{\authBranch{z}{v, v'}}
	}
    \\[\irulespacemid]
    \irule
	{\confG = \confGi \mid \confAuth{\pmvA}{\authBranch{z}{v, v'}} 
	  \quad y,y' \; \text{fresh}}
	{ \confG \mid \confDep[z]{\pmvA}{\valV + \valVi}	\xrightarrow{{\it divide}(z,\valV,\valVi)} \confGi \mid	\confDep[y]{\pmvA}{\valV} \mid \confDep[y']{\pmvA}{\valVi}	
	}
    \\[\irulespacemid]
    \irule
	{\confG \text{ contains }\confDep[z]{\pmvA}{\valV}}
	{ \confG \xrightarrow{{\it auth-donate}(\pmvA,z,\pmvB)} 
	  \confG \mid		\confAuth{\pmvA}{\authDonate{z}{\pmvB}}
	}
    \\[\irulespacemid] 
    \irule
	{\confG = \confGi \mid \confAuth{\pmvA}{\authDonate{z}{\pmvB}} 
	  \quad \varY \; \text{fresh}}
	{\confG \mid \confDep[z]{\pmvA}{\valV}
	  \xrightarrow{{\it donate}(z,\pmvB)} 
	  \confGi \mid \confDep[\varY]{\pmvB}{\valV}
	}
  \end{array}
  \]
  \caption{Semantics of \illum deposits.}
  \label{fig:semantics-deposits}
\end{figure}

\paragraph*{Advertisements (general)}
Some actions, in particular the ones related to contracts, require to be advertised before they can be performed, meaning that a participants who wants to execute them has to inform the others by introducing an \emph{advertisement term} in the configuration.
Such terms can be of two kinds: either \emph{complete} or  \emph{incomplete}.
An incomplete advertisement $\Theta$ is only used as a message, while the complete term $\confF$ is also needed by the semantics in order to carry on certain actions, as we will see in the next section. 
As a result, incomplete advertisement have the option of leaving some features unspecified.
The actions that can be advertised are the following: $(i)$ the activation of a new contract; $(ii)$ the continuation of an active contract; $(iii)$ the destruction of a set of deposits.
In the following paragraphs we will describe precisely how each of the three term is structured. The main focus will be the definition of complete advertisement, and, after that, we will mention what are the parts can be left unspecified to obtain the incomplete version.

\paragraph*{Advertisement (initial)}
A complete initial advertisement is a term $\confF = \confAdvInitN{ \clauseCallXParam }{\vec{z}}{w}{h}$, where $\clauseCallXParam$ presents the proposed contract and its parameters, $\vec{z}$ is a non empty list of deposit names (that will be spent to fund the initialization), and $w\in \mathbb{N} \cup \setenum{\star}$ is the amount of destroyed funds that are going to be taken from $\confTaint{}$ and used in the initialization. Here $\star$ is a special symbol that means no currency is taken from the counter: it will be treated as 0 in arithmetical operations \footnote{In a configuration an advertisement term with $w=\star$ behaves identically to one with $w=0$. The only difference between the two terms is that honest participants will only be allowed to produce terms that have $w= \star$. We address the reason behind the use of $\star$ when comparing the symbolic model with the computational one.}. The subscript $h \in \mathbb{N}$ is just a nonce, used to differentiate two otherwise identical terms.
In an initial advertisement the clause $\cVarX$ must have its precondition $p$ equal to $true$. This is a technical requirement, added to simplify the language implementation, but it can also be justified intuitively: since the contract is not yet started, if the participants want the arguments to satisfy some proposition $p$, they can simply choose to initiate a different contract.
In a complete advertisement the special symbol $\ast$ must not appear in the parameters $\dinPar$ passed to the clause.

\paragraph*{Advertisements (continuation)}
When a configuration contains an active contract $\confContrTime[x]{t}{\contrC}{v}$, a participant that wants to execute $\contrD$, the $j$-th branch of $x$, will produce the advertisement term $\confF = \confAdvN{\complD}{\vec{z}}{w}{x}{j}{h}$.
Like in the previous case, $\vec{z}$ and $w$ are used to specify the source of additional funds used in the continuation action (here we also allow $\vec{z}$ to be an empty list); and $h\in \mathbb{N}$ is again a nonce.
The term $\complD$ is called \emph{advertised branch}, and it needs a more detailed presentation. 
$\complD$ is constructed by taking a branch $\contrD$ while replacing the question marks $?$ inside a $\callname$ termination with actual values $b_j$ which can be freely chosen by the participant producing the advertisement term. We will identify $\complD$ and $\complDi$ if one can be obtained from the other by exchanging the decorations' order. Notice that, since every active contract is in closed form, the only expressions appearing inside $\complD$ will be constants. To denote that $\complD$ has been constructed starting from $\contrD$ we write $\complD \approx \contrD$.
The terms $\complD$ have the syntax:
\begin{align*}
  \complD \bnfdef \ &\withdrawC{ (\splitB{v_1}{\pmvA[1]} \cdots \splitB{v_n}{\pmvA[n]} ) } 
  \\
  & \bnfmid \ncadv{(\clauseCallXParam[1] \cdots \clauseCallXParam[n] ) }  
  \\
  & \bnfmid \authC{\pmvA}{\complD} \bnfmid \afterC{t}{\complD} \bnfmid \afterRelC{\delta}{\complD} 
\end{align*}
with $a^i_j \in \mathbb{Z} \cup \Part$ and $b^i_j \in \mathbb{Z}\cup\Part \cup \setenum{\ast}$.

\paragraph*{Advertisements (destruction)} 
A destruction advertisement $\confF = \confAdvDesN{\vec{z}}{w}{h}$ is produced when one wants to remove some deposits from the configuration, adding their value to the destroyed funds counter.
Here $\vec{z}$ is a nonempty list of the deposit names that are going to be destroyed, $w \in \mathbb{N} \cup \setenum{\star}$ is an amount of funds from the counter, $h$ is a nonce that differentiates two otherwise identical terms.

\paragraph*{Incomplete advertisements}
In an incomplete advertisement term $\Theta$ some informations may be left unspecified.
Again, we have three types of advertisements:
\begin{enumerate*}

\item[$(i)$] Initial, with $\Theta=\confAdvInit{\clauseCallXParam}{\vec{z}}{w}$. Here some of the values $\dinPar[i]$ may be equal to $\ast$, the sequence $\vec{z}$ may be empty, and $w$ can also take the value $\ast$.

\item[$(ii)$] Continuation, with $\Theta = \confAdv{\complD}{\vec{z}}{w}{x}{j}$. Similarly to the case above, $w$ can be set equal to $\ast$, and values $b_j$ inside $\callname$ operations in $\complD$ can also be set to $\ast$.  

\item[$(iii)$] Destruction, with $\Theta = \confAdvDes{\vec{z}}{w}$, where we allow $w$ to be equal to $\ast$.

\end{enumerate*}

\paragraph*{Validity of advertisements}
We now define \emph{validity} of an advertisement, which must hold for an advertised operation to be performed. Mainly, validity checks that all the terms in the advertisement actually occur in the configuration.
Remember that the value $\star$ is treated as a 0 in all arithmetic operations.
A complete advertisement $\confF$ is \emph{valid in $\confG$} if one of the following holds:
\begin{itemize}

\item (Initial) $\confF = \confAdvInitN{\clauseCallXParam}{\vec{z}}{w}{h}$, where $\clauseCallXParam \equiv \clauseAdv{v}{}{\contrC}$, and the configuration $\confG$ contains deposits $\confDep[z_j]{\pmvA[j]}{u_j}$ and the counter $\confTaint{w'}$, with $w'\geq w$ and $\sum_j u_j + w\geq v\geq 0$.

\item (Continuation)
  $\confF = \confAdvN{\complD}{\vec{z}}{w}{x}{j}{h}$, and the following conditions hold:
  \begin{enumerate*}
  \item[$(i)$] the configuration contains the deposits $\confDep[z_j]{\pmvA[j]}{u_j}$, the contract $\confContrTime[x]{t}{\contrC}{v}$, and $\confTaint{w'}$ with $w' \geq w$;
  \item[$(ii)$] The $j$-th branch of $\contrC$ is a $\contrD$ such that $\complD\approx \contrD$;
  \item[$(iii)$] the time $t$ in $\confG$ is greater than all $t_i$ appearing in $\afterName$ decorations of $\complD$, and $t-t_0$ is greater than all $\delta_i$ appearing in $\afterRelName$ decorations of $\complD$;
  \item[$(iv)$] if $\complD$ ends in $\withdrawC{(\splitB{v_1}{\pmvA[1]},\cdots \splitB{v_n}{\pmvA[n]})}$ then we must have $\sum_j u_j + w+ v\geq \sum_i v_i$;
  \item[$(v)$] if instead $\complD$ ends in $\ncadv{(\clauseCallXParam[1], \cdots, \clauseCallXParam[n])}$ then there must be $v_i \geq 0$, and $\contrC[i]$ such that for every $i= 1 \cdots  n$ we have $\clauseCallXParam[i]\equiv \clauseAdv{v_i}{}{\contrC[i]}$, and $\sum_j u_j + w+ v \geq \sum_i v_i$;
  \end{enumerate*}

\item (Destruction) $\confF= \confAdvDesN{\vec{z}}{w}{h}$ and the configuration $\confG$ contains the deposits $z_j$, and  $\confTaint{w'}$ with $w'\geq w$.

\end{itemize}

\paragraph*{Authorizations}
Authorization are terms of the form $\confAuth{\pmvA}{\chi}$, where $\pmvA$ is the authorizer, and $\chi = \authBranch{\cdots}{\cdots}$ has a LHS that denotes what is being authorized, and a RHS denoting the authorized action.
Authorizations are required for all deposit actions (joining, dividing or donating deposit), and for spending the deposits in an advertised action. Lastly, some contract branches may require a participant authorization:
\begin{enumerate}

\item $\authBranch{z}{\confF}$, where $z$ is a deposit, is used to authorize the use of $z$ to fund the action advertised by $\confF$.

\item $ \authBranch{x}{\confF}$, where $x$ is a contract and  $\confF = \confAdvN{\authC{\pmvA}{\complD}}{\vec{z}}{w}{x}{j}{h}$, is used to authorize the execution of the $j$-th branch of $x$, satisfying the decoration $\authC{\pmvA}{\cdots}$.

\item $\authBranch{z_1,z_2,i}{v_1 + v_2}$ is used to authorize the use of deposit $\confDep[{z_i}]{\pmvA}{v_i}$ in a join operation with another deposit $z_j$ of value $v_j$ (we have $i\neq j$ and $i,j \in \setenum{1,2}$).

\item $\authBranch{z}{v, v'}$ is used to authorize the splitting of deposit $\confDep[{z}]{\pmvA}{v+v'}$ into two, of value $v$ and $v'$ respectively.

\item $\authDonate{z}{\pmvB}$ is used to authorize the transfer of deposit $\confDep[{z}]{\pmvA}{v}$ to a participant $\pmvB$.
If $\chi = \authBranch{z}{\confF}$, the authorization allows to use the deposit $z$ to fund the action advertised by $\confF$.

\end{enumerate}

\paragraph*{Time}
A configuration $\confG$ keeps track of time by simply including an integer $t$. This term is used to check whether a branch of an active contract decorated by $\afterName$ or $\afterRelName$ can be executed.

\begin{defn}[Semantics]
The operational semantics of \illum is a labelled transition system between configurations, defined by the rules in \Cref{fig:semantics-contracts}.
%, and \Cref{fig:semantics-deposits}.
\end{defn}

\begin{figure}[b!]
\centering
\small
  \[
  \begin{array}{c}
    \irule{
      \confG \text{ does not contain } \Theta
    }{
      \confG \xrightarrow{{\it msg}(\Theta)} \confG \mid \Theta
    }
    \\[\irulespacemid]
    \irule{
      \confF \text{ valid in } \confG  \quad \confG \text{ does not contain } \confF
    }{
      \confG \xrightarrow{{\it adv}(\confF)} \confG \mid \confF
    }
    \\[\irulespacemid]
    \irule{
      \begin{array}{c}
        \confG \text{ contains } \confF  \text{ but not } \confAuth{\pmvB}{\authBranch{z}{\confF}}
        \\[\irulespacepre]
        \confF \text{ valid in } \confG \text{ and funded with deposit } z
      \end{array}
    }{
      \confG \xrightarrow{{\it auth-in}(\pmvB,z,\confF)} \confG \mid \confAuth{\pmvB}{\authBranch{z}{\confF}}
    }
    \\[\irulespacemid]
    \irule{
      \begin{array}{c}
        \confG \text{ contains } \confF  \text{ but not } \confAuth{\pmvA}{\authBranch{x}{\confF}}
        \\[\irulespacepre]
        \confF = \confAdvN{\complD}{\vec{z}}{w}{x}{j}{h} \text{ valid in } \confG \quad  \complD \approx \authC{\pmvA}{\contrDi}
      \end{array}
    }{
      \confG \xrightarrow{{\it auth-act}(\pmvA,\confF)} \confG \mid \confAuth{\pmvA}{\authBranch{x}{\confF}}
    }
    \\[\irulespacemid]
    \irule{	
      \begin{array}{c}
        \confF = \confAdvInitN{\clauseCallXParam}{\vec{z}}{w}{h}  \text{ valid in }\confGi 
        \quad
        \clauseCallXParam \equiv \clauseAdv{v}{}{\contrC}
        \\[\irulespacepre]
        \confD[pre]= \confF\mid \big( \mmid_j \confDep[z_j]{\pmvB[j]}{u_j} \big)  \mid  \big(\mmid_j \confAuth{\pmvB[j]}{\authBranch{z_j}{\confF}} \big) 
      \end{array}
    }{
    \begin{array}{ll}
      & \confGi = (\confG \mid \confD[pre] \mid \confTaint{w'} \mid t) 	
      \\
      \xrightarrow{{\it init}(\confF,x)}
      & \confG \mid \confContrTime[x]{t}{\contrC}{v} \mid \confTaint{(w'-w)} \mid t
    \end{array}
    }
    \\[\irulespacemid]
    \\[1pt]
    \irule{
      \begin{array}{c}
        \contrFmt{D}= \authC{\pmvA[1]}{ \decC{\cdots}{ \authC{\pmvA[n]}{ \contrDi }}} 
        \qquad
        \contrDi \not= \authC{\pmvA}{\contrDii}	
        %\contrFmt{D}= \authC{\pmvA[1]}{ \decC{\cdots}{ \authC{\pmvA[n]}{ \afterC{ t_1, \cdots, t_k}{ \afterRelC{\delta_1, \cdots, \delta_h }{\withdrawC{(\cdots)} } }}} } 
        \\[\irulespacepre]
        \confF  = \confAdvN{\complD}{\vec{z}}{w}{x}{j}{h}  \text{ valid in $\confGi$ with } \complD \approx \contrD
        \\[\irulespacepre]
        \complD \text{ ends in } \withdrawname \ (v_1 \rightarrow \pmvC[1] \cdots v_m \rightarrow \pmvC[m])
        \\[\irulespacepre]
        \confD[auth] = \big( \mmid_l \confAuth{\pmvB[l]}{\authBranch{z_l}{\confF}}  \big) \mid \big( \mmid_k \confAuth{\pmvA[k]}{\authBranch{x}{\confF}} \big) 
        \\[\irulespacepre]
        \text{ (if $\pmvA[k]= \pmvA[k']$ the authorization only appears once) }
        \\[\irulespacepre]
        \confD[dep] =  \big( \mmid_l \confDep[z_l]{\pmvB[l]}{u_l} \big)
        \quad
        \confD[post] = \big( \mmid_i \confDep[y_i]{\pmvC[i]}{v_i}\big)
        \\[\irulespacepre]
        \confD[pre] = \confF \mid \confD[auth] \mid \confD[dep]\mid \confContrTime[x]{t}{\contrC}{v}
        \quad
        y_1\cdots y_m \text{ fresh} 
      \end{array}
    }{
    \begin{array}{ll}
    & \confGi=(\confG \mid \confD[pre] \mid \confTaint{w'} ) 
    \\
    \xrightarrow{{\it send}(\confF)} 
    & \confG \mid \confD[post] \mid \confTaint{(w'-w)}
    \end{array}
    }
    \\[\irulespacemid]
    \\[1pt]
    \irule{
      \begin{array}{c}
        \contrFmt{D}= \authC{\pmvA[1]}{ \decC{\cdots}{ \authC{\pmvA[n]}{ \contrDi }}}
        \qquad
        \contrDi \not= \authC{\pmvA}{\contrDii}	
        %\contrFmt{D}= \authC{\pmvA[1]}{ \decC{\cdots}{ \authC{\pmvA[n]}{ \afterC{ t_1, \cdots, t_k}{ \afterRelC{\delta_1, \cdots, \delta_h }{\ncadv{(\cdots)} } }}} } 
        \\[\irulespacepre]
        \confF  = \confAdvN{\complD}{\vec{z}}{w}{x}{j}{h} \text{ valid in $\confGi$} \text{ with } \complD \approx \contrD
        \\[\irulespacepre]
        \complD \text{ ends in } \ncadv{(\clauseCallXParam [1] \cdots \clauseCallXParam[m])} 
        \\[\irulespacepre]
        \text{and } \forall i. \;\clauseCallXParam[i]\equiv \clauseAdv{v_i}{}{\contrC[i]}
        \\[\irulespacepre]
        \confD[auth] = \big( \mmid_l \confAuth{\pmvB[l]}{\authBranch{z_l}{\confF}}\big) \mid \big( \mmid_k \confAuth{\pmvA[k]}{\authBranch{x}{\confF}} \big) 
        \\[\irulespacepre]
        \text{ (if $\pmvA[k]= \pmvA[k']$ the authorization only appears once) }
        \\[\irulespacepre]
        \confD[dep] =  \big( \mmid_l \confDep[z_l]{\pmvB[l]}{u_l} \big)
        \quad
        \confD[post] = \big( \mmid_i \confContrTime[y_i]{t}{\contrC[i]}{\valV[i]}\big)
        \\[\irulespacepre]
        \confD[pre] = \confF \mid \confD[auth]\mid \confD[dep] \mid \confContrTime[x]{t_0}{\contrC}{v}
        \quad
        y_1 \cdots y_m \text{ fresh }  
      \end{array}
    }{\begin{array}{ll}
      & \confGi= (\confG  \mid \confD[pre] \mid  \confTaint{w'} \mid t) 
      \\
      \xrightarrow{{\it call}(\confF)} 
      & \confG \mid \confD[post] \mid \confTaint{(w'-w)} \mid t
      \end{array}
    }
    \\[\irulespacemid]
    \\[1pt]
    \irule
    {\delta>0}
    {\confG \mid t \xrightarrow{{\it delay}(\delta)} \confG \mid t+\delta}
  \end{array}
  \]  
  \caption{Semantics of \illum contracts.}
  \label{fig:semantics-contracts}
\end{figure}

\section{Computational Model}
\label{sec:app-computational}

In this appendix we present in more detail low level model of the blockchain that serves as target of compilation for \langname contracts.

\begin{defn}[Transaction]
  A transaction $\txT$ is defined as a 5-uple $(\txIn{}$, $\txWit{}$, $\txOut{}$, $\txAfterAbs{}$, $\txAfterRel{})$ where
  \begin{itemize}

  \item $\txIn{}$ is the list of inputs. Each element of $\txIn{}$ is a pair $(\txTi, i)$, where $\txTi$ is a transaction and $i$ is an integer.

  \item  $\txWit{}$ is the list of witnesses. It has the same length as $\txIn{}$, and each element of $\txWit{}$ is a list of integers.

  \item  $\txOut{}$ is the list of outputs. Each output is a triple  $(\txval, \txscript, \txarg)$, where $\txarg$ is a list of integers, $\txscript$ is a script (its syntax will be specified in the next paragraphs), and $\txval$ is an integer.

  \item  $\txAfterAbs{}$ is the absolute timelock, and it is a non negative integer.

  \item $\txAfterRel{}$ is the list of relative timelocks. It has the same lenth as $\txIn{}$ and each of its elements is a non negative integer.

  \end{itemize}
  Given $\txTag{}{l}\in \setenum{\txIn{},\txWit{},\txOut{},\txAfterRel{}}$, we will use $\txTag[j]{}{l}$ to denote its  $j$-th element.
  The lists $\txIn{}$, $\txWit{}$, and $\txAfterRel{}$ may be empty; if that is the case we denote them with $\bot$ and we talk about an \emph{initial} (or \emph{coinbase}) transaction.
\end{defn}

\begin{figure}
  \small
  \begin{align*}
    e &\bnfdef v  && \text{integer constant} 
    \\
    &\bnfmid e \circ e' &&\text{binary operations $(\circ \in \setenum{ + , -, =, < }) $} 
    \\
    &\bnfmid e.n &&\text{$n$-th element of a list}
    \\
    & \bnfmid \rtxw &&\text{witnesses of the redeeming tx input}
    \\
    & \bnfmid |e| && \text{size (in bytes)}
    \\
    & \bnfmid H(e) && \text{hash \footnotemark}
    \\
    & \bnfmid \ifE{e}{e'}{e''}&& \text{conditional check}
    \\
    & \bnfmid \versig{e}{e'} &&\text{signature verification}
    \\
    &\bnfmid \afterAbs{e}{e'} &&\text{absolute time constraint}		 
    \\
    &\bnfmid \afterRel{e}{e'} &&\text{relative time constraint}	
    \\
    &\bnfmid \txof{\txo}{\txf}  && \text{field $\txf \in \setenum{\txarg, \txval}$}
    \\
    &&&\text{of $\txo \in \setenum{\ctxo{}, \rtxo{}{e}}$}
    \\
    &\bnfmid \inidx &&\text{index of redeeming tx input}
    \\
    & \bnfmid \inlen{{\sf tx}} && \text{number of inputs of ${\sf tx} \in \setenum{\rtx, \ctx} $ }
    \\
    & \bnfmid \outlen{{\sf tx} } && \text{number of outputs of ${\sf tx} \in \setenum{\rtx, \ctx} $ }
    \\
    &\bnfmid \verscript{e}{e'} &&\text{checks if $\rtxo{\txscript}{e'} = e$ } 
    \\
    &\bnfmid \verrec{e} && \text{checks if $\rtxo{\txscript}{e} = \ctxo{\txscript}$}
  \end{align*}
  \caption{Syntax of scripts.}
  \label{fig:scripts:syntax}
\end{figure}

\begin{defn}[Syntax of scripts]
  The $\txscript$ field of a transaction output has the syntax
  in~\Cref{fig:scripts:syntax}.
  There, the terms $\ctx$ and $\rtx$ are used to denote the current and redeeming transaction respectively. Moreover, the term $\ctxo{}$ is used by a script to refer to the current output (i.e. the one of which it is the script); and the term $\rtxo{}{n}$, refers to the $n$-th output of the redeeming transaction.
  In $\txscript$ use some shorthands for common logical operations, setting $(i)$ $true \eqdef 1$, $(ii)$ $false \eqdef 0$, $(iii)$ $e \andE e' \eqdef \ifE{e}{e'}{false}$, $(iv)$ $\notE{e}\eqdef \ifE{e}{false}{true}$, $(v)$ $e \orE e' \eqdef \ifE{e}{true}{e'}$.
\end{defn}

\paragraph*{Script evaluation}
In order to determine if a transaction can redeem an output, its script must be executed. For this reason, we need to define an evaluation semantics for scripts. We let $\sem[(\txT,i)]{\cdot}$ be the evaluation operator, where $\txT$ is the redeeming transaction and $i$ index of the redeeming input.
For ease of notation, in the following paragraphs we will shorten  $\sem[(\txT,i)]{\cdot}$ to $\sem{\cdot}$, and implicitly assume that $\txT.\txIn[i]{}$ is the redeeming input, unless specified otherwise. We will also assume that the transaction that is being redeemed is named $\txTi$, and that the redeemed output is its $j$-th. Note that $\txTi$ and $j$ can be determined from $\txT$ and $i$, since we have that
\[
\txTi = {\it first}(\txT.\txIn[i]{}) \quad \text{ and } \quad j = {\it second}( \txT.\txIn[i]{} ).
\]
The evaluation yield $\bot$ when a failure occurs (\ie trying to access the $n$-th element of a list with less than $n$ terms).
All the operators in the script's syntax are \emph{strict}, meaning that their evaluation yields $\bot$ if one of their arguments is $\bot$. 
Here we highlight the behaviour of the evaluation semantics on the non-trivial terms:
\begin{itemize}
	\item $\sem{\rtxw}$ evaluates to $\txT.\txWit[i]{}$, which is the sequence of witnesses associated to the redeeming input.
	\item $\sem{\versig{e}{e'}}$ evaluates to $true$ if the signature $\sem{e'}$ is correctly verified on the hash of $\txFmt{T^*}$ (which denotes the transaction obtained by replacing the $\txWit{}$ field in $\txT$ with $\bot$) against the key $\sem{e}$, and to $false$ otherwise.
	\item  $\sem{\afterAbs{e}{e'}}$ evaluates to $\sem{e'}$ if the field $\txAfterAbs{}$ of the redeeming transaction $\txT$ is greater or equal to $\sem{e}$, otherwise it evaluates to $\bot$. Similary, $\sem{\afterRel{e}{e'}}$ evaluates to $\sem{e'}$ if the field $\txAfterRel[i]{}$ of $\txT$ is greater or equal to $\sem{e}$, otherwise it fails, yielding $\bot$. 
	\item 	$\sem{\verscript{e}{e'}}$ evaluates to true when the script of the $\sem{e'}$-th output of the redeeming transaction is equal to $\sem{e}$, otherwise it returns $false$.
	\item  $\sem{\verrec{e}}$ evaluates to true when the $\sem{e}$-th output of the redeeming transaction is equal to the output that is being redeemed, otherwise it returns $false$. 
\end{itemize}
Until now we have discussed a lot about a transaction's input \say{redeeming} a certain output, without giving a proper definition. Now that we have a semantics of script, we can be more precise.
\begin{defn}[Redeeming an output]
	We say that the $k$-th input of a transaction $\txT$ published at time $t$ can redeem the $k'$-th output of $\txTi$ published at time $t'$, and we write $(\txTi, k', t')\rightsquigarrow(\txT, k, t)$, if the following conditions are verified:
	\begin{enumerate}
		\item $\txT.\txIn[k]{} = (\txTi, k' )$.
		\item $\sem[(\txT,k)]{\txTi.\txOut[k']{}.\txscript} = true$.
		\item $t\geq \txT.\txAfterAbs{}$
		\item $t - t' \geq  \txT.\txAfterRel[i]{}$.
	\end{enumerate}
\end{defn}

Finally, we may define the structure that keeps track of all the transactions: the blockchain.

\begin{defn}[Blockchain]
	A \emph{blockchain} $\bcB$ is a sequence of pairs $(\txT[i], t_i)$, such that the sequence of $t_i$ is nondecreasing. 
	If $(\txT, t)$ is an element of $\bcB$ we say that the transaction \emph{$\txT$ appears at time $t$ in $\bcB$}.
	A blockchain is said to be \emph{consistent} if the following conditions hold:	
	\begin{enumerate}
		\item The first pair in $\bcB$ is $(\txT[0], 0)$ with $\txT[0]$ initial, and this is the only transaction appearing in the $\bcB$.
		\item If $\txT$ appears in the blockchain at time $t$, and it is not the first transaction, then each of its inputs redeems an output of some transaction $\txT[j]$, appearing in $\bcB$ at an earlier time.
		\item Every output of a transaction in $\bcB$ is referenced at most once by an input of a later transaction.
		\item If $\txT$ is in $\bcB$, and $\txo_{\txTag{}{j}}$ denotes the output referenced by $\txT.\txIn[j]{}$, then we have
		\[
		\sum_{i} \txT.\txOut[i]{}.\txval \leq \sum_j \txo_{\txTag{}{j}}.\txval
		\]
		Meaning that the sum of $\txT$ outputs' values must not exceed the sum of its inputs' values.
	\end{enumerate} 
	If $\bcB  (\txT, t)$ is consistent then we say that $(\txT,t)$ is a \emph{consistent update} to $\bcB$.
	An output of a transaction appearing in $\bcB$ is said to be \emph{spent} if in the blockchain there appears a transaction that redeems it. The set of unspent outputs is denoted by $\utxo{\bcB}$, (and in general we will abbreviate the expression \say{unspent transaction output} with UTXO).
\end{defn}
Notice how, in this model, a consistent blockchain presents only one initial transaction, meaning that mining is not included in our computational model. 
We conclude this section by defining deposit outputs: these are outputs with a certain structure, and they will be useful to represent symbolic deposits in the computational model.
\begin{defn}[Deposit output]\label{def:deposit-output}
	An output $\txo$ is said to be a \emph{deposit output owned by $\pmvA$} when it has exactly three arguments, the following script
	\[
	\txscript =  \versig{\compileFmt{key}(\ctxo{\txarg[3])} }{\rtxw.1},
	\]
	and the third argument is equal to $\pmvA$.
\end{defn}

\section{The \langname compiler}
\label{sec:app-compiler}

In this appendix we give the full definition of the compiler, which is will allow us to encode symbolic contracts as transaction outputs in the computational model.
The compiler will be formalized as a function that takes an initial or continuation advertisement $\confF$ and constructs a transaction $\txT[\confF]$.
First we will focus on the construction of the output(s) of $\txT[\confF]$; then we will spend a few words in order to describe the various auxiliary inputs taken by the compiler, which are used to construct the other fields of $\txT[\confF]$. We will then conclude with the full definition of the compiler $\compile{}$.
\paragraph*{Constructing the outputs (general)}
As we mentioned, we will be encoding contracts as transaction outputs. Moreover, in our implementation, all contracts descending from the same initial advertisement will share a common script. We will be able to discern what specific contract is being encoded in an output by looking at its arguments $\txarg$, which the script will access in order to enforce the correct execution path.
If $\confF$ is initial, then $\txT[\confF]$ will have a single output that represents the newly activated contract; otherwise, if $\confF$ is a continuation advertisement $\confAdvN{\complD}{\vec{z}}{w}{x}{j}{h}$, then $\txT$ will have multiple outputs, either representing deposits (if $\complD$ ends with a $\withdrawname$), or contracts (if $\complD$ ends with a $\callname$).

\paragraph*{Constructing the outputs (value)}
The value of each output of $\txT[\confF]$ is easily determined.
If $\confF = \confAdvInitN{\clauseCallXParam}{\vec{z}}{w}{h}$ is a valid initial advertisement, with $\clauseCallXParam \equiv \clauseAdv{v}{}{\contrC}$, then the single output of $\txT[\confF]$ will have value $v$.
If $\confF = \confAdvN{\complD}{\vec{z}}{w}{x}{j}{h}$ is a valid continuation advertisement, with $\complD$ ending in $\ncadv{(\clauseCallXParam[1], \cdots, \clauseCallXParam[n])}$, and for each $i= 1\cdots n$ we have $\clauseCallXParam[i]\equiv \clauseAdv{v_i}{}{\contrC[i]}$, then $\txT[\confF]$ will have $n$ outputs, and the $i$-th output will have value $v_i$.
Lastly, if $\confF$ is a valid continuation advertisement, with $\complD$ ending in $\withdrawC{(\splitB{v_1}{\pmvA[1]},\cdots \splitB{v_n}{\pmvA[n]} )}$, then $\txT[\confF]$ will have $n$ outputs, and the $i$-th output will have value $v_i$.

\paragraph*{Constructing the outputs (arguments)}
The value of the $\txarg$ field of $\txT[\confF]$ will be determined differently depending on the type of advertisement $\confF$. Here, we also present the various names that we will give to arguments to avoid having to refer to them only by their position in the sequence.

If $\txT[\confF]$ is compiled from an initial advertisement $\confF = \confAdvInitN{\clauseCallXParam}{\vec{z}}{w}{h}$, with then the first three element of its $\txarg$ field will be called $\txnonce$, $\txbranch$ and $\txname$. 
The first is just a computational counterpart to the symbolic nonce $h$ that appears in $\confF$, and it gives us a way to force two otherwise identical transactions to be distinct; the second is just a dummy argument, used to make this case more similar to the continuation case, and we will set it to 0; the third is equal to $\cVarX$, the name of the clause that this output will encode.

After those, there are two sequences of arguments $\txalpha[i]$ and $\txbeta[i]$, one for each parameter of the clause $\clauseDefX{\alpha}{\beta}$: they store the values $a_i, b_i$ specified in $\confF$.

If instead the transaction is compiled from a continuation advertisement $\confF  = \confAdvN{\complD}{\vec{z}}{w}{x}{j}{h}$, with $\complD$ ending in $\ncadv{ ( \clauseCallX[1]{a^1}{b^1}, \cdots, \clauseCallX[n]{a^n}{b^n} )}$, then we need to specify the sequence of arguments for each of its $n$ outputs. 
For each $k\in \setenum{1,\cdots n}$, the first three elements of $\txT[\confF].\txOut[k]{}.\txarg$ are again denoted with $\txnonce$, $\txbranch$ and $\txname$. $\txnonce$ will again be the counterpart of $h$; $\txbranch$ will be set to $j$, which denotes the branch of the \say{parent} contract that this transaction is continuing; and $\txname$ will be $ \cVarX[k]$.
Then, we have the $\txalpha[i]$ and $\txbeta[i]$ arguments, referring to the parameters of the clause $\cVarX[k]$. These will take the values $a_i^k$ and $b_i^k$ specified in $\complD$.

Lastly, if a transaction is compiled from a continuation advertisement $\confAdvN{\complD}{\vec{z}}{w}{x}{j}{h}$, with $\complD$ ending in $\withdrawC{(\splitB{v_1}{\pmvA[1]}, \cdots, \splitB{v_n}{\pmvA[n]})}$, then each of its $n$ outputs will only need three arguments: $\txnonce$, $\txbranch$, and $\txowner$.
The value of the first two is set like in the previous case, while the $\txowner$ argument of the $k$-th output is set to $\pmvA[k]$.

\paragraph*{Notation for arguments and expressions}
In the construction of the script we will need to replace the parameters $\alpha_i,\beta_i$ appearing in the contract expressions $\sexp$ with the respective arguments $\ctxo{\txalpha[i]}$, $\ctxo{\txbeta[j]}$. 
In order to make the script more readable we denote this substitution with $\ctxo{\sexp}$.
For example if the contract has a term $\afterC{t}{\cdots}$ with $ t = \alpha_1+ \alpha_2 + \beta_1 + 3$, we will write $\ctxo{t}$ instead of $\ctxo{\txalpha[1]}+ \ctxo{\txalpha[2]} + \ctxo{\txbeta[1]} + 3$.
Similarly, whenever an expression uses the arguments of a redeeming transaction's output, we denote it as $\rtxo{\sexp}{j}$. This is less frequent, but needed when dealing with the preconditions of a called clause.

\paragraph*{Constructing the outputs (script)}
The construction of the script is fully detailed in \Cref{sec:compiler} of the main text.
\paragraph*{Inputs of the compiler}
What we have shown until now is how the outputs are constructed starting from a given advertisement term $\confF$. However, the compiler that we are going to define does not only create the outputs, but an entire transaction, which is the computational counterpart of the symbolic advertisement. In order to do so, the compiler will need to take some additional inputs.
First, we take two auxiliary parameters that give us information about the state of the relationship between the symbolic and the computational model, and help us check that the transaction respects certain constraints. 
These are $\keyMap$, which maps symbolic participants to their public key; and $\txMap$, which maps names of contracts or deposits in the current symbolic configuration to outputs $(\txT, j)$ in the blockchain.
Moreover, we take four additional parameters that will be used to construct certain fields of the transaction.
These are $\inCmp$, which is a non empty list of outputs $(\txT, j)$ that will be used to construct the inputs of the transaction; $\tAbsCmp$, which is an integer used to construct the absolute timelock; $\tRelCmp$, which is a list of integers with the same length of $\inCmp$ that will be used to construct the relative timelocks; and $\nonceCmp$, a non empty list of integers that will be used to construct the $\txnonce$ argument in each output.

\begin{defn}[Compiler]
	Below we define the function $\compile{}$, also known as \emph{compiler}.
	This definition is structured in two phases: first we show how to construct a transaction $\txT$ that is the \say{candidate output} of 
	\[
	\compile{\confF, \txMap, \keyMap, \inCmp, \tAbsCmp,\tRelCmp, \nonceCmp}
	\]
	and then we check that it satisfies certain constraints. If it does then $\txT$ is the actual output, otherwise the compilation fails and we output $\bot$. 
	\paragraph*{Construction}
	\begin{itemize}
		\item (Inputs) The $k$-th input of $\txT$ is set to be equal to $\inCmp[k] = (\txT[k], j_k)$.
		\item (Timelocks) The absolute timelock of $\txT$ is set to be equal to $\tAbsCmp$, while the $k$-th relative timelock is set to be equal to $\tRelCmp[k]$.
		\item (Output - initial) If $\confF =  \confAdvInitN{\clauseCallX{a}{b}}{\vec{z}}{w}{h}$, then $\txT$ will only have one output.
		The output will have $3 + |\vec{a}|+ |\vec{b}|$ arguments: we have $\txname= \cVarX$, $\txalpha[i]= a_i$, $\txbeta[i]=b_i$, and the $\txnonce$ argument is given by $\nonceCmp[1]$.The output's script is constructed as detailed in the above paragraphs, and the output's value is $v$, where $\clauseCallX{a}{b}= \clauseAdv{v}{}{\contrC}$.
		\item (Outputs - call)  If $\confF =  \confAdvN{\complD}{\vec{z}}{w}{x}{j}{h}$, where $\complD$ ends in
		$\ncadv{(\clauseCallXParam[1], \cdots, \clauseCallXParam[n])}$, with $\clauseCallXParam[k] \equiv \clauseAdv{v_k}{}{\contrC[k]}$, then $\txT[\confF]$ has $n$ outputs.
		Let $\inCmp[1]= (\txT[1],j_1)$: each of the $n$ outputs of $\txT$ will have the same script as the $j_1$-th output of $\txT[1]$. Moreover, for each $k$, the $k$-th output of $\txT$ will have arguments $\txnonce= \nonceCmp[k]$, $\txname = \cVarX[k]$, $\txbranch = j$, $\txalpha[i]= a_i^k$, $\txbeta[i]= b_i^k$, and value $v_k$.
		\item (Outputs - send) If $\confF =   \confAdvN{\complD}{\vec{z}}{w}{x}{j}{h}$, where $\complD$ ends with $\withdrawC{(\splitB{v_1}{\pmvA[1]}, \cdots, \splitB{v_n}{\pmvA[n]})}$, then $\txT$ will have $n$ outputs. For each $k$, the $k$-th output will be a deposit output of value $v_k$ owned by $\pmvA[k]$ in the sense of \Cref{def:deposit-output}. Its $\txnonce$ argument will be equal to $\nonceCmp[k]$, and $\txbranch$ will be set to $j$. 
	\end{itemize} 
	\paragraph*{Conditions} If one of the following is not satisfied then the return value of  $\compile{}$ is set to be $\bot$, otherwise it is $\txT$:
	\begin{itemize}
		\item If $\confF$ is a continuation advertisement, then the first input of $\txT$ must be $\txMap(x)$. Aside from that, regardless of the type of advertisement, if $\confF$ takes as input the deposits $\vec{z}$ then  among the inputs of $\txT$ there must appear (in order) the outputs $\txMap(z_j)$, for all $j = 1 \cdots |\vec{z}|$.   
		If $w=\star$ then $\txT$ has no other inputs. Otherwise, removing the inputs listed above leaves a list of inputs without a symbolic counterpart (i.e. not in $\ran{\txMapC}$) such that the amount of their values is equal to $w$, the sum that $\confF$ takes from the destroyed funds counter.
		\item If $\confF$ is a continuation advertisement, then $\txT.\txAfterAbs{}$ must be greater than all $t$ appearing in any $\afterName$ decoration in $\complD$. Similarly $\txT.\txAfterRel[1]{}$ must be greater than all $\delta$ appearing in any $\afterRelName$ decoration.
	\end{itemize}
\end{defn}

It's important to notice that the compiler $\compile{}$ does not constructs a new transaction starting from an active contract, but from that the advertisement that propose it. This means that two active contracts that are equal in the symbolic model, may correspond to transactions with very different outputs. We can say that a transaction will remember the whole \say{history} of a contract, while a symbolic configuration only sees active contracts in the \say{present}. In order to be able to talk, about the actual contract that is being represented in a transaction output, we give the following definition.
\begin{defn}[Contract encoded by an output]
	Take an output $\txo$ of a compiler generated transaction and a clause $\cVarX$ such that $\clauseCallXParam \equiv \contrAdv{v}{}{\contrC}$, and let $k_{\alpha}$, $k_{\beta}$ be equal to the lengths of $\vec{a}$, $\vec{b}$ respectively.
	We say that the output $\txo$ \emph{encodes the contract $\confContr{\contrC}{v}$}, if it has a total of $3+k_{\alpha} + k_{\beta}$ arguments, and the following equivalences hold:
	\begin{enumerate}[$(i)$]
		\item $\txo.\txarg[3]= \cVarX$;
		\item $\forall i \in \setenum{1 \cdots k_{\alpha}}. \ \txo.\txarg[3+i] = a_i$;
		\item $\forall i \in \setenum{1 \cdots k_{\beta}}. \ \txo.\txarg[3+k_{\alpha}+ i] = b_i$;
	\end{enumerate}
\end{defn}
Notice, that just like deposit outputs, it may be that an output \say{encoding a contract} does not actually correspond to any active contract in the configuration. However, we will show that if the computational run is coherent to the symbolic run, then at every active contract in the configuration will correspond  a unique contract output in the blockchain.

\paragraph*{Destroyed funds} 
At this point, we have talked about how contracts and deposits are represented in the computational model, but there is a third term used by the symbolic model to store currency: the destroyed funds counter $\confTaint{w}$, which must be handled in a different way.
There will not be a precise correspondence between $\confTaint{w}$ and some specific outputs: the value $w$ will instead be an approximate representation of every output that does not encode a deposit nor a contract\footnote{More precisely, this is an approximation from above, as in the next section will prove that the amount of currency stored in outputs that do not encode deposits or contracts is lower or equal to the value $w$ specified by the counter.}.

\begin{comment}
Now that both models have been presented, we can justify the existence of destroyed funds in the symbolic setting. Why do we need to model the destruction of deposits? And why do we care about tracking the amount of destroyed funds in the symbolic setting? 
The first question has a simple answer: we cannot expect an adversary to have a strategy that only appends to the blockchain simple deposit operations or  compiler-generated transactions. So, it is always possible that some of the funds representing a deposit owned by the adversary are spent in a way that cannot be expressed as a symbolic contract. Regarding the second question, we need to consider what would happen if we did not symbolically track these non-compiler generated transactions. 
Indeed, in a symbolic model that is completely blind to these transaction, we have a big limitation: the continuation of an active contract cannot add any money to it.
This is because the covenants used in the computational model are not powerful enough to set conditions on the inputs of the transaction, and it is impossible to determine where they came from. This means that once additional funds are allowed the computational model has no way to enforce the requirement of them being an encoding of symbolic deposits, so, since the two models are supposed to mirror each other, the symbolic model has to somehow keep track of these funds and allow their use in advertisement terms.
\end{comment}
\paragraph*{The $\star$ symbol}
Now that the concept of destroyed funds has been clarified, we can address the difference between choosing  $w= \star$ and $w = 0$ in an advertisement term, and show how the computational model justifies using these two different terms in the symbolic setting.

By looking at the compiler's definition (more precisely at the first condition that the constructed transaction needs to satisfy), we can see that if $\confF$ has $w = \star$, then all the inputs of $\txT[\confF]$ will belong to the image of $\txMap$, meaning that they all correspond to some symbolic structure (deposits in general, and a contract if $\confF$ is a continuation advertisement).
Instead if $\confF$ has $w \not= \star$ then there is at least one input that does not belong to $\ran{\txMap}$. In particular, if $w=0$ this means that those additional inputs have all value 0.
This difference can be used to justify the fact that if the symbolic strategy of an honest participants may only choose an action that produces or consumes a symbolic advertisement $\confF$ if the term $w$ inside of $\confF$ is equal to $\star$ (see the third condition of Definition \ref{def:symbolic-participant-strategies}).

This requirement is tied to the assumption that whenever a honest participant proposes an action, they want to be sure that the currency it is funded with can actually be spent. In this regard, deposit outputs do not pose any problem: from the structure of the script a participant knows that the output only needs the owner's signature in order to be spent.
Instead, outputs that do not have a symbolic counterpart may have an irredeemable script, that prevents them from being spent. In the computational model it might actually be very difficult to confirm if that is indeed the case, but in the symbolic context it is outright impossible. So, we simply assume that honest participants will ignore these funds when proposing a contract.

\section{Adversary model}
\label{sec:app-adversary-model}

We will now formalize the adversary model for \langname, defining symbolic runs, strategies and conformance.
We denote with $\PartT\subseteq \Part$ the set of honest participants, and with $\Adv\notin \Part$ the \emph{symbolic adversary}. We will assume that the adversary is able to control the choices of all dishonest participants. 

\paragraph*{Randomness in symbolic strategies}
As we will soon see, in we will model the choices of participants as probabilistic algorithms, which we construct by giving as considering a deterministic algorithm that takes a random sequence of bits as an additional input.
When defining strategies, we will assume that each honest participant  $\pmvA$ will have access to their own seed $\nonce{\pmvA}$, and that the adversary has access to $\nonce{\Adv}$.
We define a function $\rndR$ called \emph{randomness source} that associates to each participant (and to the adversary) their random seed.
Once the randomness source is assigned, every probabilistic algorithm can be seen as a deterministic one.
With a slight notational abuse, the random seed will be listed among the algorithms' inputs only when we explicitly need it.
\begin{defn}[Symbolic run]
	A \emph{symbolic run} $\runS$ is a (possibly infinite) sequence of configurations $\confG[i]$ connected by transition labels $\labS[i]$. The first configuration in the sequence is called \emph{initial configuration} and it is  $\confG[0]= \confD[dep] \mid t_0 \mid \confTaint{0}$, where $\confD[dep]$ only contains deposits and $t_0=0$.
	If $\runS$ is finite we denote with $\confG[\runS]$ its last configuration. A run is written as $\runS = \confG[0] \xrightarrow{ \labS[0]} \confG[1] \xrightarrow{\labS[1]} \cdots$. Without loss of generality, we will assume that if an action $\labS[i]$ introduces a new symbolic name, then that name is never used: not only in $\confG[i]$ (as the semantics requires), but also in all the other previous configurations.
\end{defn}
\begin{defn}[Symbolic strategies] \label{def:symbolic-participant-strategies}
	A \emph{symbolic strategy} $\stratS{\pmvA}$ is a probabilistic polynomial time  algorithm that takes as input the  symbolic run and outputs a finite set of transition labels $\labS$, representing the actions that $\pmvA$ wants to perform in order to advance the run.
	The output of $\stratS{\pmvA}$ (i.e. the set of $\pmvA$'s choices) is subject to the following constraints: 
	\begin{enumerate}
		\item $\pmvA$ must chose labels that are enabled by the semantics.
		Formally this is expressed by saying that if $\labS \in \stratS{\pmvA}(\runS)$,  then it exists $\confG$  such that  $\confG[\runS] \xrightarrow{\labS} \confG$.
		\item $\pmvA$ can not impersonate a different participant $\pmvB$ and forge their authorizations.
		In order to express this formally, we first define $\mathcal{A}_{\pmvB}$, the set of labels that express authorizations given by participant $\pmvB$, as 
		\begin{align*}
			\mathcal{A}_{\pmvB} = \{ \ &auth-in(\pmvB, \cdot, \cdot), \ auth-act(\pmvB, \cdot, \cdot), 
			\\
			& auth-join(\pmvB,\cdot,\cdot,\cdot), \ auth-divide(\pmvB, \cdot, \cdot,\cdot),
			\\  &auth-donate(\pmvB, \cdot,\cdot) \ \}.
		\end{align*}
		Then, we have that if
		$\labS \in \stratS{\pmvA}(\runS) \cap \mathcal{A}_{\pmvB}, \text{ then } \pmvB = \pmvA $.
		\item $\pmvA$ can only choose $adv(\confF)$ if the term $w$ inside $\confF$ is equal to $\star$. Similarly, $\pmvA$ can only choose an $init$, $call$,$send$, or $destroy$ action only if the consumed term $\confF$ has $w = \star$.
		\item The strategy is \emph{persistent}, meaning that if at a certain point $\pmvA$ chooses an action $\labS$ (that is not a delay), then at the next step of the run that same action $\labS$ must be chosen again, if it is still enabled. Formally we can say that if $ \labS \in \stratS{\pmvA}(\runS )$, $\labS \not = delay(\delta)$, $\runS \xrightarrow{\labSi} \runSi$, and it exists $\confG$ such that $ \confG[\runSi] \xrightarrow{\labS} \confG$;  then $\labS \in \stratS{\pmvA}(\runSi)$
	\end{enumerate} 		
\end{defn}
\begin{defn}[Symbolic adversarial strategy]
	An adversarial strategy $\stratS{\Adv}$ is a probabilistic polynomial algorithm that takes as input the run $\runS$, together with the set of transition labels $\LabS_i$ chosen by each honest participant $\pmvA[i]$.
	The strategy returns as output a single label $\labAdv$ that will be used to update the symbolic run. If $ \labAdv= \stratS{\Adv}(\runS,\vec{\LabS})$, then one of the following cases holds:
	\begin{enumerate}
		\item $\labAdv$ is neither an authorization nor a delay, and it is enabled by the semantics.
		Formally we can say that $\labAdv=\labS$ where: (i) $\labS \not= delay(\delta)$; (ii) there is no $\pmvB$ for which $\labS\in \authSet{\pmvB}$; (iii) there exists $\confG$ such that  $\confG[\runS] \xrightarrow{\labS} \confG$.
		\item $\labAdv$ is an authorization given by a honest participant, and it is chosen by the strategy of the corresponding participant.
		In short we can say that if $\labAdv \in \authSet{\pmvA[i]}$ then we also must have $\labAdv \in \LabS_i$.
		\item  $\labAdv$ is an authorization given by a dishonest participant $\pmvB \not \in \PartT = \setenum{\pmvA[1],\dots, \pmvA[n]}$ and it is enabled by the semantics.
		\item  $\labAdv$ is a delay and it is chosen by the strategies of all honest participants. Symbolically this is $\labAdv = delay(\delta)$ and $\forall i \in  \setenum{1\dots k}.$ $(\LabS_i = \emptyset$ or $ delay(\delta_i) \in \LabS_i$ with  $\delta_i \geq \delta )$
	\end{enumerate}
\end{defn}
\begin{defn}[Symbolic conformance]
	Given a random source $\rndR$ and a set of strategies $\stratSSet$ that includes those of honest participants $\pmvA[1],\dots, \pmvA[k]$ and the adversarial strategy $\stratS{\Adv}$, it is possible to uniquely determine a run $\runS$. We say that this run \emph{conforms} to the pair $(\stratSSet, r)$.
	More precisely, the conformance relations between $\runS$ and $(\stratSSet, r)$ holds if and only if one of the two following conditions is verified
	\begin{enumerate}
		\item $\runS = \confG[0]$ where $\confG[0] = (\mmid_j \confDep[x_j]{\pmvA[j]}{\valV[j]})  \mid 0\mid \confTaint{0}$  is an initial configuration.
		\item $\runSi\xrightarrow{\labS} \runS$ where $\runSi$ conforms to $(\stratSSet, r)$ and, given $ \forall i. \ \LabS_i = \stratS{\pmvA[i]}(\runSi, r_{\pmvA[i]})$, we have $\stratS{\Adv}(\runSi, \vec{\LabS}, \nonce{\Adv})= \labS$
	\end{enumerate}
	If we are considering a set $\stratSSet$ that does not include an adversarial strategy, then we say that a run $\runS$ conforms to $(\stratSSet, r)$ if there exists an adversarial strategy $\stratS{\Adv}$ such that $\runS$ conforms to $(\setenum{\stratS{\Adv}} \cup \stratSSet, r)$.
\end{defn}

\paragraph*{Computational adversary model}
Below, we model adversaries at the computational level. We will follow a structure similar to what we did above for the symbolic model.

\paragraph*{Randomness and keys} 
The choices of participants will again be modelled as probabilistic algorithms. For this reason, we provide a randomness source $\nonce{\pmvA}$ to each honest participant, and one to the adversary.
In the computational model, these randomness sources are not only passed as an additional input to the strategies in order to make then a probabilistic, but they are also used to generate the participants' keys. 
Given a security parameter $\eta$, we have that each honest participant will use the first $\eta$ bits of their random sequence to produce an asymmetric key pair, $\key{\pmvA}{\nonce{\pmvA}}$. This key will be used to produce a witness by signing a transaction, whenever the script requires it. The public and secret part of the key $\key{}{}$ are denoted respectively with $\pubkey{}{}$ and $\privkey{}{}$.
Similarly, the adversary will use their random source $\nonce{\Adv}$ to generate the keys of all dishonest participant, using $\eta$ bits for each key pair.
With a slight notational abuse, the random seed will not be listed among the inputs of the algorithms that use it, unless it is explicitly needed. 

\begin{defn}[Computational run]
	A \emph{computational run} $\runC$ is a (possibly infinite) sequence of labels $\labC$, each encoding one of these possible actions.
	\begin{align*}
		& \txT
		&& \text{appending transaction $\txT$ to the blockchain}
		\\
		& \delta
		&& \text{performing a delay}
		\\
		& \pmvA \rightarrow \ast:m
		&& \text{broadcasting of message $m$ from $\pmvA$}
	\end{align*} 	
	Every computational run starts with the initial transaction $\txT[0]$ that distributes a certain quantity of currency to each participant. We will assume that all of $\txT[0]$'s outputs are deposit outputs. Immediately after that, every participant broadcasts the public part of the  key that they have generated using their random source.
	So we have
	\begin{align*}
		\runC[0] = \txT[0]
		&\cdots  \pmvA[i] \rightarrow * : (\pubkey{\pmvA[i]}{\nonce{\pmvA[i]}}) \cdots
		\\
		&\cdots \pmvB[j] \rightarrow * : (\pubkey{\pmvB[j]}{\nonce{\Adv}})\cdots
	\end{align*}
	where  all outputs of $\txT[0]$ are standard, $\pmvA[i]$ are all the participants in $\PartT$, and $\pmvB[j]$ are all the others.
\end{defn}

\begin{defn}[Blockchain of computational runs]
	In order to keep track of the transactions that occur in $\runC$, while ignoring the messages, we define the \emph{blockchain of the computational run} $\bcB[\runC]$ as follows:
	\[
	\bcB[{\txT[0]}] = (\txT[0],0) \qquad \bcB[\runC \labC]= 
	\begin{cases}
		\bcB[\runC] (\txT, \delta_{\runC}) &\text{ if } \labC = \txT
		\\
		\bcB[\runC] & \text{otherwise}
	\end{cases} 
	\]
	where $\delta_{\runC}$ denotes the sum of the delays present in the run up until that point.
\end{defn}

We can use the notion of consistency for a blockchain, to define consistency for the runs.
\begin{defn}[Consistent computational run]
	A computational run $\runC$ is said to be \emph{consistent} if:
	\begin{enumerate}
		\item Its blockchain $\bcB[\runC]$ is consistent.
		\item If $\runC= \runCi \, \txT  \cdots $, then among the labels of $\runCi$ we can find, in this order: a message $\pmvB \rightarrow \ast : \txT$ that encodes the transaction sent after all $\txT$'s timelocks have expired\footnote{
			In order to formalize what is meant by \say{all $\txT$'s timelocks have expired}, assume the following: $(i)$ the $j$-th input of $\txT$ is an output of a transaction $\txT[j]$ which appears in $\bcB[\runC]$ at time $t_j$; $(ii)$ the sum of delays from the beginning of the run up until the sending of the message is $t$ ; $(iii)$ the transaction has relative timelocks $\txT.\txAfterRel[j]{\delta_j}$ and an absolute timelock $\txT.\txAfterAbs[]{t_0}$.
			We say that all $\txT$'s timelocks have expired if $t\geq t_0$ and $\forall j. \ t-t_j \geq \delta_j$.};
		and all messages $\pmvB \rightarrow \ast: (\txT, j, wit,i)$ where $wit$ is the $i$-th witness of the $j$-th input of $\txT$.
		These messages may be sent from different participants.
	\end{enumerate} 
\end{defn}
The second condition required to have a consistent run amounts to saying that a transaction and its witnesses are broadcast before they are appended to the blockchain. This is a reasonable assumption: the blockchain is public, so in order to include a transaction one also has to broadcast it.
If $\runC \labC$ is a consistent run, then we say that $\labC$ is a \emph{consistent update} to $\runC$. Note that delays and messages are always consistent updates to a run.
We now present the \emph{computational strategies}, the way in which participants can interact with the run, updating it.
\begin{defn}[Computational strategies]
	\label{def:computational-participant-strategies}
	A \emph{computational strategy} for an honest participant $\pmvA$ is a probabilistic polynomial algorithm $\stratC{\pmvA}$, that takes as input the computational run $\runC$ and outputs a finite set of computational labels $\LabC$, whose elements $\labC$ consistently update $\runC$.
	Moreover, if $\labC$ is a message, then it is sent from $\pmvA$.
	We require strategies to be persistent: if $\labC \in \stratC{\pmvA}(\runC)$ is not a delay, and $\dot{\labC} \not= \labC$ is such that both $\runC \dot{\labC}$ and $\runC \dot{\labC} \labC$ are consistent, then we must have $\labC \in \stratC{\pmvA}(\runC\dot{\labC})$.
\end{defn}

\begin{defn}[Computational adversarial strategies]
	The \emph{computational adversarial strategy} is a polynomial algorithm $\stratC{\Adv}$ that takes as input the computational run and the set of choices taken by each honest participant, and returns as output a single computational label used to update the run.
	If $\labC = \stratC{\Adv}(\runC, \vec{\LabC})$ then $\labC$ is a consistent update to $\runC$, with the following constraint:
	if $\labC= \delta$ is a delay, then it must have been chosen by all the honest participants' strategies. Formally, we can express this as  $\forall i \in \setenum{1,\dots,k}.$ $(\LabC_i = \emptyset$ or $ \delta_i \in \LabC_i$ with $\delta_i \geq \delta)$.
\end{defn}
Notice that we are allowing the adversary to impersonate every participant $\pmvB$, by introducing in the run messages $\pmvB \rightarrow * : m$. However, since they do not have direct access to $\pmvB$'s random source, they will be, with overwhelming probability, unable to forge $\pmvB$'s signature, since they are restricted to using a polynomial time strategy. 
\begin{defn}[Computational conformance]
	Take a randomness source $r$, and a set of strategies $\stratCSet$ containing those of all honest participants $\setenum{\pmvA[1],\cdots, \pmvA[k]}$ as well as the adversarial strategy. We say that a run $\runC$ \emph{conforms} to $(\stratCSet,r)$ if one of the two following conditions holds
	\begin{enumerate}
		\item $\runC$ is initial, with keys derived from the randomness source $r$.
		\item $\runC = \runCi \labC$ with $\runCi$ conforming to $(\stratCSet,r)$, and, given $\LabC_i = \stratC{\pmvA[i]}(\runCi, r_{\pmvA[i]}) $, we have $\stratC{\Adv}(\runCi,\vec{\LabC}, r_{\Adv})= \labC$. 
	\end{enumerate}
\end{defn}

\section{Coherence}
\label{sec:app-coherence}

In this appendix we formally define the \emph{coherence} relation between symbolic and computational runs, and we prove some of its properties.

Before we begin with the definition, it is important to note that the coherence relation does not only involve the two runs, but also three auxiliary functions, $\txMapC$, $\keyMapC$, and $\advMapC$.
The first two are essentially the same functions that we used in the previous chapter to provide context to the compiler, and they are used to map symbolic names to transaction outputs; to map participants to their public deposit's key.
The third function $\advMapC$ maps a symbolic advertisement term $\confF$ to a transaction $\txT[\confF]$ that encodes it. It will be used to keep track of previous advertisements to avoid repeating them.

\begin{defn}[Coherence]\label{def:coherence}
  The definition of the coherence relation 
  \[
  \coher{\runS}{\runC}{\txMapC}{\keyMapC}{\advMapC}
  \] 
  is split into three sections, each one consisting of one or more inductive cases.
  \paragraph*{Base case}
  The relation:
  \[
  \coher{\runS}{\runC}{\txMapC}{\keyMapC}{\advMapC}
  \]
  holds if the following conditions are verified:
  \begin{enumerate}
  \item $\runS = \confG[0]$ it is an initial configuration; 
  \item $\keyMapC$ maps each owner of a deposit in $\runC$ to a public key; 
  \item $\runC $ is initial, and the keys broadcast after the first transaction are the ones in the image of $\keyMapC$;
  \item $\txMapC$ maps exactly the name $x$ of each deposit $\confDep[x]{\pmvA}{\valV}$ in $\confG[0]$ to a different deposit output of $\txT[0]$ of value $\valV$ owned by $\pmvA$%
  \footnote{Note that we may not have a unique way to determine $\txMapC$: for example this happens if in the initial configuration a participant owns two deposits of the same value. We do not give a detailed explanation on how to handle this edge cases: the key fact here is that even in those case it is still possible to construct $\txMapC$ so that it is injective.};
  \item $\advMapC$ is the empty function.
  \end{enumerate}
  
  \paragraph*{Inductive case 1}
  The relation:
  \[
  \coher{\runS \xrightarrow{\labS} \confG}{ \runC \labC}{\txMapC}{\keyMapC}{\advMapC}
  \]
  holds if $\coher{\runS}{\runC}{\txMapCi}{\keyMapCi}{\advMapCi}$ holds, in addition to one of the following conditions.
  
  In all cases where explicit changes are not mentioned we will assume that $\txMapC = \txMapCi$, $\keyMapC = \keyMapCi$, and $\advMapC= \advMapCi$.
  \begin{enumerate}
  \item $\labS = msg(\Theta)$, $\labC = \pmvA \rightarrow \ast: \mathcal{C}$,
    where $\Theta$ is an incomplete advertisement and $\mathcal{C}$ is a bitstring that encodes it. In $\mathcal{C}$ every deposit name $z_j$ is represented by the transaction output $\txMapC(z_j)$.
    
  \item $\labS = adv(\confF)$,  $\labC = \pmvA \rightarrow \ast:\txT[\confF]$, where
    $\confF= \confAdvInitN{\clauseCallXParam}{\vec{z}}{w}{h} $ is a valid initial advertisement in $\runS$ and
    \[\txT[\confF] = \compile{\confF, \txMapCi,\keyMapCi, \inCmp, \tAbsCmp,\tRelCmp,\nonceCmp}\]
    for some $\inCmp$, $\tAbsCmp$, $\tRelCmp$ and $\nonceCmp$.
    Moreover, we require this to be is the first time that $\labC$ appears in $\runC$ after all timelocks in $\txT[\confF]$ are expired. The function
    $\advMapC$ extends $\advMapCi$ by mapping $\confF$ to $\txT[\confF]$.
    
  \item $\labS = adv(\confF)$,  $\labC = \pmvA \rightarrow \ast:\txT[\confF]$, where
    $\confF=\confAdvN{\complD}{\vec{z}}{w}{x}{j}{h}$ is a valid continuation advertisement in $\confG[\runS] $ and
    \[\txT[\confF] = \compile{\confF, \txMapCi,\keyMapCi, \inCmp, \tAbsCmp,\tRelCmp ,\nonceCmp}\]
    for some $\inCmp$, $\tAbsCmp$,$\tRelCmp$, and $\nonceCmp$.
    We have the same additional restriction of $(2)$, and again $\advMapC$ extends $\advMapCi$ by mapping $\confF$ to $\txT[\confF]$.
    
  \item $\labS = init(\confF,x)$, $\labC = \txT[\confF]$, where $\txT[\confF]= \advMapCi(\confF)$. In the symbolic setting this step produces $\confContr[x]{\contrC}{\valV}$, and so we extend $\txMapCi$ to $\txMapC$ by mapping $x$ to the single output of $\txT[\confF]$.
    
  \item $\labS = call(\confF)$, $\labC = \txT[\confF]$, where $\txT[\confF]= \advMapCi(\confF)$.
    With this action the symbolic run produces the active contracts
    $\confContrTime[y_1]{t}{\contrC[1]}{\valV[1]} \cdots \confContrTime[y_k]{t}{\contrC[k]}{\valV[k]}$, so we extend $\txMapCi$ to $\txMapC$ by mapping each $y_i$ to the $i$-th output of $\txT[\confF]$.
    
  \item $\labS = send(\confF)$, $\labC = \txT[\confF]$,  where $\txT[\confF]= \advMapCi(\confF)$. This time the symbolic action doesn't produce any contract, but deposits $\confDep[y_1]{\pmvA[1]}{\valV[1]}$, $\cdots$, $ \confDep[y_k]{\pmvA[k]}{\valV[k]}$. For this reason $\txMapCi$ is  extended by mapping $\txMapC(y_i)$ to be equal to the $i$-th output of $\txT[\confF]$.
    
  \item  $\labS = auth-action(\pmvA,\confF)$, $\labC = \pmvB \rightarrow \ast: m$, where $m$ is a quadruple $(\txT[\confF],1, wit, i)$.
    We have that $wit$ is a the signature with $\pmvA$'s key on the first input of $\txT[\confF]= \advMapCi(\confF)$. The index $i$ is the position of $wit$ as witness, and can be deduced by looking at the script of  $\advMapCi(\confF)$.
    Moreover, we want $\labC$ to be is the first instance of this signature being sent after a broadcast of $\txT[\confF]$ that, in turn, appears after all of $\txT[\confF]$'s timelocks have expired.
    
  \item  $\labS = auth-in(\pmvA,z,\confF)$, $\labC = \pmvB \rightarrow \ast: (m,i)$, where $m$ is a quadruple $(\txT[\confF],j, wit, 1)$,  with being a signature with $\pmvA$'s key on the $j$-th input of the transaction $\txT[\confF]= \advMapCi(\confF)$.
    As in (8), we want this to appear in the run after a message encoding $\txT[\confF]$ (which, again, has been sent after all $\txT[\confF]$'s timelocks have expired).
    Moreover we ask that the $j$-th input is exactly $\txMapCi(z)$, or equivalently that $z$ is the $j$-th deposit specified in $\confF$.
    Differently from (8) this case may happen even with advertisements of any form.
    
  \item $\labS = delay(\delta)$, $\labC = \delta$.
    Symbolic delays trivially translate to computational ones, without changing the mapping between outputs and names. 
    This means that two coherent runs always have the same time.
    
  \item $\labS = auth-join(\pmvA,x,x',i)$,  $\labC = \pmvB \rightarrow \ast: m$, where $m$ is a quadruple $(\txT, i, wit, j)$ such that $wit$ signature on input $\txMapCi(x)$ for a previously broadcast transaction $\txT$ that takes two inputs ($\txMapCi(x)$ and $\txMapCi(x')$, in order) and has a single output that encodes a deposit $\confDep{\pmvA}{\valV+\valVi}$.
    Moreover $\labC$ is the first instance of such signature being broadcast in the computational run (from the broadcast of $\txT$ onward).

  \item $\alpha = join(x,x')$, $\labC = \txT$, where $\txT$ has exactly two inputs given by $\txMapCi(x)$ and $\txMapCi(x')$, of value $v$ and $v'$ respectively, and a single deposit output owned by $\pmvA$ of value $\valV + \valVi$. In the symbolic run the action $\alpha$ consumes the two deposits $x$ and $x'$ in order to produce $\confDep[y]{\pmvA}{\valV+\valVi}$.
    We extend $\txMapCi$ by mapping $\txMapC(y)$ to the output of $\txT$.
  \item $\alpha = auth-divide(\pmvA,x,\valV, \valVi)$, continues like (10).
  \item $\alpha = divide(x,\valV,\valVi)$, continues like  (11)
  \item $\alpha = auth-donate(\pmvA,x,\pmvB)$, continues like (10)
  \item $\alpha = donate(x,\pmvB)$, continues like (11)
    
  \item $\labS = adv(\confF)$,  $\labC = \pmvA \rightarrow \ast:\txT$, where $\confF$ is a destroy advertisement  $\confAdvDesN{\vec{z}}{w}{h}$; the transaction $\txT$ is not compatible with either one of the two previous $adv()$ cases, nor it represents a join, divide or donate operation.
    Moreover, among the inputs of $\txT$ there appear (in order) the outputs $\txMap(z_j)$, for all $j = 1 \cdots |\vec{z}|$. If $w=\star$ then these are all the inputs, otherwise removing those leaves a list of inputs outside of $\ran{\txMapCi}$, such that the sum of their values is $w$.
    Lastly, we have the same additional restrictions of $(2)$, and $\advMapC$ extends $\advMapCi$ by mapping $\confF$ to $\txT$.
    
  \item $\alpha = destroy(\confF)$, $\labC = \txT$, where $\txT= \advMapCi(\confF)$. Notice that from the advertisement of $\confF$ we know that $\labC$ cannot correspond to any of the already mentioned cases.
  \end{enumerate}

  \paragraph*{Inductive case 2}
  The predicate:
  \[
  \coher{\runS}{ \runC \labC}{\txMapC}{\keyMapC}{\advMapC}
  \]
  holds if $\coher{\runS}{\runC}{\txMapC}{\keyMapC}{\advMapC}$ holds, in addition to one of the following conditions:
  \begin{enumerate}
  \item $\labC = \txT$ with no input of $\txT$ belonging to the image of $\txMapC$.
  \item $\labC = \pmvA \rightarrow \ast : m$ that does not correspond to any of the symbolic moves described in the first inductive case of the definition.
  \end{enumerate}
\end{defn}

Now that we have formalized the concept of coherence, we can establish some results about the relation between the two models.
Notice that, for most of the following propositions, a rigorous proof by induction would need to examine all 20 inductive rules appearing in the definition.
For the sake of brevity we will focus only on the cases that are relevant to each proof.
In most cases we will only be interested in the runs and the $\txMapC$ map, so we will write $\coherRel{\runS}{\runC}{\txMapC}$.

We start with a lemma regarding the $\txMapC$ map.
\begin{lem} \label{lem:coher-txout-injective}
  Assuming $\coher{\runS}{\runC}{\txMapC}{\keyMapC}{\advMapC}$ holds, the map $\txMapC$ is injective.
\end{lem}

\begin{proof}[Proof]
  By induction.
  In the base case we are mapping every deposit to a different output of $\txT[0]$, so $\txMapC$ is injective.
  In all of the inductive cases we can assume the injectivity of $\txMapCi$, and we will need to check that $\txMapC$ remains injective. Obviously, we will only need to look at the cases that have $\txMapC \not= \txMapCi$.
  
  For this reason, we are only concerned with cases corresponding to $init$, $call$, $send$, $join$, $divide$, and $donate$ actions. 
  All of these create new symbolic names, which are then mapped to outputs of a newly created transaction that is appended to the blockchain in that very step.
  This means that all these outputs are different from all the ones in $\ran{\txMapC}$ (since the transaction they belong to was not present on the blockchain in previous steps).
  This, in conjunction with the fact that each new name is mapped to a different output of the new transaction, proves that the new map $\txMapC$ is still injective.
  %Lastly, we can notice that the second inductive case does not modify $\txMapC$, meaning that injectivity is trivially preserved.
\end{proof}

Next we have a proposition that clarifies the relationship between unspent outputs in the computational model's blockchain and active contract or deposit in the symbolic configuration. The proposition will also shows that if coherence holds, then the $\txMapC$ map does what it is intuitively expected to do: it associates deposits and contracts to transaction outputs of the same value.

\begin{prop}
  \label{prop:coher-output-correspondence}
  Assume that:
  \[
  \coher{\runS}{\runC}{\txMapC}{\keyMapC}{\advMapC}
  \]
  and let $x$ be the name of a deposit or a contract in the symbolic run $\runS$. 
  If $x$ is the name of a deposit or of an active contract in $\confG[\runS]$, then the output $\txMapC(x)$ is unspent in $\bcB[\runC]$.
  If instead $x$ is the name of a deposit or contract in $\runS$, appearing in some previous configuration but not in $\confG[\runS]$, then the output $\txMapC(x)$ is spent $\bcB[\runC]$.
  Moreover, the map $\txMapC$ preserves the value, so that deposits (resp. active contracts) of value (resp. balance) $v$ are always mapped to outputs of value $v$.
\end{prop}
\begin{proof}
  By looking at the previous proof, we can see that once $\txMapC(x)$ is determined, it does not change. So, we just need to prove two statements: 
  $(i)$ whenever a new contract or deposit is inserted in the configuration, the domain of $\txMapC$ is extended, and the image of that new name is a newly created output (which is obviously unspent) of correct value; 
  $(ii)$ a UTXO belonging to $\ran{\txMapC}$ is spent in the computational setting if and only if its pre-image is consumed in the symbolic run.
  
  To prove $(i)$ notice that in Definition \ref{def:coherence} the cases in which the symbolic action creates a new deposit or a new contract are exactly the cases in which the domain of $\txMapC$ is extended. These all happen in the first set of inductive cases, in the items related to the following operations: $send$, $join$, $divide$ and $donate$ (for deposits), $init$ and $call$ (for contracts).
  We can immediately see from the coherence definition that all these cases
  append a transaction $\txT$ to the computational run. Moreover, this $\txT$ has the correct number of outputs; and $\txMapC$ is updated accordingly. 
  The fact that the output's value matches the symbolic value is ensured either by a direct specification in the coherence definition (for $join$, $divide$, and $donate$ operations), or by the script's covenant \footnote{Later, in Proposition \ref{prop:coher-encoding-outputs} we will give a more profound justification on why the covenant really forces the value to be correct} in cases that append a compiler generated transaction (which happens in $send$, $init$ and $call$ operations).
  
  To prove the second condition we need to only check the inductive steps that spend a UTXO associated to some symbolic name, or the ones that remove some symbolic name from the configuration. 
  By looking at the definition of coherence we can see that these two cases coincide. They happens only in the first set of inductive cases, and specifically in the cases related to $init$, $call$, $send$, $join$, $divide$, $donate$ and $destroy$ operations.
  The semantic transition rule for each of those actions consumes some deposit or contract, and we can see (either thanks to a direct specification in the coherence definition, or to the definition of the compiler) that the corresponding transaction appended to the computational run always spends the corresponding UTXOs. 
\end{proof}

The next result further refines the above proposition, by showing that we can always determine the structure of an output associated to a symbolic term.
It also serves as a justification for the definitions of deposit output and of output encoding a contract that were given in the previous sections.

\begin{prop} \label{prop:coher-encoding-outputs}
  Assume that:
  \[
  \coher{\runS}{\runC}{\txMapC}{\keyMapC}{\advMapC}
  \]
  If an output in $\runC$ is the image of a symbolic name, then we can fully determine its structure
  \begin{itemize}
  \item If $\confDep[x]{\pmvA}{v}$ is in $\runS$, then $\txMapC(x)$ is a deposit output owned by $\pmvA$.
  \item If $\confContrTime[x]{t}{\contrC}{v}$ is in $\runS$, then $\txMapC(x)$ is the output of a compiler generated transaction.
    Moreover $\txMapC(x)$ is a contract output encoding $\contrC$.
    \begin{comment} Moreover $\txMapC(x)$'s third argument is a clause name $\cVarX$, and the following property holds: if $\cVarX$ takes $k$ static and $h$ dynamic variables, then $\txMapC(x)$ has a total of $3+k+h$ arguments, and we have that $\clauseDefNonVector{\cVarX}{ \txarg[4] \textup{\texttt{,}} \cdots \textup{\texttt{,}} \txarg[3+k]}{ \txarg[4+k] \textup{\texttt{,}} \cdots \textup{\texttt{,}} \txarg[3+k+h]}\equiv \clauseAdv{v}{}{\contrC}$.
    \end{comment}
  \end{itemize} 
\end{prop}
\begin{proof}
  Again, we proceed by induction.
  In the base case there is no contract in the configuration, and the symbolic deposits are all mapped by $\txMapC$ to deposit outputs of $\txT[0]$, so both statements hold.
  Among the inductive cases we only need to check those that introduce a deposit or a contract in the symbolic configuration, and hence extend $\txMapCi$.
  We start with the deposit operations $join$, $divide$, and $donate$. In those cases, the definition of coherence explicitly states that $\txMapCi$ is extended by mapping the newly created deposit to a deposit output, with the correct value and owner.
  
  This only leaves us with the $send(\confF)$ case for deposits, and the  $init(\confF)$ and $call(\confF)$ cases for contracts.
  In those cases the inductive step adds the transaction $\txT[\confF] = \advMapCi(\confF)$ to the computational, and extends the  $\txMapC$ relation by mapping the newly created symbolic terms (deposits or contracts) to $\txT[\confF]$'s outputs.
  
  Notice that if $\confF$ is a valid initial or continuation advertisement term in the symbolic run, then the transaction $\advMapCi(\confF)$ must be compiler generated. 
  This fact can easily be proved by induction on the coherence definition: in the base case $\advMapCi$ is the empty map, and the only inductive cases that we need to check are the ones that modify $\advMapC$, extending its domain to a new initial or continuation advertisement $\confF$. Those cases are the one corresponding to a symbolic $adv(\confF)$ operation, and the conditions that they need to satisfy directly imply that $\advMapC(\confF)$ must be compiler generated.
  
  But now, if we look at how the compiler generates the output(s) of $\txT[\confF]$, we see that if $\confF=\confAdvN{\complD}{\vec{z}}{w}{x}{j}{h}$, with $\complD$ ending in $\withdrawC{(\splitB{\pmvA[1]}{v_1}, \cdots \splitB{\pmvA[n]}{v_n})}$, then (by the compiler definition) the $i$-th output of $\txT[\confF]$ will be a deposit output of value $v_i$ owned by $\pmvA[i]$, and (by definition of coherence) it will be the image under $\txMapC$ of $i$-th deposit created by $send(\confF)$, which proves our claim for this inductive case.
  If instead $\confF = \confAdvN{\complD}{\vec{z}}{w}{x}{j}{h}$ with $\complD$ ending in $\ncadv{(\clauseCallX[1]{a^1}{b^1}, \cdots , \clauseCallX[n]{a^n}{b^n})}$, with $\clauseCallX[i]{a^i}{b^i} \equiv \clauseAdv{v_i}{}{\contrC[i]}$; then (by the compiler definition) the $i$-th output of $\txT[\confF]$ will be an output of value $v_i$ encoding the contract $\contrC[i]$, and (by definition of coherence) it will be the image under $\txMapC$ of $i$-th contract created by $call(\confF)$,  which proves our claim for this inductive case.
  The initial advertisement case is identical to the $call(\confF)$ seen above, with only one output.
\end{proof}

The above propositions state that it is possible to \say{keep track} of symbolic deposits and active contracts by seeing them as computational outputs. However, there is another term in the symbolic configuration that is used to store an amount of currency usable by the participants: the destroyed fund counter $\confTaint{}$.
In the following proposition we will show how the counter approximates from above the amount of currency contained in output that do not correspond to any other symbolic term.

\begin{prop}
  \label{prop:coher-destroyed-funds}
  Assume that:
  \[
  \coher{\runS}{\runC}{\txMapC}{\keyMapC}{\advMapC}
  \]
  The sum of all values stored in transaction outputs that do not belong to $\ran{\txMapC}$ is smaller or equal to the value $w$ specified by the destroyed fund counter $\confTaint{w}$ in $\confG[\runS]$.
\end{prop}
\begin{proof}
  By induction.
  In the base case all outputs in $\runC$ are images of symbolic deposits, and $\confG[0]$ contains $\confTaint{0}$, so this proposition holds.
  
  In the definition of coherence there are two sets of inductive cases: in either of them, we will look at the inductive premise and denote with $w_0$ the amount of funds contained in the destroyed funds counter present in the last configuration of the symbolic run, and with $W_0$ the sum of the values of all unspent transaction outputs in the computational blockchain that do not correspond to symbolic deposits or contracts. This means that our inductive hypothesis states that $W_0 \leq w_0$, and, after having defined $W_1$ and $w_1$ in a similar way, we want to prove that $W_1 \leq w_1$.
  
  While looking at the items in the first set of inductive cases in Definition \ref{def:coherence}, we only care about proving our proposition in the cases that either modify the counter with a symbolic action, or insert in the computational blockchain a transaction that spends or produces some inputs outside of $\ran{\txMapC}$.
  
  When the symbolic action is $init(\confF)$, $call(\confF)$, or $send(\confF)$, the corresponding transaction $\txT[\confF]$ may spend some inputs that do not have a symbolic correspondent. Since $\txT[\confF]$ must be compiler generated, we know that the total value of these inputs must amount exactly to the value $w$ that appears in $\confF$ (or to 0 if $w= \star$).
  By spending these outputs we decrease $W_0$ to $W_1 = W_0 - w$.
  The semantics of these actions tells us that, in the symbolic run, the amount $w$ is removed from the counter, giving us $\confTaint{w_1}$ with $w_1= w_0 -w$. This means that the inequality $W_1 \leq w_1$ holds.
  
  Next, we need to check what happens when the symbolic action is $destroy(\confAdvDesN{\vec{z}}{w}{h})$. 
  The semantics of $destroy$ tells us that the value $w_0$ in the counter is increased by $\sum_j u_j$, where each $u_j$ is the value contained in the deposit $z_j$. 
  In the computational run instead we are creating a transaction $\txT$ which spends $w$ from inputs without a symbolic counterpart. Let us denote with $w'$ the sum of the values of $\txT$'s outputs. This value must be smaller or equal to the sum of $\txT$'s input values, which is $\sum_j u_j + w$. Moreover, notice that none of $\txT$'s output will have a symbolic counterpart. For this reason, the total the value stored by outputs without a symbolic counterpart becomes $W_1 = W_0 - w + w'$, and since we have $w' \leq \sum_j u_j + w$, this gives us $ W_1 \leq W_0 +\sum_j u_j$. In turn, this implies $W_1 \leq w_1 = w_0 + \sum_j u_j$, and the inequality is preserved.
  
  Lastly, we look at the second set of inductive definitions, where the first case tells us that we can insert any transaction $\txT$ whose input do not have a symbolic counterpart without performing any action on the symbolic run. $\txT$ may only reduce the total amount of funds stored in outputs that are outside of $\ran{\txMapC}$, since the sum of the values of its inputs must be greater or equal to the sum of the values of its outputs. This means that $W_1 \leq W_0$ while $w_1 = w_0$, and the inequality holds.
\end{proof}

\section{Correctness of the compiler}
\label{sec:app-correctness-of-compiler}
\ricnote{RILEGGERE PER TYPO}

In our implementation of $\langname$, the logic of contracts is only enforced through the script of a compiler-generated transaction. 
By looking at how such scripts are constructed, it is intuitively obvious that they can be redeemed only by following the symbolic contract logic: in this section we will prove two theorems that justify more precisely why that is actually true, showing that the compiler correctly implements the language.
%We will then state two simpler \say{reverse} propositions, that show how to construct symbolic advertisements $\confF$ in a way that ensures the existence of a transaction compiled starting from $\confF$, that is consistent with the computational run.

From~\Cref{prop:coher-encoding-outputs} we know that if the two runs are coherent, each active contract corresponds to a transaction's output that encodes it. 
However, we want to make sure that the only transactions that are able to redeem an output that encodes an active contract are the compiler generated ones.
This is important because otherwise it would be really easy to \say{break} the coherence relation, by redeeming the balance of a contract with a transaction that is not compiler generated, and hence does not carry any meaning to the symbolic setting. 
The following theorem proves that this may never happen.

\begin{thm}
  \label{thm:coher-compiler-generated}
  Assume that:
  \[
  \coher{\runS}{\runC}{\txMapC}{\keyMapC}{\advMapC}
  \]
  and let $\confContrTime[x]{t}{\contrC}{v}$ be an active contract in $\confG[\runS]$.
  Take a transaction $\txT$ that consistently updates $\runC$, and has an input that redeems the UTXO $\txMapC(x)$. 
  Then $\txT$ is compiler generated (and possibly completed by including witnesses), meaning that there exists $\confF$, $\inCmp$, $\tAbsCmp$, $\tRelCmp$, $\nonceCmp$ such that 
  \[
  \txT = \compile{\confF, \txMapC, \keyMapC, \inCmp, \tAbsCmp, \tRelCmp, \nonceCmp}
  \]
  for some $\confF$, $\inCmp$, $\tAbsCmp$, $\tRelCmp$, $\nonceCmp$.
  Moreover, the input of $\txT$ that redeems $\txMapC(x)$ is the first, and 
  we have  $\confF = \confAdvN{\complD}{\vec{z}}{w}{x}{j}{h}$, for some values $\vec{z}$, $w$, $x$, $h$, and for $\complD$, $j$, such that $\complD \approx \contrD$ where $\contrD$ is the $j$-th branch of $\contrC$.
\end{thm}
\begin{proof}
  The proof is organized in two parts: first we will show how to construct the terms $\confF$, $\inCmp$, $\tAbsCmp$, $\tRelCmp$, $\nonceCmp$ starting from the fields of a transaction $\txT$ that redeems $\txMapC(x)$; then we will prove that using the constructed terms as inputs for the compiler yields exactly $\txT$.
  
  Constructing $\inCmp$ is easy: we just take $\inCmp[i]$ to be the $i$-th the input of $\txT$.
  Obviously, one of these inputs will be $\txMapC(x)$.
  By~\Cref{prop:coher-encoding-outputs} $\txMapC(x)$ is compiler generated, so we know the structure of its script.
  The first part of the script sets the condition $\inidx = 1$, which tells us that the output must be redeemed by an input in position 1. This means that $\inCmp[1] = \txMapC(x)$. The fact that this same instructions is present in every compiler generated transaction means that none of the other inputs of $\txT$ can be the image of a contract in the symbolic configuration. 
  This implies that all other inputs of $\txT$ are either outside of $\ran{\txMapC}$, or are the image of a deposit.
  We are now able to construct $\vec{z}$ and $w$: the first is constructed by taking the preimage of all inputs of $\txT$ that are in $(\ran{\txMapC} \setminus \setenum{\txMapC(x)})$, and the other is set to be the sum of the values of the inputs that are not in $\ran{\txMapC}$ (or it is set to $\star$ if there are no such inputs). 
  $\tAbsCmp$ is set to be equal to $\txT$'s absolute timelock, while every other $\tRelCmp[i]$ is set to the value of $\txT.\txAfterRel[i]{}{}$.
  We construct $\nonceCmp[k]$ by taking the first argument of the $k$-th output of $\txT$ (this is possible since, as we will show later in the proof, each output of $\txT$ has more than one argument).
  
  Choosing $\complD$ and $j$ requires a bit more work. Note that in the rest of the proof we will be referring to arguments by their name, instead of more precisely tracking their position. The paragraph \say{Constructing the outputs: arguments} of the previous appendix motivates why we are able to do so.  % Again, we will to use the fact that we know how $\txMapC(x)$'s script is structured.
  
  From Proposition \ref{prop:coher-encoding-outputs} we know that $\txMapC(x)$ encodes contract $\contrC$, which means that $\txT$ will have to satisfy $\scr_{\contrC}$. 
  This term is organized as a conditional check, so $\txT$ must satisfy one of its branches. We assume that the taken branch is the $k$-th:
  \[
  \mathsf{if}~{\mathsf{B_{k}}}~\mathsf{then}~{\scr_{\contrD[k]}},
  \]
  where $	\mathsf{B_k} $ is a shorthand for the expression
  \begin{align*}
    &\outlen{\rtx}= n_k \andE \rtxo{\txbranch}{1} = k \andE 
    \\ &\ \cdots \andE \rtxo{\txbranch}{n_k} = k.
  \end{align*}
  The value of $\txbranch$ argument (the second) will then be the same across all outputs of $\txT$. The value $j$, which represent the branch in the symbolic advertisement, will be set to be equal to the second argument of any output of $\txT$.
  
  We can now use the fact that $\txT$ must satisfy $ \scr_{\contrD[k]}$ in order to construct the last term, $\complD$. 
  We have two cases, since $\contrD[k]$ ends either in a $\withdrawname$ or in a $\callname$. Now that we know which is the branch that is being executed, we can easily inspect $\contrC$, to determine which of these two cases we are dealing with.
  If $\contrD[k]$ ends in a $\withdrawname$, then we set $\complD$ to be equal to $\contrD[k]$.
  If instead $\contrD[k]$ ends in $\ncadv{(\clauseCallX[1]{a^1}{?^1},\cdots, \clauseCallX[n]{a^n}{?^n})}$, then, in order to construct $\complD$, we need to \say{complete} it, assigning a value to the placeholders.
  The script $\scr_{\contrD[k]}$ specifies that the $i$-th output has $3 + |\vec{\alpha^i}| + |\vec{\beta^i}|$ arguments, where $\alpha^i_l$ and $\beta^i_h$ are the parameters in $\cVarX[i]$, the $i$-th called clause. 
  For this reason we are able to construct $\complD$ by filling the placeholders $\vec{?^i}$ with the last $|\vec{\beta^i}|$ arguments of the $i$-th output of $\txT$.
  
  At this point we have constructed the terms $\confF$,  $\inCmp$, $\tAbsCmp$, $\tRelCmp$, $\nonceCmp$, so we can pass them as inputs to the compiler and construct  
  \[
  \txTi = \compile{\confF,\txMapC, \keyMapC, \inCmp, \tAbsCmp, \tRelCmp, \nonceCmp}
  \] 
  It is easy to check that the conditions set by the compiler are satisfied, proving that $\txTi$ is a proper transaction and not $\bot$.
  \begin{enumerate}
  \item We already know that $\inCmp[1]$ is $\txMapC(x)$. The rest of the condition follows from the fact that $\inCmp$, $\vec{z}$ and $w$ have been constructed together, starting from $\txT$'s inputs.
  \item We know that $\tAbsCmp$ and $\tRelCmp[1]$ have been constructed from $\txT$'s timelocks. But these timelocks  must be greater than the value appearing in the $\afterName$ (and respectively $\afterRelName$) decorations of $\complD$, since the $\txMapC(x)$'s script specifies the conditions
    \begin{gather*}
      \scr_{\afterC{t}{\contrD}} = \afterAbs{\ctxo{t}}{\scr_{\contrD}} 
      \\
      \scr_{\afterRelC{\delta}{\contrD}} = \afterRel{\ctxo{\delta}}{\scr_{\contrD}}.
    \end{gather*}
  \end{enumerate}
  In order to conclude the proof, we now need to show that $\txT = \txTi$.
  \begin{enumerate}
  \item (Inputs and timelocks)
    
    The inputs and timelocks of $\txTi$ are determined by $\inCmp$, $\tAbsCmp$ and $\tRelCmp$. By constructions of the parameters, $\txTi$ must have the same inputs and timelocks of $\txT$.
  \item (Outputs - call)
    
    If  $\complD$ ends in $\ncadv{\clauseCallX[1]{a^1}{b^1}, \cdots, \clauseCallX[n]{a^n}{b^n}}$, with $\clauseCallX[i]{a^i}{b^i}\equiv \clauseAdv{v_i}{}{\contrC[i]}$, then we have the following
    \begin{enumerate}
    \item (Number of outputs)
      $\txTi$ must have $n$ outputs. The same happens for $\txT$, since there is an $\outlen{\rtx}=n$ condition specified by the script in the same $\mathsf{if}$ statement that checks the branch.
    \item (Arguments) 
      According to the compiler definition, the $i$-th output of $\txTi$ will have arguments $\txnonce= \nonceCmp[i]$, $\txname = \cVarX[i]$, $\txbranch = j$, $\txalpha[l]= a_l^i$, and $\txbeta[l]= b_l^i$.
      This coincides with the number of arguments of the $i$-th output of $\txT$, since the script $\scr_{\contrD[j]}$, which $\txT$ must satisfy,
      contains the term $\arglen{\rtxo{}{i}}= |\vec{\alpha^i}|+ |\vec{\beta^i}| + 3$.
      The value of $\nonceCmp[i]$ has been chosen to be exactly equal to the first element of the $i$-th output of $\txT$, and this is also the first argument of $\txTi$.
      The same reasoning holds for the last $|\vec{\beta}^i|$, which were used in the construction of the values $b_l^i$.
      Regarding the remaining argument we can see that the script of $\txMapC(x)$ forces each of them to have a precise value:
      the second (the branch argument) must be equal to $j$, the third (the name argument) must be equal to $\cVarX[i]$, and for all the other $m= |\vec{\alpha^i}|$ arguments we have the following constraint:
      \begin{align*} 
        & \rtxo{\txalpha[1]}{i}= \ctxo{a^i_1} \andE \cdots \andE
        \\
        & \rtxo{\txalpha[{m}]}{i}= \ctxo{a^i_{m}} 
      \end{align*}
      which appears in the last part of the script for a clause operation. Remember that the expression $\ctxo{a^i_l}$ is a shorthand for whatever combination of parameters have been used to specify the value assigned to the variable $\alpha^i_l$ in $\contrC$. However, we already know that these values must evaluate to $a_l^i$, since they are evaluated from the arguments of $\txMapC(x)$ (which we know to be compiler generated and encoding $\contrC$).
      This means that the arguments of each output of $\txT$ are the same to the one of the corresponding output of $\txTi$.
    \item (Value) The $i$-th output of $\txTi$ has value $v_i$. 
      In the script for a branch that contains a call operation, we have the following term
      \[
      \rtxo{\txval}{i} = \rtxo{ \sexp[i] }{i},
      \]
      where $\sexp[i]$ is the expression in the precondition of $\cVarX[j]$. We know that $\sexp[i]$ must evaluate to $v_i$, since we have $\clauseCallX[i]{a^i}{b^i}\equiv \clauseAdv{v_i}{}{\contrC[i]}$. So, since $\txT$ has to satisfy the script, the value of its $i$-th output must be $v_i$.
    \item (Script)
      The script of each output of $\txTi$ is the same of its first input, which is $\txMap(x)$.
      The script of $\txMap(x)$ contains the covenant $\verrec{i}$ that forces the $i$-th output of any transaction who redeems it to be equal to its own.
      From this we can conclude that the script of each output of $\txT$ coincides with the script of each output of $\txTi$ 
    \end{enumerate}
  \item (Outputs - send)
    If $\complD$ ends in a $\withdrawname$, then we can use a reasoning similar to the $\callname$ case to show that the outputs of $\txTi$ must be equal to the outputs of $\txT$. Actually, the situation is even simpler, since in this case the redeeming transaction must only have two arguments. However we will not delve into the details to avoid excessively lengthening this already long proof. 
  \end{enumerate}
  
\end{proof}

Essentially, we have just shown that any transaction that can redeem an output representing an active contract can be represented symbolically with a continuation advertisement term. This result plays a fundamental role in the proof of the computational soundness theorem.
We can take this correspondence between transactions and advertisements even further, by showing that if the computational transaction respects the timing conditions set in the coherence definition (in particular in item 3), then the corresponding symbolic advertisement is actually valid in the configuration.

\begin{thm}
  \label{thm:coher-redeeming-contract-implies-valid-adv}
  Under the same hypotheses of~\Cref{thm:coher-compiler-generated},
  let $\confF = \confAdvN{\complD}{\vec{z}}{w}{x}{j}{h} $ be the continuation advertisement constructed in the proof. Then $\confF$ is valid in $\confG[\runS]$.
\end{thm}
\begin{proof}
  We know that the deposits $z_j$ are the pre-image under $\txMapC$ of some  outputs in $\bcB[\runC]$. Moreover, these outputs are unspent so, by Proposition \ref{prop:coher-output-correspondence} they must actually appear in the configuration.
  Also, thanks to Propositon \ref{prop:coher-destroyed-funds}, and remembering how $w$ was constructed, we know that $w$ must be smaller or equal to the value stored in $\confTaint{}$.
  
  Then, we have the timing requirements: the time in the configuration must be so that all waiting decorations in $\complD$ are satisfied. In the computational setting all of $\txT$ timelocks are expired, and those same timelocks were subject to the script's constrictions, which in turn were based on the $\afterName$ and $\afterRelName$ decorations of $\complD$. Since the time increases in the same way in both models, the timing requirements are satisfied.
  
  Then, we have a condition which states that the sum of the \say{output} values of $\complD$ (meaning the funds of each clause if $\complD$ ends in $\callname$ and the values distributed to each participant if it ends in a $\withdrawname$) must be greater than 0 and lower or equal to the sum of the inputs values (the deposits $z_i$, the value $w$ and the balance of the contract $x$). Since these directly translate to inputs and outputs of $\txT$ we do not need to prove anything. 
  
  The last condition for validity applies only if $\complD$ ends in a $\callname$ operation: the proposition $p_i$ in each clause precondition must be satisfied. Again, the fact that $\txT$ must satisfy a compiler generated script is enough to prove this condition, since by including 
  \[
  \andE \rtxo{p}{1} \andE \cdots \andE \rtxo{p}{n}
  \]
  the script ensures that all clauses are satisfied.
\end{proof}

\section{Translating symbolic strategies} \label{section:translating-symbolic-strategies}
This appendix aims to construct an algorithmic map $\stratMap$ that transforms an honest symbolic strategy $\stratS{\pmvA}$ into a computational strategy $\stratC{\pmvA} = \stratMap(\stratS{\pmvA})$.
By Definition \ref{def:computational-participant-strategies} $\stratMap(\stratS{\pmvA}) $ will be an algorithm that takes as input a computational run $\runC$ and a randomness source $r_{\pmvA}$, and returns a set of computational labels $\LabC$, while attaining to some constraints.

The general idea behind our construction of $\stratMap(\stratS{\pmvA})$ is the following: the algorithm will first parse $\runC$ in order to create a symbolic run $\runS$, then it will use $\stratC{\pmvA}$ to produce a set of symbolic actions, which will lastly be translated into computational labels, concluding the process. In this way, $\stratC{\pmvA}= \stratMap(\stratS{\pmvA})$ is emulating its symbolic counterpart $\stratS{\pmvA}$.
These procedures closely resemble the definition of the coherence relation, so we will not present every detail.

\paragraph*{Parsing the computational run}
Here, we will take a consistent computational run $\runC$ and parse it, in order to construct a symbolic run $\runS$ coherent to it.
This will be a step-by-step construction, that takes a single label and finds a corresponding symbolic action. While doing that, we update the maps $\txMapC$ (between names and outputs), $\advMapC$ (between advertisements and transactions), and $\keyMapC$ (between participants and their public key): these will helps us to keep track of the symbolic terms that we created.

We begin with the initial prefix of $\runC$, which contains a transaction $\txT[0]$ followed by messages that transmit the computational participant's public keys. 
By looking at those messages, we create a set of symbolic participants, and the $\keyMapC$ that associates to each of them their public key. Then, by looking at the outputs of $\txT[0]$ we obtain a series of deposits, which, together with an empty destroyed funds counter $\confTaint{0}$, and the time $t=0$, will form the initial symbolic configuration $\confG[0]$, which will be the prefix of the symbolic run. The map $\txMapC$ is constructed to map each deposit of $\confG[0]$ to the corresponding output.

Then, we have different scenarios according to the next computational step $\labC$.
If $\labC$ is a message we ignore it, except for the following cases:
\begin{enumerate}
\item It is the encoding of a incomplete advertisement $\Theta$, in which case we perform the symbolic step $msg(\Theta)$.
\item It is the encoding of a compiler generated transaction $\txT[\confF]$, never sent before in the computational run (i.e. not belonging to $\ran{\advMapC}$). In this case $\labC$ corresponds to the advertisement of the valid term $\confF$ that has a subscript $h$ never used in the symbolic configuration (and we update $\ran{\advMapC}$).
\item It is the encoding of a transactions $\txT$ that takes at least an input in $\ran{\txMapC}$, is neither compiler-generated nor correspondent to a join divide or donate operation, and never sent before in the computational run. In this case $\labC$ corresponds to a destroy advertisement. 
\item It is a quadruple $(\txT, j,wit,i)$, where $wit$ is the signature with $\pmvA$'s key on the $j$-th output of $\txT$ (a transaction that is already present as a message in the run).
  Moreover $\advMapC^{-1} (\txT) = \confF$;
  and it is the first time $\labC$ is broadcast after a broadcast of $\txT$.
  In this case $\labC$ corresponds to a symbolic authorization step, either $auth-in(\pmvA, z, \confF)$ or $auth-act(\pmvA,\confF)$ depending on what kind of advertisement $\confF$ is, and on which input of $\txT$ is being signed.
\end{enumerate}
If $\labC$ is a transaction with at least one input in $\ran{\txMapC}$, then we can analyse its structure and find the corresponding action among the following: $init(\confF)$, $call(\confF)$, $send(\confF)$, $join(x,y)$, $divide(x,v,v')$, $donate(\pmvB,x)$, or $destroy(\confF)$.
Lastly if $\labC$ is a computational delay we directly translate it to a symbolic one.

Each step in this conversion process is uniquely determined, up to different choices for the names of participants, deposit, contracts, and $h$ subscripts, meaning that the above paragraph can be seen as proof for this proposition:
\begin{prop} 
  Given $\runC$ we can find $\runS$, $\txMapC$, $\keyMapC$, $\advMapC$ such that
  \[\coher{\runS}{\runC}{\txMapC}{\keyMapC}{\advMapC}.\]
  Moreover if  $\runSi$, $\txMapCi$, $\keyMapCi$, $\advMapCi$ are such that
  \[
  \coher{\runSi}{\runC}{\txMapCi}{\keyMapCi}{\advMapCi},
  \]
  then we can get $\runSi$ from $\runS$ by substituting each name $x$ with $\txMapCi^{-1}(\txMapC(x))$, each advertisement $\confF$ with $ \advMapCi^{-1}(\advMapC(\confF))$, and each participant name $\pmvA$ with $\keyMapCi^{-1}(\keyMapC(\pmvA))$.
\end{prop} 
\paragraph*{Randomness} Every strategy, symbolic or computational, takes as input a random seed $\nonce{\pmvA}$. In order provide a proper translation between strategies we need to make a few remarks on this randomness source. 
Every computational strategy needs to use its randomness source in order to produce the key pairs that will be broadcast at the start of the run. This action does not have any symbolic counterpart, since it is assumed that symbolic participants can give authorizations without needing to worry about the low level signature details.
It is also very important that the random bits used for key generation are never reused when choosing which action to perform in later steps, since this would cause a correlation between the keys and the later outputs of the run, potentially leaking information about the secret keys.
So, in order to avoid this problem, the symbolic strategy that we are trying to emulate must be prevented from seeing the part of the random sequence used in the keys generation process.
For this reason we will split the sequence $r_{\pmvA}$ in two, $\pi_1 (r_{\pmvA}) $ and $\pi_2 (r_{\pmvA})$ where the first part can be used in strategies and is given as input to the $\stratS{\pmvA}$, while the second is only used for the initial keys generation. 
\paragraph*{From symbolic actions to computational labels}
Once we have converted $\runC$ to $\runS$ we can compute $\LabS = \stratS{\pmvA}(\runS, \pi_1(r_{\pmvA}))$. Then, we can transform each element $\labS$ of $\LabS$ into a computational label $\labC$, by following the corresponding case inside the coherence definition.
However we must ensure that the constraints posed in Definition \ref{def:computational-participant-strategies} are respected. 
In the following paragraphs we will show how to do that. When calling the compiler $\compile{}$ we will always assume that the auxiliary functions are the ones constructed by the parsing step.
\paragraph*{Advertisements}
\begin{enumerate}
\item (Incomplete advertisement).
  If $\labS=msg(\Theta)$, then $\labC$ is simply a message encoding $\Theta$.
\item (Complete initial advertisement).
  Remember that by Definition \ref{def:symbolic-participant-strategies} $\stratS{\pmvA}$ can choose $\labS= adv(\confF)$ with $\confF$ complete only if the advertisement has $w=\star$.
  If $\confF = \confAdvInitN{\clauseCallX{a}{b}}{\vec{z}}{\star}{h}$ then we can compile it to $\txT[\confF]$ by choosing the compiler's inputs in the following way:
  $\inCmp[i] = \txMapC(z_i)$ for all $i$; $\tAbsCmp$ is the time in $\confG[\runS]$; $\tRelCmp[1]$ is the biggest delay specified in a $\afterRelName$ in $\contrD$; $\tRelCmp[j]= 0$ for all $j>1$; and $\nonceCmp$ is a list of numbers chosen so that the compiled transaction $\txT[\confF]$ is different from any previously broadcast transaction.
  The corresponding computational label in this case is $\labC= \pmvA: \txT[\confF]\rightarrow \ast $.
\item (Complete continuation advertisement).
  If $\confF= \confAdvN{\complD}{\vec{z}}{\star}{x}{j}{h}$ then, similarly to the case above, we can construct $\txT[\confF]$ by setting the appropriate compiler inputs. 
  We then have that $\labS= adv(\confF)$ corresponds to $\labC=  \pmvA: \txT[\confF]\rightarrow \ast $.
\item (Complete destroy advertisement).
  If $\labS = adv(\confF)$ with $\confF=\confAdvDesN{\vec{z}}{\star}{h}$ then $\labC=  \pmvA: \txT[\confF]\rightarrow \ast$, where $\txT[\confF]$ is a transaction with inputs given by $\txMapC(z_j)$ and an irredeemable output that has script $\txscript =\mathsf{false}$.
\end{enumerate}
\paragraph*{Authorizations} 
\begin{enumerate}
\item (Advertised actions). If $\labS = auth-act(\pmvA,\confF)$ or $auth-in(\pmvA, \confF, z)$ then $\labC= \pmvA: m\rightarrow \ast $ where $m$ is a quadruple $ (\txT[\confF], j, wit, i)$ encoding the corresponding witness.
  Notice that since $\labS$ can be sent only if $\confF$ is in the configuration, there must have already been an action $adv(\confF)$ in the symbolic run. The fact that $\runS$ is obtained by parsing $\runC$, which is consistent, means that if this step is reached $\txT[\confF]$ is already present in the run. 
\item (Deposits). If $\labS$ is an authorization for a $join$, $divide$, or $donate$ action, then we are not sure if the corresponding transaction is already present in the run (since they do not require a symbolic advertisement).
  So, if there $\txT$ is not in the blockchain then $\labC= \pmvA: \txT\rightarrow \ast$. If instead it is already present we act like in item 1, with $\labC = \pmvA: m\rightarrow \ast$, and $m$ encodes the required signature. 		
\end{enumerate}
\paragraph*{Actions}
\begin{enumerate}
\item (Advertised actions). If $\labS$ is an $init$, $call$, $send$ or $destroy$ action, consuming a term $\confF$, then the corresponding computational label is $\labC= \txT[\confF]= \advMapC(\confF)$.
  Remembering again that in $\confF$ we must have $w=\star$, we will prove that $\labC$ satisfies the constraints given to symbolic strategies.
  Indeed, for $\labS$ to be possible the advertised term $\confF$ and all the authorizations must be in the configuration $\confG[\runS]$. Since $\runS$ is constructed by parsing $\runC$ this means that $\txT[\confF]$ has been sent on the computational run (after its timelocks are exhausted), and that all the witnesses that correspond to a symbolic authorization are present. However, since $w= \star$ these are all the needed witnesses and $\txT$ may be choosen as an action by a computational strategy. 
\item (Deposit actions).If $\labS$ is a $join$, $divide$, or $donate$ action, then the  corresponding computational label is $\labC= \txT$. Again, the parsing ensures us that $\txT$ and all its witnesses are already sent in the computational run.
\end{enumerate}

\section{Security of the compiler}
\label{sec:app-computational-soundness}

In this appendix we prove the main result of this paper: the security of the \illum compiler.
We will see that if the participants choose their computational strategy by translating a symbolic strategy, then no matter what the computational adversary does, it is possible, with overwhelming probability, to simulate any of its computational action in the symbolic world, therefore maintaining coherence.

\begin{thm}[Security of the compiler]
  \label{thm:computational-soundness}
  Let $\stratSSet$ be a set of computational strategies for all honest participants, and $\stratCSet$ be a set of computational strategies consisting of $\stratC{\pmvA} = \stratMap(\stratS{\pmvA})$ for all $\pmvA \in \PartT$ and of an adversary strategy $\stratC{\Adv}$. 
  Given the security parameter $\eta$ and any $k \in \mathbb{N}$, we define
  \begin{align*}
    P(r)= \ &\forall \runC \text{ conforming to } (\stratCSet,r) \text{ with } \vert \runC \vert \leq \eta^k,
    \\ &
    \exists \runS, \txMapC, \keyMapC, \advMapC, \text{ such that }
    \\ & \qquad
    \coher{\runS}{\runC}{\txMapC}{\keyMapC}{\advMapC} \text{ holds}		
    \\ & \qquad
    \text{and } \runS \text{ conforms to } (\stratSSet,\pi_1(r)). 
  \end{align*}
  then, the set $\setenum{r \vert P(r)}$ has overwhelming probability
\end{thm}
\begin{proof}
  Consider any given $r$, and take $\runC$ that satisfies the conformance hypothesis and the length requirement. Assume also that there is no corresponding symbolic run $\runS$. We will show that this happens with negligible probability.
  
  Take $\runCi$, the longest prefix of $\runC$ such that there exist a corresponding run $\runSi$ and maps $\txMapC$, $\keyMapC$, and $\advMapC$ for which $\coher{\runSi}{\runCi}{\txMapC}{\keyMapC}{\advMapC}$ holds.
  This $\runCi$ is not empty, since the initial prefix of $\runC$ (consisting of $\txT[0]$ and the broadcast of public keys) can always be transformed into a corresponding initial symbolic run (and the conformance with strategies trivially holds for initial runs). 
  We will now proceed by cases on all the possible labels $\labC$ that can extend $\runCi$ to show that either it's possible extending $\runSi$ to a run coherent with $\runCi \labC$ (reaching a contradiction), or that the adversary has managed to produce a signature forgery (which only happens with negligible probability). 
  \begin{enumerate}
  \item $\labC = \pmvB \rightarrow \ast : m$. Looking at the coherence definition we can see that we need to consider four distinct cases for the message: $(i)$ $m$ encodes an incomplete advertisement $\confF$; $(ii)$ $m$ encodes a transaction $\txT[\confF]$, where $\confF$ is a valid advertisement term; $(iii)$ $m$ is quadruple encoding a witness for an input (with a symbolic counterpart) of some $\txT$ that was already broadcast in the run after its timelocks have expired; $(iv)$ $m$ is any other message.
    The first two cases are handled with a $msg(\confF)$ and an $adv(\confF)$ symbolic action respectively, and in the fourth case the symbolic run ignores the message $m$. 
    This leaves us with case $(iii)$ where the adversarial strategy has been able to produce a witness: this  means either that it has forged a signature (and this happens with negligible probability), or that some honest $\pmvA$ chose to provide it. However, if that's true, then the symbolic strategy of $\pmvA$ must have enabled the authorization at some point, since $\stratC{\pmvA} = \stratMap(\stratS{\pmvA})$.  This means that $\stratS{\Adv}$ can choose it as next action $\alpha$, meaning that $\runSi \xrightarrow{\alpha} \confG$ is still coherent with $\runC \labC$.
    
  \item  $\labC = \txT$. Again, we have multiple cases.
    \begin{enumerate}
    \item If $\txT$ does not have any inputs in $\ran{\txMapC}$ then coherence is achieved without adding any additional step to $\runSi$.
    \item If $\txT$ has some inputs in $\ran{\txMapC}$, and one of them is the image of an active contract, then, by Theorem \ref{thm:coher-compiler-generated} we have that $\txT$ must be a compiler-generated transaction $\txT[\confF]$.
      Since $\labC$ has been chosen by $\stratC{\Adv}$, we know it must follow the rules for strategies, so $\txT$ has been broadcast earlier (and not before its timelock are over), and all its witnesses have been broadcast too.
      This means that there has been a corresponding symbolic advertisement, (since $\runSi$ is coherent to $\runCi$), so $\confF$ has been included in the configuration. Theorem \ref{thm:coher-compiler-generated} also tells us that $\confF$ is in the form  $\confAdvN{\complD}{\vec{z}}{w}{x}{j}{h}$ and Theorem \ref{thm:coher-redeeming-contract-implies-valid-adv} proves that $\confF$ is valid in $\confG[\runSi]$. Notice also that, thanks to the conditions on strategies, we know that $\txT$'s witness have been broadcast in some previous step of the run, and, by coherence, this means that all the required symbolic authorizations are present in $\confG[\runS]$.
      The validity of $\confF$ and the presence of the deposit's authorization ensure that the continuation action corresponding to $\txT$ (either $call(\confF)$ or a $send(\confF)$) can be performed in $\runSi$, meaning that we can extend $\runSi$ and still achieving coherence. 
    \item If $\txT$ has some inputs in $\ran{\txMapC}$, but none of them is the image of any active contract, we have 3 possible situations:  $\txT= \txT[\confF]$ is compiler-generated starting from an initial advertisement $\confF$; $\txT$ is a transaction associated to a deposit action $join$, $divide$, or $donate$; or $\txT$ is some other transaction.
      In the first case we can carry out the same reasoning of step (b) to conclude that $\labS= init(\confF)$ is a continuation that achieves coherence.
      In the second case we can achieve coherence by letting $\labS$ be the corresponding symbolic deposit operation. In this case too all deposits and authorizations must be present in the run.
      In the third case we will choose $\labS$ to be a destroy operation. Again, we notice that $\txT$ must have been broadcast at some point, and since it is not a compiler generated transaction, nor it does correspond to a deposit action, the broadcast message falls into the case described by item 4 of the first inductive case of Definition \ref{def:coherence}; and the broadcast is thus mirrored in the symbolic run by the advertisement $adv(\confF)$, where $\confF= \confAdvDesN{\vec{z}}{w}{h}$. In this case too all of $\txT$ witnesses must have been advertised, meaning that all authorization for deposits $z_j$ are present in $\confG[\runSi]$. This means that in this case too we can achieve coherence by extending $\runSi$ with $\labS =destroy(\confF)$. 
    \end{enumerate}
    
  \item $\labC = \delta$. Here we can extend $\runSi$ with $delay(\delta)$. This trivially keep coherence between runs. Moreover the resulting symbolic run still conforms to the strategies: in the computational case all honest participant had to agree on the delay, which, by the definition of $\stratMap$ implies that it is also the case for the symbolic strategies. 
  \end{enumerate}  
  In each of these cases we manage, with overwhelming probability, to extend $\runSi$ to something that is coherent with $\runCi \labC$ (against the maximality of the prefix $\runCi$), and this concludes our proof.	
\end{proof}

\section{Compiling \hellum into \illum}
\label{sec:hllc}
\label{sec:app-high-level}

We describe in this section how to compile high-level contracts written in \hellum to the intermediate-level language \illum, as sketched in~\Cref{sec:hellum}.
% We show how a programmer can start coding a smart contract
% in a high-level language similar to those used on Ethereum,
% translate it into the intermediate language \langname,
% and finally execute it in a bare-bone UTXO blockchain after compilation.

We start by providing more details about \hellum, referring to \url{\github} for its concrete syntax and typing rules.
A contract has a set of variables that define its state,
and a set of functions with an imperative, loop-free body that can modify the contract state and transfer tokens.
Base types comprise \hllinline{bool}, \hllinline{int}, \hllinline{uint}, \hllinline{string}, and \hllinline{address}. Variables can also be \hllinline{mapping}s from base types to base types.  
A function can have \emph{modifiers} that must be satisfied before it can be called (as in Solidity), and \emph{continuations} that specify which functions can be called after it.
The general form of contracts is in~\Cref{fig:hll:general-form}.

\begin{figure}
\centering
\begin{tcolorbox}[left=2pt,right=2pt,top=0pt,bottom=0pt,colback=gray!3,boxrule=1pt]
\begin{lstlisting}[language=hellum,morekeywords={f,g,constructor,Foo}]
contract Foo {
  int x;            // integer
  uint u;         // unsigned integer	    
  address a;    // address (externally-owned)    
  mapping (address => uint) m;
  ...                  // other state variables
    
  constructor(...) { ... }	// Entry point    

  function f(int x, ..., address b, ...)
   input(e:T)  // Receive tokens
   auth(c)       // Authorizations
   after(t)      // Time constraint
   ...               // other modifiers...
  {
    int y;         // local variables
    ...	// function body
  } next(g1,...,gn) // Possible continuations

  ...   // other functions

  function g(...) view { // pure function
    ... // return expression
  }
}
\end{lstlisting}
\end{tcolorbox}
\negcaptionspace
\caption{General form of \hellum contracts.}
\label{fig:hll:general-form}
\end{figure}

\hellum functions have four possible modifiers:
% tagged by decorators \precode{@after}, \precode{@afterRel},
% \precode{@auth}, \precode{@pre}, and \precode{@receive}.
\begin{itemize}

\item \hllinline{after(t)}
  requires that the function is called only after (absolute) time \code{t}.
  Here, we abstract from the granularity of time:
  it could be \eg a block number (as in Solidity) or a timestamp;
  
\item \hllinline{auth(a)}
  requires that the function call is authorized by address \code{a} (through \code{a}'s private key);

\item \hllinline{input(e:T)}
  requires that \code{e} tokens of type \code{T}
  are sent to the contract alongside with the function call, by any address;

\item \hllinline{next(g1 ... gn)} 
    specifies the functions that can be called after the current function has been executed. When the \hllinline{next} modifier is omitted, any continuation (except the constructor) is possible.
\end{itemize}

Note that the expressions appearing within the \hllinline{after} modifier may only depend on the contract variables, while the expressions within \hllinline{input} and \hllinline{auth} may also depend on the function parameters.
A function can use multiple instances of the same modifiers, except for \hllinline{next}.

Function bodies are as in Solidity, but for the absence of loops and  contract calls: they comprise assignments (to variables and mappings), sequences of commands, conditionals, and \hllinline{require(e)} statements, which make the function fail when the expression \code{e} evaluates to false.
The command \hllinline{a.transfer(e:T)} transfers \code{e} units of token \code{T} to address \code{a}. 
Local variables, not contributing to the contract state, can be declared and used. 
%% For instance:
%% \begin{lstlisting}[language=hellum, morekeywords = {f}]
%%   if ( balance(T1) > 10 and balance(T2) > y ){
%%     x := x + y;
%%     pay(x:T1 -> A , y:T2 -> B)
%%   } else {
%%     x := 10;
%%   }
%% \end{lstlisting}
Expressions are standard, and follow the Solidity syntax. The special expression \hllinline{balance(T)} gives the number of units of token \hllinline{T} currently available in the contract.
Expressions can contain calls to pure functions (tagged as \hllinline{view} in the contract).
The \hellum compiler includes a semantic analyzer that performs type checking and other checks to ensure the well-formedness of contracts.   
\paragraph*{Compilation: normal form}

The first phase of the \hellum compiler is a series of code transformations to bring contracts in the normal form described in~\Cref{sec:hellum}.
This phase is split into several steps:
\begin{enumerate}

\item macro-expand calls to pure functions into the corresponding expressions;

\item rewrite each function as a chain of conditional statements, and merge the \hllinline{require} statements in a single \hllinline{require} at the top of the function;

\item rewrite the body of each conditional branch in static-single-assignment (SSA) form~\cite{Rosen88popl}, where each variable is written exactly once. Besides the contract variables, in this step we also add auxiliary variables keeping track of the varying contract token balances;

\item rewrite the body of each conditional branch so that the token transfers occur before all the assignments;

\item rewrite the body of each conditional branch so that all the assignments are folded into a single, simultaneous assignment of all the contract variables.

\end{enumerate}

Below, we illustrate the code transformations 2 to 5 through a series of examples, referring to the repository \url{\github} for the full details and for the transformation from  normal form contracts to \illum, as sketched in~\Cref{sec:hellum}.  

For step (2) of the normal form construction, we match patterns of the function body, and rewrite them to pull \hllinline{require} statements out of conditional blocks, and push \hllinline{transfer} and assignment commands within conditional blocks. These transformations modify the guards conditionals and other expressions preserving the semantics.
We illustrate the patterns through code snippets, showing how the left part is transformed into the right part.

Payments and assignments before a conditional are pushed within the conditional, adapting the guards to match the state updates. 
For instance:
\begin{tcolorbox}[left=2pt,right=2pt,top=0pt,bottom=0pt,colback=gray!3,boxrule=1pt]
\begin{minipage}{0.45\columnwidth}
\begin{lstlisting}[language=hellum]
a.transfer(y:T);
if (balance(T)<7) {
    // c1
}
else {
    // c2
}
\end{lstlisting}
\end{minipage}
\hfill
\vline
\hfill
\begin{minipage}{0.45\columnwidth}
\begin{lstlisting}[language=hellum]
if (balance(T)-y<7) {
    a.transfer(y:T);
    // c1
}
else {
    a.transfer(y:T);
    // c2
}
\end{lstlisting}
\end{minipage}
\end{tcolorbox}    

\begin{tcolorbox}[left=2pt,right=2pt,top=0pt,bottom=0pt,colback=gray!3,boxrule=1pt]
\begin{minipage}{0.45\columnwidth}
\begin{lstlisting}[language=hellum]
x=y-5;
if (x<3) {
    // c1
}
else {
    // c2
}
\end{lstlisting}
\end{minipage}
\hfill
\vline
\hfill
\begin{minipage}{0.45\columnwidth}
\begin{lstlisting}[language=hellum]
if (y-5<3) {
    x=y-5;
    // c1
}
else {
    x=y-5;
    // c2
}
\end{lstlisting}
\end{minipage}
\end{tcolorbox}    

The commands (of any kind) after a conditional are pushed within all the conditional branches. For instance: 
\begin{tcolorbox}[left=2pt,right=2pt,top=0pt,bottom=0pt,colback=gray!3,boxrule=1pt]
\begin{minipage}{0.45\columnwidth}
\begin{lstlisting}[language=hellum]
if (x<3) {
    // c1
}
else {
    // c2
}
x=x+1;
\end{lstlisting}
\end{minipage}
\hfill
\vline
\hfill
\begin{minipage}{0.45\columnwidth}
\begin{lstlisting}[language=hellum]
if (x<3) {
    // c1
    x=x+1;
}
else {
    // c2
    x=x+1;
}
\end{lstlisting}
\end{minipage}
\end{tcolorbox}

Nested conditional statements are flattened:
\begin{tcolorbox}[left=2pt,right=2pt,top=0pt,bottom=0pt,colback=gray!3,boxrule=1pt]
\begin{minipage}{0.45\columnwidth}
\begin{lstlisting}[language=hellum]
if (x<=9) {
    if (x>5) {
        // c1
    }
    else {
        // c2
    }
}
else {
    // c3
}
\end{lstlisting}
\end{minipage}
\hfill
\vline
\hfill
\begin{minipage}{0.45\columnwidth}
\begin{lstlisting}[language=hellum]
if (x<=9 && x>5) {
    // c1;
} 
else if (x<=9) {
    // c2
} 
else {
    // c3
}
\end{lstlisting}
\end{minipage}
\end{tcolorbox}

Each \hllinline{require} command is moved to the top of the function by swapping it with the previous command, and updating the guard accordingly. For instance: 
\begin{tcolorbox}[left=1pt,right=1pt,top=0pt,bottom=0pt,colback=gray!3,boxrule=1pt]
\begin{minipage}{0.495\columnwidth}
\begin{lstlisting}[language=hellum]
x=x+y;
require x<500;
\end{lstlisting}
\end{minipage}
\hfill
\vline
\hfill
\begin{minipage}{0.495\columnwidth}
\begin{lstlisting}[language=hellum]
require x+y<500;
x=x+y;
\end{lstlisting}
\end{minipage}
\end{tcolorbox}

\begin{tcolorbox}[left=2pt,right=2pt,top=0pt,bottom=0pt,colback=gray!3,boxrule=1pt]
\begin{minipage}{0.495\columnwidth}
\begin{lstlisting}[language=hellum]
a.transfer(x:T);
require balance(T)>5;
\end{lstlisting}
\end{minipage}
\hfill
\vline
\hfill
\begin{minipage}{0.495\columnwidth}
\begin{lstlisting}[language=hellum]
require balance(T)-x>5;
a.transfer(x+y:T);
\end{lstlisting}
\end{minipage}
\end{tcolorbox}

The most complex case is when a a \hllinline{require} occurs in each branch of a conditional statement. In this case, we pull all the \hllinline{require} out of the branches, and we combine the guards in the \hllinline{require} commands with the guards of the conditional, obtaining a single \hllinline{require}. For instance: 
\begin{tcolorbox}[left=2pt,right=2pt,top=0pt,bottom=0pt,colback=gray!3,boxrule=1pt]
\begin{minipage}{0.40\columnwidth}
\begin{lstlisting}[language=hellum]
if (x<2) {
    require x>0;
    // c1
} else if (x<4) {
    require x>2;
    // c2
} else {
    require x<8;
    // c3
}
\end{lstlisting}
\end{minipage}
\hfill
\vline
\hfill
\begin{minipage}{0.61\columnwidth}
\begin{lstlisting}[language=hellum]
require ((x<2 && x>0) || 
    (x>=2 && (x<4 && x>2))) || 
    ((x>=2 && x>=4) && x<8);
if (x<2) {
    // c1
}
else if (x<4) {
    // c2
}
else { 
    // c3
}
\end{lstlisting}
\end{minipage}
\end{tcolorbox}

%\iftoggle{arxiv}{%arxiv
%\input{app-transformations}
%}{}% arxiv

For step (3) of the normal form construction, we rewrite every conditional branch in SSA form. We do so by introducing an expression \hllinline{balance_pre(T)} (which returns the amount of token \hllinline{T} stored in the contract before the function invocation) as well as local variables when they are needed.
For illustration, we consider a single branch and assume that  \hllinline{a}, \hllinline{x},  \hllinline{y} are global variables of the contract, while \hllinline{z} is a function parameter.

\begin{tcolorbox}[left=2pt,right=2pt,top=0pt,bottom=0pt,colback=gray!3,boxrule=1pt]
\begin{lstlisting}[language=hellum, morekeywords={C,f}]
y = x + balance(T);
x = z;
a.transfer(x:T);
a.transfer(y:T);
y = balance(T) + z;	 
\end{lstlisting}
\end{tcolorbox}

The transformation introduces new local variables at each step, to keep track of the values of \hllinline{a,x,y,z} and of the contract balance.
For instance, the variables \hllinline{x_i} are introduced at every assignment of \hllinline{x}. 
A new variable \hllinline{bal_T_i} is introduced upon each \hllinline{transfer} to keep track of the balance of token \code{T}. 
Initially, we let \hllinline{bal_T_i} to be
\hllinline{balance_pre(T)} (plus eventual function inputs).
% Here we only have one token to keep track: in general, we would need to keep track of the balance of each token.
Our branch ends up rewritten as:
\begin{tcolorbox}[left=2pt,right=2pt,top=0pt,bottom=0pt,colback=gray!3,boxrule=1pt]
\begin{lstlisting}[language=hellum, morekeywords={C,f}]
x_0,y_0,a_0,z_0,bal_T_0 = 
x,   y,   a,   z,   balance_pre(T);
y_1 = x_0+bal_T_0;
x_1 = z_0;
a_0.transfer(x_1:T);
bal_T_1 = bal_T_0-x_1;
a_0.transfer(y_1:T);
bal_T_2 = bal_T_1-y_1;
y_2 = bal_T_2+z_0;
x,y,a,bal_T_fin = x_1,y_2,a_0,bal_T_2;
\end{lstlisting}
\end{tcolorbox}

For step (4), we now move the two \lstinline[language=hellum]|transfer()| statements to the top by exchanging them with the assignments. 
To do this, we replace the variables appearing in \lstinline[language=hellum]|transfer()| with the expression on the right hand side of the assignment. In our example, we get:
\begin{tcolorbox}[left=2pt,right=2pt,top=0pt,bottom=0pt,colback=gray!3,boxrule=1pt]
\begin{lstlisting}[language=hellum, morekeywords={C,f}]
a.transfer(z:T);
a.transfer(x+balance_pre(T):T);
x_0,y_0,a_0,z_0,bal_T_0 = 
x,   y,   a,   z,   balance_pre(T);
y_1 = x_0+bal_T_0;
x_1 = z_0;
bal_T_1 = bal_T_0-x_1;
bal_T_2 = bal_T_1-y_1;
y_2 = bal_T_2+z_0;
x,y,a,bal_T_fin = x_1,y_2,a_0,bal_T_2;
\end{lstlisting}
\end{tcolorbox}

For step (5), we collapse all the assignments into a single simultaneous one, that assigns the new values to the contract variables. 
In our example:
\begin{tcolorbox}[left=2pt,right=2pt,top=0pt,bottom=0pt,colback=gray!3,boxrule=1pt]
\begin{lstlisting}[language=hellum, morekeywords={C,f}]
a.transfer(z:T);
a.transfer(x+balance_pre(T):T);
x,y,a,bal_T_fin = z,((balance_pre(T)-z)-(x+balance_pre(T)))+z,a,(balance_pre(T)-z)-(x+balance_pre(T));
\end{lstlisting}
\end{tcolorbox}

From this last normal form, we can generate the \illum function clauses \illinline{f_run} and \illinline{f_next} as discussed in~\Cref{sec:hellum}.

\paragraph{On loops in \hellum}
The \hellum language does not feature loops. On the one hand, this makes the compilation to \illum easier, but on the other hand this reduces the expressivity of \hellum. 
Allowing loops in \hellum could be done in three ways.
The simplest option is to extend the language with specific iterators on key-value maps (\eg, map, filter, fold). These operators could then be compiled in corresponding operators in a suitably extended \illum.
More specifically, this would only require extending the \illum and UTXO script expressions with suitable operators.
Since such loops would be bounded, this option does not strictly require a gas mechanism to prevent divergent behaviours.
A second option would be to allow arbitrary (unbounded) loops in \hellum
and suitably extend the \illum expressions
with operators that can simulate such arbitrary \hellum loops (\eg, a fixed point operator). Note that such an extension would make the evaluation of \illum expressions potentially divergent, hence it would require a gas mechanism or some other means to bound the computation.
For example, the Cardano scripting language (Plutus Core) 
is an untyped lambda calculus, thus allowing for unbounded computation,
% allows general computation, 
but the Cardano platform limits the execution of scripts to a given amount of computation steps.
A last option would be to allow arbitrary \hellum loops but compile them to a \emph{chain} of recursive \illum clauses. Intuitively, calling such a recursive clause would only perform a part of the loop (say, the first iteration), and then call itself with the updated state. The recursion then stops whenever the loop is over, and proceeds to call another clause. While this mechanism effectively makes \illum Turing-complete, it requires the users to perform a potentially large number of calls, hence to append a large number of transactions on the blockchain, paying the fees for all of them. Further, this could lead to Denial of Service attacks. A malicious participant could call a \hellum function which performs a long loop, pay the fees for the first few iterations and then stop interacting. In this way, the other participants are prevented to call other methods until they first complete the long loop by paying all the fees themselves.
Worse, there is nothing stopping a malicious participant from calling the method again after its completion, blocking honest users from accessing the contract and forcing them to pay the fees once again. 
Therefore, this last option for handling loops would require more complex protocols to counter attacks like the ones described above.

\end{document}